\pdfoutput=1
\documentclass[aps,onecolumn,10pt]{revtex4-2}

\usepackage[T1]{fontenc}
\fontencoding{T1}  
\usepackage[utf8]{inputenc}

\usepackage{enumitem}
\usepackage{tikz}
\usepackage{amsfonts}
\usepackage{amsmath,amssymb,amsthm}
\usepackage{graphicx}
\usepackage[breakable]{tcolorbox}
\usepackage{bbm}
\usepackage{url}
\usepackage{mathtools}
\usepackage{braket}
\usepackage{upgreek}
\usepackage{dsfont}
\usepackage{collectbox}
\usepackage{braket}
\usepackage{quantikz}

\usepackage[plain]{fancyref}

\usepackage{hyperref}

\newtheorem{thm}{Theorem}
\newtheorem*{thm*}{Theorem}
\makeatletter
\newcommand{\setthmtag}[1]{
  \let\oldthethm\thethm
  \newcommand{\thethm}{#1}
  \g@addto@macro\endthm{
    \addtocounter{thm}{-1}
    \global\let\thethm\oldthethm}
  }
\makeatother

\newtheorem{prop}[thm]{Proposition}
\newtheorem*{prop*}{Proposition}
\newtheorem{lemma}[thm]{Lemma}
\newtheorem*{lemma*}{Lemma}
\newtheorem{cor}[thm]{Corollary}
\newtheorem*{cor*}{Corollary}

\newtheorem*{cj*}{Conjecture}
\newtheorem{defn}[thm]{Definition}
\newtheorem*{Def*}{Definition}

\theoremstyle{definition}

\newtheorem*{rem*}{Remark}

\def\beq{\begin{equation}}
\def\eeq{\end{equation}}
\def\bq{\begin{quote}}
\def\eq{\end{quote}}
\def\ben{\begin{enumerate}}
\def\een{\end{enumerate}}
\def\bit{\begin{itemize}}
\def\eit{\end{itemize}}

\def\r|{\right|}

\newcommand{\ketbra}[2]{|#1\rangle\langle #2|}

\newcommand{\norm}[1]{\left\|#1\right\|}

\newcommand\be{\begin{equation}}
\newcommand\ee{\end{equation}}

\begin{document}
\title{Order $p$ quantum Wasserstein distances from couplings}

\author{\begingroup
\hypersetup{urlcolor=navyblue}
\href{https://orcid.org/0009-0006-5387-5291}{Emily Beatty
\endgroup}
}
\email[Emily Beatty]{emily.beatty@ens-lyon.fr}
 \affiliation{Univ Lyon, ENS Lyon, UCBL, CNRS, Inria, LIP, F-69342, Lyon Cedex 07, France}

\author{\begingroup
\hypersetup{urlcolor=blue}
\href{https://orcid.org/0000-0001-9699-5994}{Daniel Stilck Fran\c{c}a}
\endgroup}
\email[Daniel Stilck Fran\c ca ]{daniel.stilck\_franca@ens-lyon.fr}
\affiliation{Univ Lyon, ENS Lyon, UCBL, CNRS, Inria, LIP, F-69342, Lyon Cedex 07, France}

\begin{abstract}
Optimal transport provides a powerful mathematical framework with applications spanning numerous fields. A cornerstone within this domain is the $p$-Wasserstein distance, which serves to quantify the cost of transporting one probability measure to another. While recent attempts have sought to extend this measure to the realm of quantum states, existing definitions often present certain limitations, such as not being faithful. In this work, we present a new definition of quantum Wasserstein distances. This definition, leveraging the coupling method and a metric applicable to pure states, draws inspiration from a property characterising the classical Wasserstein distance — its determination based on its value on point masses. Subject to certain continuity properties, our definition exhibits numerous attributes expected of an optimal quantum rendition of the Wasserstein distance. Notably, our approach seamlessly integrates metrics familiar to quantum information theory, such as the trace distance. Moreover, it provides an organic extension for metrics, like Nielsen's complexity metric, allowing their application to mixed states with a natural operational interpretation. 
We analyse this metric's attributes in the context of random quantum states, unveil phase transitions concerning the complexity of subsystems of random states and use it to derive circuit lower bounds. In addition, we show how we can use our definition to study hypercontractive inequalities for quantum channels that do not admit a faithful fixed point, allowing us to derive concentration inequalities. Finally, we discuss how our distance is well-suited to define quantum Wasserstein generative adversarial networks.
\end{abstract}

\maketitle
\tableofcontents

\section{Introduction}

Optimal transport has established itself as a powerful tool in various areas of science and pure mathematics, such as machine learning~\cite{frogner2015learning}, information theory~\cite{raginsky2015concentration}, partial differential equations~\cite{villani2008optimal} and economics~\cite{galichon2018optimal}. In light of this, it should come as no surprise that the last years have seen a surge of interest in the quantum generalisation of optimal transport~\cite{de2020quantum, Carlen_2014, chakrabarti2019quantum, Golse2015OnTM, DePalma2021, Friedland_2022, de2022limitations, gao2021geometric, datta2020relating, bunth2024metric}.

One of the central concepts of classical optimal transport is the set of $p$-Wasserstein distances $\mathcal{W}_p$, which is a family of distances on the set of probability measures on a metric space. Roughly speaking, if we imagine a probability measure as describing a distribution of mass, these distances measure the cost of transporting one measure onto another in terms of a cost function on the underlying metric space. They can recover many widely studied metrics in probability spaces in a unified way. For instance, the total variation distance corresponds to the Wasserstein distance obtained when we equip the space with the discrete metric.
In the classical setting~\cite{villani2008optimal}, they admit many different equivalent characterisations, which we can broadly categorise into a \emph{geometrical formulation}, a \emph{dual formulation} and a \emph{coupling formulation}.

There now exist many different proposals of quantum Wasserstein distances that take one of the approaches mentioned above~\cite{de2020quantum, De_Palma_2023, DePalma2021, Carlen_2014, Carlen_2019, Rouze2019, cole2022quantum, Friedland_2022, Golse2015OnTM}. Important examples of definitions that follow the geometric route are the definitions by Carlen and Maas~\cite{Carlen_2014, Carlen2017}, which have been extensively studied in a series of works~\cite{Rouze2019, Carlen_2019, datta2020relating}. Roughly speaking, the geometric framework gives a natural generalisation of $\mathcal{W}_2$ in terms of a Riemannian metric on the set of full rank states such that the relative entropy with respect to a reference state arises as a gradient flow.
However, it has the undesirable feature of depending on an underlying semigroup converging to a reference state. 

A good example of a distance that is based on a dual approach is the definition of De Palma et al~\cite{de2020quantum}. Whereas the geometric approach is well-suited to generalise the case $p=2$, this approach works well for $p=1$, and the authors were able to define a natural quantum generalisation of $\mathcal{W}_1$ with respect to the Hamming distance on the hypercube. The dual approach works by first defining a good generalisation of a so-called Lipschitz constant of observables, which in some sense measures how much the expectation value of the observable can change between two points when normalised by the underlying metric.
From this, defining the distance induced by such observables is straightforward through duality. This generalisation has also found numerous applications in quantum information and computation~\cite{depalma2021quantum,de2022limitations,rouze2021learning, Kiani2022, De_Palma_2023, depalma2023shadows}, and has been extended to a distance between channels \cite{Duvenhage1_2023}.

Finally, some works have followed the coupling approach~\cite{cole2022quantum, Friedland_2022, DePalma2021, chakrabarti2019quantum,Duvenhage2_2021, Duvenhage2_2022}. Recall that given two random variables, a coupling is a joint distribution whose marginals are the two random variables. In the classical setting, the Wasserstein distance can be defined as the infimum of the expected cost of transporting one random variable onto the other, where the infimum is taken over all couplings of the two random variables. The main advantage of the coupling approach is that it can yield a definition of a quantum $p$-Wasserstein distance for all values of $p$. Existing approaches that attempt to quantise the coupling approach have significant downsides and usually do not satisfy one or more of the key properties expected from the Wasserstein distance. For example, the definition of \cite{cole2022quantum} is not a semidistance, and the definition of \cite{DePalma2021} is not faithful.  

In this work, we propose a new definition of the quantum Wasserstein distance based on the coupling approach. Our novel definition departs from the observation that, classically, the Wasserstein distance is completely determined by its value on point masses. As pure states are quantum analogues of point masses, we propose a coupling definition of the quantum $p$-Wasserstein distance induced by a metric on pure states and then consider an optimisation over all separable couplings of the underlying states.

We show that, as long as the metric on pure states satisfies some simple continuity properties, the induced quantum $p$-Wasserstein has many desirable properties expected from a good quantisation of the Wasserstein distance. In particular, we show that our definition recovers widely studied metrics on the set of quantum states, such as the trace distance, and is closely related to De Palma's quantum Wasserstein distance of order $1$. Furthermore, we show that it offers a natural way of extending distances from pure to mixed states in a way that preserves all the symmetries of the original metric and recovers its value on pure states. We also discuss some interesting ways in which our quantised distance behaves differently from its classical counterpart.

The main aim of this work is to cover a gap in the existing literature by expanding the theory on $p$-Wasserstein distances beyond the established orders $p=1$ and $p=2$, and to establish a general definition that is not restricted by the structure or natural geometry of the space to which it is applied. The approach given defines a quantum $p$-Wasserstein distance in very broad generality, allowing for much greater flexibility in application. This flexibility is enhanced by its ability to adapt to any underlying geometry on the system's projective Hilbert space subject to the aforementioned simple continuity properties. It follows that these distances can be adapted to represent the transport cost with respect to different qualitative properties of quantum states, as long as those properties can be captured by a distance on the projective Hilbert space.

As an application of our new definition, we give a generalisation of Nielsen's complexity geometry~\cite{Nielsen2006} to mixed states that has a natural operational interpretation of measuring the complexity of mapping one mixed state to another. We then study the behaviour of this metric for various ensembles of random quantum states and identify phase transitions in the expected Wasserstein distance between them. More concretely, we show that if we look at reduced density matrices of regions whose size is at most one-ninth of the total size, then the complexity distance is exponentially small. In contrast, for large enough subsystems, the distance is maximal. This result can be interpreted as formalising the intuitive notion that small enough subsystems of random states are essentially maximally mixed and trivial from a complexity point of view, whereas large enough subsystems are highly entangled and complex. This generalisation also allows us to obtain bounds which reflect the exponential nature of complexity in the number $n$ of qubits, in contrast with existing methods from quantum Wasserstein distances which can only give bounds which are linear in $n$.

We also study De Palma et al's $W_1^H$ distance under random quantum states, showing that they are essentially maximally far apart. Previous works~\cite{rouze2021learning} hinted at the possibility that the $W_1^H$ captures how well states can be distinguished by local observables. However, as random states cannot be distinguished by local observables, and as our results show that random states can be distinguished by the $W_1^H$ distance, we see that this original intuition is false. This mirrors a result found concurrently in \cite{depalma2023shadows}. This shows that the $W_1^H$ distance, and in turn the set of Lipschitz observables defined from its dual, have a significantly richer structure than previously expected.

Finally, we discuss three further applications of our distance: an operational interpretation of our newly introduced distance in terms of classical-quantum sources, analysis of the noise in a channel in terms of its hypercontractivity under the quantum $p$-Wasserstein distance, and theoretical improvements to the qualitative accuracy of quantum Wasserstein generative adversarial networks (qWGANs) using $p$-Wasserstein distances. Our hypercontractive inequalities have the advantage of not requiring the underlying quantum channels to have a faithful fixed point, an issue with other approaches~\cite{bardet2018hypercontractivity}, while still allowing for some the typical applications of hypercontractivity like concentration inequalities.

Each of these is an application which is only possible thanks to the broad flexibility of this new definition. We are able to discuss arbitrary $p^{\text{th}}$ moments of the distance between the output of classical-quantum sources thanks to the definition of $p$-Wasserstein distances for arbitrary $p$, something which has not yet been possible. 

Theoretical improvements to qWGANs come from the flexibility in terms of the underlying distance on the projective Hilbert space. And discussing hypercontractive properties of quantum channels requires comparing $p$-Wasserstein distances of two different orders $p_1, p_2$, which has been made possible by this new definition. Although existing definitions with enough basic properties to be usable cover both $p=1$ and $p=2$, these are applicable in wildly different contexts, and so direct comparison is not feasible without the new definition presented in this work.

\section{Basic concepts and notation}

\subsection{Classical optimal transport}

In the classical setting, the domain of optimal transport introduced in~\cite{monge1781memoire} is well-developed and finds applications in a wide variety of areas. 
The focus is to understand the cost of transporting one measure onto another and find the minimal cost of such transportation. More formally, for measurable spaces $\mathcal{X}, \mathcal{Y}$ with probability measures $\mu, \nu$ respectively, we assign a cost $c: \mathcal{X} \times \mathcal{Y} \to \mathbb{R}_{\geq 0}$ which represents the cost $c(x,y)$ of moving one unit of measure from $x$ to $y$. From here we optimise the total cost of transporting $\mu$ onto $\nu$ for each coupling~\cite{Kantorovich1942} $\gamma$ of $\mu$ and $\nu$.

\begin{defn}
Let $\mu, \nu$ be probability measures on measurable spaces $\mathcal{X},\mathcal{Y}$ respectively. A \textbf{coupling} of $\mu$ and $\nu$ is a measure $\gamma$ on $\mathcal{X} \times \mathcal{Y}$ such that for all measurable sets $\mathcal{A}\subseteq \mathcal{X}$, we have $\gamma(\mathcal{A}\otimes \mathcal{Y}) = \mu(\mathcal{A})$, and all measurable sets $\mathcal{B} \subseteq \mathcal{Y}$ we have $\gamma(\mathcal{X}\otimes \mathcal{B}) = \nu(\mathcal{B})$. Denote by $\mathcal{C}(\mu,\nu)$ the set of couplings of $\mu$ with $\nu$.
\end{defn}

In other words, a coupling is a measure $\gamma$ on $\mathcal{X} \otimes \mathcal{Y}$ with marginals $\mu$ and $\nu$. For a cost function $c$ on $\mathcal{X}\times \mathcal{Y}$, we can then define the \textbf{transport cost} of a coupling $\gamma$ as
\begin{equation} \label{eq:class_trans_cost}
    \mathcal{T}(\gamma) = \int_{\mathcal{X}\times \mathcal{Y}} c(x,y) \mathrm{d}\gamma(x,y).
\end{equation}
The optimal transport cost of transporting $\mu$ and $\nu$ is then given by minimising \eqref{eq:class_trans_cost} over all couplings $\gamma$. Throughout we will only consider the case where $\mathcal{X} = \mathcal{Y}$, $\mathcal{X}$ is equipped with a metric $d$, and the cost function $c$ is a power of the metric $d$. This leads to the definition of a classical Wasserstein distance~\cite[p.~207]{villani2003topics}.

\begin{defn}
    Let $(\mathcal{X},d)$ be a measurable metric space with probability measures $\mu$, $\nu$. Given $p\geq1$, the $p^{\text{th}}$\textbf{-order Wasserstein distance} between $\mu$ and $\nu$ is given by
    \begin{equation} \label{eq:class_wass_dist}
        \mathcal{W}_p(\mu,\nu) = \left(\inf_{\gamma \in \mathcal{C}(\mu,\nu)} \int_{\mathcal{X}\times\mathcal{X}} d(x,y)^p \mathrm{d}\gamma(x,y) \right)^{1/p}.
    \end{equation}
\end{defn}

This notion also extends to $p = \infty$ in the following way:
\begin{defn}
    Let $(\mathcal{X},d)$ be a measurable metric space with probability measures $\mu$, $\nu$. The \textbf{infinite-order Wasserstein distance} between $\mu$ and $\nu$ is given by
    \begin{equation} \label{eq:class_wass_dist_infty}
        \mathcal{W}_\infty (\mu,\nu) = \inf_{\gamma \in \mathcal{C}(\mu,\nu)} \sup_{(x,y) \in \text{{\em  supp}}(\gamma)} d(x,y).
    \end{equation}
\end{defn}

It is desirable to generalise these distances to quantities on quantum states to mirror many of their classical applications. We're particularly interested in generalisations of the case where the underlying metric $d$ is the discrete metric, for which $\mathcal{W}_p$ is the total variation distance, and the case where the underlying metric is the Hamming distance on the hypercube, for which $\mathcal{W}_1$ is Ornstein's $\bar{d}$-distance~\cite{ornstein1973}. These are arguably the most widely used distances in this context so far and a good generalisation should also recover them.

\subsection{Quantum information framework}

We represent a quantum system by a separable Hilbert space $\mathcal{H}$ over $\mathbb{C}$. Pure states on this system are modelled as elements $\ket{\psi}$ of the projective space $\mathbb{P}\mathcal{H}$, where $\mathbb{P}\mathcal{H} = (\mathcal{H} \setminus \{0\})/\mathbb{C}$. We notate the Hermitian conjugate of such a pure state by $\bra{\psi}$, and then in a slight abuse of notation we write $\ketbra{\psi}{\psi}$ for the projection onto $\ket{\psi}$, taking the convention $\text{Tr}[\ketbra{\psi}{\psi}] = 1$. We also use the term `pure state' to refer to this projection in the set of operators on $\mathcal{H}$, and the distinction will be clear from the context. For most of this manuscript, and in particular the examples and applications, we will discuss finite-dimensional Hilbert spaces.

The set of mixed states on $\mathcal{H}$ is defined as the set of positive semidefinite linear operators of trace $1$, otherwise known as density operators, and we will denote this set by $\mathcal{D}(\mathcal{H})$. These operators are all compact, and so can be decomposed spectrally in the form
\begin{equation} \label{eq:mixed_state_def}
    \rho = \sum_j q_j\ketbra{\psi_j}{\psi_j}
\end{equation}
where $\sum_j q_j = 1$ and $q_j > 0$.
The term `state' could refer to either an element $\ket{\psi}$ of $\mathbb{P}\mathcal{H}$ or a trace one positive semidefinite linear operator $\rho$, and the meaning will again be clear from the context.

The combination of two systems $\mathcal{H}_1$ and $\mathcal{H}_2$ is represented by their tensor product $\mathcal{H}_1 \otimes \mathcal{H}_2$. Their projective space is then $\mathbb{P}(\mathcal{H}_1 \otimes \mathcal{H}_2)$. Similarly, the set of density operators on this bipartite system is $\mathcal{D}(\mathcal{H}_1 \otimes \mathcal{H}_2)$.

The quantum equivalent of taking the marginal is the partial trace operation. Using this, we can define a quantum coupling of states $\rho$, $\sigma$.

\begin{defn} \label{defn:quantum_coupling}
    Let $\rho$, $\sigma$ be quantum states on Hilbert spaces $\mathcal{H}_1, \mathcal{H}_2$ respectively. A \textbf{quantum coupling} of $\rho$ with $\sigma$ is a quantum state $\tau$ on $\mathcal{H}_1 \otimes \mathcal{H}_2$ such that
    \begin{equation} \label{eq:quantum_coupling_def}
        \text{{\em Tr}}_2 \tau = \rho \qquad \text{and} \qquad \text{{\em Tr}}_1 \tau = \sigma.
    \end{equation}
\end{defn}
As we focus primarily on distances in this manuscript, we will almost always refer to the case where $\mathcal{H}_1 = \mathcal{H}_2$.
Throughout this work, we will also reference a few key objects from quantum information theory and quantum computation.

\begin{defn}
    For states $\rho$ and $\sigma$ on Hilbert space $\mathcal{H}$, their \textbf{trace distance} is half the trace norm of their difference, defined by
    \begin{equation}
        d_1(\rho,\sigma) = \frac{1}{2} \norm{\rho-\sigma}_1 = \frac{1}{2} \text{{\em Tr}}\left[\sqrt{(\rho-\sigma)^{\dagger}(\rho-\sigma)}\right].
    \end{equation}
We also have the equivalence
    \begin{equation}
        d_1(\rho,\sigma) = \sup_{P^\dagger = P, 0\leq P \leq I} \text{{\em Tr}}[P(\rho-\sigma)].
    \end{equation}
\end{defn}

\begin{defn}
    Let $\mathcal{U}(n)$ be the group of $n \times n$ unitary matrices. The \textbf{Haar measure} \cite{Weil1940} on $\mathcal{U}(n)$, denoted $\mu_{\text{{\em Haar}}}$, is the unique measure on the group $\mathcal{U}(n)$ which is invariant under left-multiplication and for which $\mu_{\text{{\em Haar}}}(\mathcal{U}(n)) = 1$.
\end{defn}

The Haar measure on the unitary group induces a measure on $\mathbb{P}\mathcal{H}$ for $\mathcal{H}$ of dimension $n$. This measure has the form of the pushforward $M_{\#} \mu_{\text{Haar}}$ for the map $M(U) = U\ket{0}$ for some fixed $\ket{0}$. The choice of $\ket{0}$ does not change the final measure. We refer to this as the Haar measure on $\mathbb{P}\mathcal{H}$.

\begin{defn}
    The $2 \times 2$ \textbf{Pauli operators} $X, Y, Z$ are given by the matrices
    \begin{equation} \label{eq:2by2Paulis}
        X = \begin{pmatrix}
            0 & 1 \\ 1&0 \end{pmatrix} \qquad
        Y = \begin{pmatrix}
            0 & -i \\ i&0 \end{pmatrix} \qquad
        Z = \begin{pmatrix}
            1 & 0 \\ 0 & -1 \end{pmatrix}
    \end{equation}
seen as unitary transformations with respect to the computational basis. On the $n$-fold tensor product of $\mathbb{C}^2$, a \textbf{Pauli string} is a product $\sigma_1 \otimes \dots \sigma_n$ where each $\sigma_i \in \{\mathbb{I},X,Y,Z\}$. The \textbf{weight} of such a string is the number of indices for which $\sigma_i \neq \mathbb{I}$, so the weight of $\mathbb{I} \otimes X \otimes \mathbb{I}$ is one and the weight of $X \otimes Y \otimes Z$ is three.
\end{defn}

\section{Motivations and definitions}

To generalise the Wasserstein distances to the quantum setting in a way that is useful, it is desirable to replicate as many basic properties of the classical definition as possible. To allow this, we note a key property of the classical Wasserstein distances. For all orders $p \geq 1$ and all $x, y \in \mathcal{X}$, the $\mathcal{W}_p$ distance between point masses agrees with the metric distance. That is, for Dirac measures $\delta_x$ and $\delta_y$, we have

\begin{equation}
\mathcal{W}_p(\delta_x,\delta_y) = d(x,y).    
\end{equation}

It follows, therefore, that the classical Wasserstein distances are determined by their values on point masses. We also see that we can recover the underlying metric $d$ from the Wasserstein distance. 

This property of determination by point masses is the key motivation of this work. For any quantum Wasserstein distance which is defined for general order $p$, the distances between the quantum versions of point masses (pure states) should be independent of $p$, and should agree with the underlying metric on the space if one exists. This definition of a quantity on mixed states stemming from a quantity on pure states is akin to the convex roof construction in the definition of entanglement measures on mixed states, as in \cite{Toth_2015, Wei2003, Coffman_2000, Wootters1998}.

To replicate this property, consider a Hilbert space $\mathcal{H}$ with projective space $\mathbb{P}\mathcal{H}$, equipped with metric $d$. In this setting, for all orders $p \geq 1$ and states $\ket{\psi}$ and $\ket{\varphi}$ we would expect a quantum Wasserstein distance $W_p^d$ to satisfy

\begin{equation}
    W_p^d(\ketbra{\psi}{\psi},\ketbra{\varphi}{\varphi}) = d(\ket{\psi},\ket{\varphi})
\end{equation}
taking the convention that $\ketbra{\psi}{\psi}$ and $\ketbra{\varphi}{\varphi}$ both have unit trace.
Given this motivation for a property the distance should satisfy, it remains to generalise the definition of a transport cost in \eqref{eq:class_trans_cost} to the quantum setting. One possible generalisation is given in the definition below.

\begin{defn}  Let $\mathcal{H}_1$, $\mathcal{H}_2$ be separable Hilbert spaces with projective spaces $\mathbb{P}\mathcal{H}_1$, $\mathbb{P}\mathcal{H}_2$ respectively, and let $\rho \in \mathcal{D}(\mathcal{H}_1)$, $\sigma \in \mathcal{D}(\mathcal{H}_2)$. We define a \textbf{quantum transport plan} between $\rho$ and $\sigma$ as any countable set $Q$ of triples
\begin{equation} \label{eq:defn_transport_plan}
    Q = \{ (q_j, \ket{\psi_j},\ket{\varphi_j}) \}_{j \in J}
\end{equation}
 such that
\begin{equation}
    \sum_{j \in J} q_j \ketbra{\psi_j}{\psi_j} = \rho \qquad \text{\em and} \qquad \sum_{j \in J} q_j \ketbra{\varphi_j}{\varphi_j} = \sigma. \label{eq:quantum-transport-plan}
\end{equation}
where $q_j > 0$, and taking convention $\text{{\em Tr}}\ketbra{\psi_j}{\psi_j} = \text{{\em Tr}}\ketbra{\varphi_j}{\varphi_j}= 1$. This is equivalent to an expression of a quantum coupling $\tau$ of $\rho$ and $\sigma$ as a convex combination of pure bipartite states
 \begin{equation} \label{eq:example_transport_plan}
     \tau = \sum_{j \in J} q_j \ketbra{\psi_j}{\psi_j} \otimes \ketbra{\varphi_j}{\varphi_j}
 \end{equation}
    where the indexing set $J$ is countable.
\end{defn}
We denote the set of all quantum transport plans between $\rho$ and $\sigma$ by $\mathcal{Q}(\rho,\sigma)$. Note that $\mathcal{Q}(\rho,\sigma) \neq \emptyset$; indeed for spectral decompositions $\rho = \sum_{j \in J} \lambda_j \ketbra{\psi_j}{\psi_j}$ and $\sigma = \sum_{k \in K} \mu_k \ketbra{\varphi_k}{\varphi_k}$ we have 
\begin{equation}
Q = \{ (\lambda_j \mu_k, \ket{\psi_j}, \ket{\varphi_k} ) \}_{j \in J, k \in K} \in \mathcal{Q}(\rho,\sigma).    
\end{equation}
Transport plans can be defined in any way satisfying the criteria above. However when discussing the number of elements in a finite transport plan, we say that two elements $(q,\ket{\psi},\ket{\varphi})$, $(q',\ket{\psi},\ket{\varphi})$ transporting the same states never appear in the same transport plan: any plan written as such is implied to contain the element $(q+q',\ket{\psi},\ket{\varphi})$ in their place.

This equivalence between transport plans and ways of writing couplings reflects the classical case, although we note that one quantum coupling could give rise to multiple quantum transport plans.  For example, the separable quantum coupling $\tau = \ketbra{0}{0} \otimes \frac{\mathbb{I}}{2}$ between $\ketbra{0}{0}$ and $\frac{\mathbb{I}}{2}$ on a two-qubit space gives rise to transport plans $\{(1/2,\ket{0},\ket{0}),(1/2,\ket{0},\ket{1})\}$ and $\{(1/2,\ket{0},\ket{+}),(1/2,\ket{0},\ket{-})\}$.

In building quantum transport plans, we can also refer to \textbf{partial quantum transport plans}. This is a quantum transport plan where $\rho$ and $\sigma$ are instead positive semidefinite operators of equal trace at most $1$, and the partial plan transports the `partial state' $\rho$ onto the `partial state' $\sigma$. 

We can then use this notion of a quantum transport plan to replicate the classical definitions of transport cost and Wasserstein distance in the quantum setting. The concept of a quantum transport plan is defined for all separable Hilbert spaces $\mathcal{H}_1, \mathcal{H}_2$, but from this point forwards we require $\mathcal{H}_1 = \mathcal{H}_2$ in order to have a metric between elements of transport plans. From here onwards we refer to these both as $\mathcal{H}$.

\begin{defn}
Let $\mathcal{H}$ be a separable Hilbert space and let $d$ be a distance on $\mathbb{P}\mathcal{H}$. Let $p \geq 1$. For any transport plan $Q = \{ (q_j, \ket{\psi_j},\ket{\varphi_j})\}_{j \in J}$ we define its $p^{\textbf{th}}$\textbf{-order quantum transport cost} as

\begin{equation} \label{eq:quantum_trans_cost}
T_p^d(Q) = \sum_{j \in J} q_j d \left( \ket{\psi_j},\ket{\varphi_j} \right)^p.
\end{equation}

The $p^{\textbf{th}}$\textbf{-order quantum Wasserstein distance} on $\mathcal{D}(\mathcal{H})$ is then defined as
\begin{equation} \label{eq:quantum_wass_dist}
    W_p^d (\rho,\sigma) = \left( \inf_{Q \in \mathcal{Q}(\rho,\sigma)} T_p^d(Q) \right)^{1/p}.
\end{equation}

\end{defn}

We will see in Proposition \ref{prop:trans_plan_inf_attained} that for finite-dimensional $\mathcal{H}$ and for $d$ continuous, the infimum in this definition is equivalent to the minimum. We will also see in that proposition that the minimum is attained at a transport plan with at most $2D^2$ elements, where $D$ is the dimension of the Hilbert space $\mathcal{H}$. However, we do not strictly limit the number of states in the definition of a transport plan in finite dimension. This is because defining a plan $Q$ from $\rho$ to $\sigma$ upper bounds the Wasserstein distance $W_p^d(\rho,\sigma)$ by $T_p^d(Q)^{1/p}$, and we do not want to restrict the set of plans for which this is possible.

Note that we have restricted to transporting from pure states in $\mathbb{P}\mathcal{H}$ to pure states on $\mathbb{P}\mathcal{H}$ above. This is because there is no natural extension of the distance $d$ to a transportation cost for entangled pure states on $\mathcal{H} \otimes \mathcal{H}$. In Appendix~\ref{appendix:entangled_transport} we discuss one possible extension and then argue that, at least for $p=1$, couplings with entangled states are not advantageous and we can restrict to the couplings defined above without loss of generality.

We could also have considered transport plans defined by an integrable measure $q$ on $\mathbb{P}\mathcal{H}_1 \times \mathbb{P}\mathcal{H}_2$. We discuss this possibility in Appendix~\ref{appendix:integrable_plans}, and show these do also not give an advantage over the transport plans as defined above. However, note that it is still possible to upper bound $W_p^d(\rho,\sigma)$ by $T_p^d(q)^{1/p}$ for such a $q$, where $T_p^d(q)$ is defined in Appendix~\ref{appendix:integrable_plans}.

As with many other ordered distances, we can also define the infinite-order quantum Wasserstein distance.

\begin{defn} \label{def:inf_order_wass_dist}
Let $\mathcal{H}$ be a separable Hilbert space and let $d$ be a distance on $\mathbb{P}\mathcal{H}$. For any states $\rho, \sigma \in \mathcal{D}(\mathcal{H})$, we define their \textbf{infinite-order quantum Wasserstein distance} as
    \begin{equation} \label{eq:quantum_wass_dist_infty}
        W_{\infty}^d (\rho,\sigma) = \inf_{Q \in \mathcal{Q(\rho,\sigma)}} \sup_{j \in J} d(\ket{\psi_j},\ket{\varphi_j})
    \end{equation}
    for quantum transport plans $Q = \{(q_j,\ket{\psi_j},\ket{\varphi_j})\}_{j \in J}$.
\end{defn}

As a natural consequence of the convex nature of these definitions, we see that this family of Wasserstein distances is monotone in $p$. Indeed, for any transport plan $Q$ between $\rho$ and $\sigma$ we have for $p_1 < p_2$ that
\begin{align} 
    T_{p_1}(Q)^{1/{p_1}} 
        &= \left(\sum_{j \in J} q_j d(\ket{\psi_j}, \ket{\varphi_j})^{p_1} \right)^{1/{p_1}} \\
        &\leq \left(\left(\sum_{j \in J} q_j \left(d(\ket{\psi_j}, \ket{\varphi_j})^{p_1}\right)^{p_2 / p_1} \right)^{1/{p_1}}\right)^{p_1 / p_2}\\
        &= \left(\sum_{j \in J} q_j d(\ket{\psi_j}, \ket{\varphi_j})^{p_2} \right)^{1/{p_2}} \\
        &= T_{p_2}(Q)^{1/{p_2}} 
\end{align}
by an application of Jensen's inequality. Taking the infimum over $Q \in \mathcal{Q}(\rho,\sigma)$, we get 
\begin{equation}
    W_{p_1}^d(\rho,\sigma) \leq W_{p_2}^d(\rho,\sigma). \label{eq:hierarchy_in_p}
\end{equation}
The same applies when $p_2 = \infty$, as $T_{p_1}(Q)^{1/p_1} \leq \max_{j \in J} d(\ket{\psi_j},\ket{\varphi_j})$.

This definition is given for all $1 \leq p \leq \infty$. However, at least in the classical case, the most important cases to consider for interesting applications are $p = 1, 2, \text{ and } \infty$~\cite{panaretos2019review, villani2008optimal, trashorras2018}.

Although many desirable properties of these $W_p^d$ are proven in Sec.~\ref{section:general_properties}, establishing a triangle inequality remains a difficult open problem. This is true of many attempts to establish a quantum generalisation of the classical Wasserstein distances, and further discussion of attempts to establish this are discussed in Appendix~\ref{appendix:tri_ineq}. However, by looking at the dual picture in Sec.~\ref{section:dual_picture}, we establish a closely related norm on the space of traceless Hermitian operators for which $W_1^d$ is an upper bound, and which agrees with the original distance $d$ on differences between pure states. For order $p = 1$ this norm allows a trade-off between ease of proving the triangle inequality, and operational interpretation in terms of moving mass. This dual formulation closely mirrors the Kantorovich-Rubinstein theorem in the classical case~\cite[p. 34]{villani2003topics}.
Furthermore, we show in Sec.~\ref{sec:special_instances} that it allows us to essentially recover the trace distance and the definition of \cite{de2020quantum}, thus our results on the dual definition give strong evidence that this $W_p^d$ is a good quantum generalisation.

\subsection{Other approaches}

Many proposals have been suggested for generalisations of the classical Wasserstein distance to quantum states, in various contexts, to varying degrees of success. Arguably the most desirable classical metrics to generalise are the discrete metric on a finite space and the Hamming distance on the hypercube.

The most prominent example is perhaps the first-order distance defined in~\cite{de2020quantum} which generalises the $\mathcal{W}_1$ distance on the hypercube whose underlying geometry is the Hamming distance.
They define states $\rho,\sigma$ on $(\mathbb{C}^d)^{\otimes n}$ to be {\em neighbouring} if there exists some qudit $i$ for which $\text{Tr}_i[\rho-\sigma] = 0$. In other words, $\rho$ and $\sigma$ are the same when qudit $i$ is traced out. The norm $\norm{\cdot}_{W_1^H}$ over traceless Hermitian operators is then defined such that its unit ball is the convex hull of all differences between neighbouring states. This means $\norm{\rho-\sigma}_{W_1^H}$ can be written as
\begin{equation} \label{eq:depalma_defn}
    \norm{\rho-\sigma}_{W_1^H} = \min \left\{\sum_{i=1}^n c_i : c_i > 0, \sum_{i=1}^n c_i (\rho^{(i)} - \sigma^{(i)}) = \rho-\sigma, \text{Tr}_i[\rho^{(i)}-\sigma^{(i)}] = 0 \right\}.
\end{equation}
This is a true metric and it has led to many applications such as in quantum spin systems~\cite{De_Palma_2023} and variational quantum algorithms~\cite{de2022limitations}. However, the approach taken is very specific to the Hamming distance and does not lend itself to general transport costs.

For order $p=2$, a number of definitions have been proposed in various contexts, such as~\cite{Golse2015OnTM} which defines a second-order optimal transport cost in the context of mean field limits that was shown in~\cite{caglioti2021on} to have links with the Brenier formulation of classical optimal transport~\cite[p.\ 238]{villani2003topics}.

More precisely, for a set of Hermitian operators $\{R_1, \dots, R_K\}$ on $\mathcal{H}$, the second-order transport cost of a coupling $\Pi$ on $\mathcal{H} \otimes \mathcal{H}$ by
\begin{equation} 
    \mathcal{C}(\Pi) = \sum_{i=1}^K \text{Tr}[(R_i \otimes \mathbb{I} - \mathbb{I} \otimes R_i)\Pi (R_i \otimes \mathbb{I} - \mathbb{I} \otimes R_i)] 
\end{equation}
from which a second-order Wasserstein distance is derived by taking the square root of the infimum of the cost of all couplings between $\rho$ and $\sigma$ in the usual way. Defining quantity $d(\ket{\psi},\ket{\varphi}) = \sum_{i=1}^K \norm{R_i\ket{\psi}}^2+\norm{R_i\ket{\varphi}}^2 - 2\braket{\psi | R_i | \psi}\braket{\varphi|R_i|\varphi}$, we see that this is a variation of our definition above, although with a specific $d$ which is not a true metric, and with the infimum taken over all couplings as opposed to just those which are separable.

Following the Brenier formulation more closely, the distance in~\cite{Carlen_2014} is also a second-order distance and has been used~\cite{datta2020relating} to prove a quantum version of the HWI inequality. Specifically, the 2-Wasserstein distance is defined here as the geodesic distance on the set of full-rank states equipped with a Riemannian metric defined by the continuity equation. This definition is a world away from our framework, as it leans heavily into the intricacies of the classical dynamical formulation.

We should also mention the definition in~\cite{chakrabarti2019quantum} and refined in~\cite{muellerhermes2022monotonicity} which defines a second-order cost based on couplings with a specific cost function given by an asymmetric projection. While it is conjectured that this gives a true distance, we show in Appendix~\ref{appendix:asym_bad} that this definition cannot be extended to other underlying geometries, such as analogues of the Hamming distance on multipartite spaces. Further approaches include the more flexible~\cite{DePalma2021}, an alternative to \cite{Golse2015OnTM}, which defines a second-order distance based on couplings that is not faithful. It has been conjectured~\cite{trevisanOTQS2022} that a natural modification of this quantity is a true distance, though this remains an open problem. A similarly flexible approach appears in~\cite{Friedland_2022}, following a naïve translation of the classical formulation into the quantum setting. This definition takes a cost matrix $C$ on $\mathcal{H} \otimes \mathcal{H}$, and defines a Wasserstein distance by
\begin{equation}
    W_p(\rho,\sigma) = \left(\inf_{\tau \in \mathcal{C}(\rho,\sigma)}\text{Tr}[C^p \tau] \right)^{1/p}
\end{equation}
where $\mathcal{C}(\rho,\sigma)$ is the set of quantum couplings of $\rho$ with $\sigma$. The particular case of $C$ being the projection onto the asymmetric subspace is studied in detail and coincides with the definition in~\cite{chakrabarti2019quantum}. In the general case, however, this is not shown to be a semidistance. 

Each of the approaches seems to be generalising one particular aspect or application of Wasserstein distances to the noncommutative setting, be it obtaining a distance for a given value of $p$ or a given underlying geometry of the Hilbert space. However, it is not clear how they relate to each other or how to extend them beyond their original setting. The definition in this work adapts to any order $p$ and any underlying metric $d$ on the set of pure states of the Hilbert space provided that $d$ satisfies a few basic continuity properties. This broad flexibility allows us to talk about the moments of the cost of moving between classical-quantum sources in great generality (Sec.~\ref{section:operational}) and also allows us to talk about the noise of an operator by comparing transport distances of different orders in an analogue of hypercontractivity (Sec.~\ref{section:hypercontractivity}).

To avoid confusion, we give in Table \ref{tab:W} a summary of all relevant Wasserstein distances used throughout this work.
\begin{table}[ht]
    \centering
    \begin{tabular}{c|p{0.7\textwidth}} \hline
        Notation & Definition \\ \hline \hline
        $\mathcal{W}_p^d(\mu,\nu)$ & The \textbf{classical} Wasserstein distance of order $p$ between measures $\mu$, $\nu$ on metric space $(X,d)$; see Equations \eqref{eq:class_wass_dist} and \eqref{eq:class_wass_dist_infty}.\\ \hline
        $W_p^d(\rho,\sigma)$ & The \textbf{quantum} Wasserstein distance of order $p$ between states $\rho$, $\sigma$ on Hilbert space $\mathcal{H}$ and metric $d$ on $\mathbb{P}\mathcal{H}$; see Equations \eqref{eq:quantum_wass_dist} and \eqref{eq:quantum_wass_dist_infty}. \\  \hline
        $\norm{\rho-\sigma}_{W_1^H}$ & The quantum Wasserstein norm of order one from~\cite{de2020quantum}; see Equation \eqref{eq:depalma_defn}. \\ \hline
        $W_p^H(\rho,\sigma)$ & The quantum Wasserstein distance of order $p$ on on $\mathcal{H} = (\mathbb{C}^d)^{\otimes n}$ defined by metric $d_H(\ket{\psi},\ket{\varphi}) = \norm{\ketbra{\psi}{\psi}-\ketbra{\varphi}{\varphi}}_{W_1^H}$; see Sec.~\ref{section:WpH}. Write $W_{p=1}^H$ when $p=1$ to avoid confusion with norm $\norm{\cdot}_{W_1^H}$. \\ \hline
        $W_p^1(\rho,\sigma)$ & The quantum Wasserstein distance of order $p$ on on $\mathcal{H} = \mathbb{C}^D$ defined by the trace distance $d_1(\ket{\psi},\ket{\varphi}) = \frac{1}{2}\norm{\ketbra{\psi}{\psi}-\ketbra{\varphi}{\varphi}}_{1}$; see Sec.~\ref{section:Wp1}. \\ \hline
        $W_p^C(\rho,\sigma)$ & The quantum Wasserstein distance of order $p$ on $\mathcal{H} = (\mathbb{C}^2)^{\otimes n}$ defined by the complexity geometry metric $d_C$ \cite{Nielsen2006}; see Sec.~\ref{section:WpC}. \\ \hline
    \end{tabular}
    \caption{An overview of the Wasserstein distances used in this work, and the relevant notation.}
    \label{tab:W}
\end{table}

\section{General properties} \label{section:general_properties}

The goals of this section are to prove some fundamental attributes of $W_p^d$ which will give an idea of how it behaves in the case of general $d$. Throughout we will require some basic regularity conditions on $d$. As we will see in more detail in the discussion following Cor. \ref{cor:Wp_nonneg_faithful}, not every metric $d$ on the set of pure states induces a faithful Wasserstein distance. Thus, we will restrict to metrics with respect to which the $2$-norm on the set of traceless self-adjoint operators is Hölder continuous. In other words, we require that there exists $\alpha \in (0,1]$ and $C>0$ for which for all $\ket{\psi},\ket{\varphi} \in \mathbb{P}\mathcal{H}$ we have
\begin{equation} \label{eq:holder_cts}
    Cd(\ket{\psi},\ket{\varphi})^{\alpha} \geq \norm{\ketbra{\psi}{\psi} - \ketbra{\varphi}{\varphi}}_2.
\end{equation}
Note that as we work in finite dimension, the $2$-norm is chosen here without loss of generality. In the case $\alpha = 1$, this is equivalent to the 2-norm being Lipschitz with respect to $d$. The constant $C$ is a function of the metric space $(\mathbb{P}\mathcal{H},d)$ and so can depend on the dimension of $\mathcal{H}$. For some other properties such as continuity, we will require that $d$ be continuous, noting that on $\mathcal{H}$ of finite dimension this is equivalent to $d$ being uniformly continuous.

Given the main philosophical motivation of this definition, that the transport distance between point masses in the classical setting is given by the underlying metric, it's important to show the quantum equivalent here. That is, for all orders $p$ the quantum Wasserstein distance between pure states agrees with the underlying metric.

\begin{prop} \label{prop:agree-on-pure-states}
    Let $\mathcal{H}$ be a separable Hilbert space with distance $d$ on $\mathbb{P}\mathcal{H}$. Then for pure states $\ketbra{\psi}{\psi}, \ketbra{\varphi}{\varphi}$, we have
\begin{equation}
    W_p^d(\ketbra{\psi}{\psi}, \ketbra{\varphi}{\varphi}) = d(\ket{\psi},\ket{\varphi}).
\end{equation}
\end{prop}
\begin{proof}
    As $\ketbra{\psi}{\psi}$ and $\ketbra{\varphi}{\varphi}$ are pure, the only permitted transport plan is $Q = \{(1,\ket{\psi},\ket{\varphi})\}$, which has cost 
    \begin{equation}
        T_p^d(Q) = d(\ket{\psi},\ket{\varphi})^p
    \end{equation}
    and therefore we have $W_p^d(\ketbra{\psi}{\psi}, \ketbra{\varphi}{\varphi}) = d(\ket{\psi},\ket{\varphi})$.
\end{proof}

We can then employ the enforced Hölder continuity condition to prove that $W_p^d$ is nondegenerate. This is captured in the following lemma. Zero self-distance in $W_p^d$ is also fairly clear, as is symmetry. Note that these fundamental properties have not always been present in previous definitions such as in \cite{DePalma2021} where the definition has non-zero self-distance.

\begin{lemma} \label{prop:more-than-norm}
    Let $\mathcal{H}$ be a separable Hilbert space with distance $d$ on $\mathbb{P}\mathcal{H}$ with respect to which the $2$-norm is Hölder continuous with constant $C > 0$ and exponent $\alpha \in (0,1]$. Then for all $\rho, \sigma \in \mathcal{D}(\mathcal{H})$, we have
    \begin{equation}
        W_p^d(\rho,\sigma) \geq \frac{1}{C}\norm{\rho-\sigma}_2^{1/\alpha}.
    \end{equation}
\end{lemma}
\begin{proof}
We have $Cd(\ket{\psi},\ket{\varphi})^{\alpha} \geq \norm{\ketbra{\psi}{\psi} - \ketbra{\varphi}{\varphi}}_2$ on $\mathbb{P}\mathcal{H}$, for $C>0$ and $\alpha \in (0,1]$.
Let $\left\{\left(q_j, \ket{\psi_j}, \ket{\varphi_j}\right) \right\}$ be a quantum transport plan between $\rho$ and $\sigma$. The functions $x \mapsto x^p$ and $X \mapsto \norm{X}_2^{1 / \alpha}$ for traceless $X$ are both convex. Hence
    \begin{align}
        \sum_j q_j d(\ket{\psi_j}, \ket{\varphi_j})^p 
            &\geq \left(\sum_j q_j d(\ket{\psi_j}, \ket{\varphi_j}) \right)^p \\ 
            &\geq \left(\sum_j q_j \frac{1}{C}\norm{\ketbra{\psi_j}{\psi_j}-\ketbra{\varphi_j}{\varphi_j}}_2^{1/\alpha} \right)^p \\
            &\geq \left(\frac{1}{C} \norm{\sum_j q_j (\ketbra{\psi_j}{\psi_j}-\ketbra{\varphi_j}{\varphi_j})}_2^{1/\alpha} \right)^p \\
            &= \frac{1}{C^p}\norm{\rho-\sigma}_2^{p/\alpha}.
    \end{align}
Hence taking the infimum and the $p^{\text{th}}$ root, we have $W_p^d(\rho,\sigma) \geq \frac{1}{C}\norm{\rho-\sigma}_2^{1/\alpha}$.
\end{proof}

\begin{cor} \label{cor:Wp_nonneg_faithful}
    Let $d$ be a distance on $\mathbb{P}\mathcal{H}$ with respect to which the $2$-norm is Hölder continuous. For all $\rho, \sigma \in \mathcal{D}(\mathcal{H})$, we have $W_p^d(\rho,\sigma) \geq 0$ with equality iff $\rho=\sigma$.
\end{cor}
\begin{proof}
    $W_p^d(\rho,\sigma) \geq 0$ is clear, and equality when $\rho = \sigma$ can be obtained from the taking transport plan $\{(c_j, \ket{\psi_j},\ket{\psi_j}) \}_{j \in J}$ where $\rho = \sum_{j \in J} c_j \ketbra{\psi_j}{\psi_j}$ is a spectral decomposition. Faithfulness is a direct consequence of Proposition \ref{prop:more-than-norm}.
\end{proof}

For a general distance $d$, not necessarily satisfying property \eqref{eq:holder_cts}, we get from the same proof above that $W_p^d(\rho,\sigma) \geq 0$ and $W_p^d(\rho,\rho) = 0$. It is not guaranteed, however, that a general distance on $\mathbb{P}\mathcal{H}$ leads to a nondegenerate $W_p^d$. For example, let $\mathcal{H} = \mathbb{C}^2$ with the standard basis $\{ \ket{0}, \ket{1}\}$. For any non-negative real-valued function $f$ on $\mathbb{P}\mathcal{H}$ with $f(\ket{0}) = 0$ and $f$ positive elsewhere, we can define a metric $d$ as follows:
\begin{equation}
    d(\ket{\psi},\ket{\varphi}) =
\begin{cases}
    0 & \text{if } \ket{\psi} = \ket{\varphi} \\
    f(\ket{\psi}) + f(\ket{\varphi}) & \text{otherwise}. \\
\end{cases}
\end{equation}
This forms a version of the SNCF metric, also known as the centralised railway metric ~\cite[p.~327]{Deza2013}, with $\ket{0}$ at the centre. This is a metric in which travel is only permitted along rays emanating from a central point. We will consider $f$ such that 
\begin{equation}
    f\left(\sqrt{\frac{n-1}{n}}\ket{0} + \sqrt{\frac{1}{n}}\ket{1}\right) = f\left(\sqrt{\frac{n-1}{n}}\ket{1} - \sqrt{\frac{1}{n}}\ket{0}\right) = \frac{1}{n}
\end{equation}
for $n \in \mathbb{N}$ and $f(\ket{\psi}) = 2$ otherwise. Consider the sequence of transport plans
\begin{equation}
   Q_n = \left\{\left(\frac{1}{2},\ket{0},\sqrt{\frac{n-1}{n}}\ket{0} + \sqrt{\frac{1}{n}}\ket{1}\right),\left(\frac{1}{2},\ket{0},\sqrt{\frac{n-1}{n}}\ket{1} - \sqrt{\frac{1}{n}}\ket{0}\right)\right\}
\end{equation}
for $n \in \mathbb{N}$. These are all plans which transport $\ketbra{0}{0}$ onto $\frac{\mathbb{I}}{2}$ and each has cost $T_1^d(Q_n) = \frac{1}{n}$. This gives $W_1^d\left(\ketbra{0}{0}, \frac{\mathbb{I}}{2}\right) \leq T_1^d(Q_n) = \frac{1}{n} \to 0$, and therefore $W_1^d\left(\ketbra{0}{0},\frac{\mathbb{I}}{2}\right) = 0$. We have shown that the Hölder continuity condition is sufficient for $W_p^d$ to be nondegenerate, but it is not immediately obvious whether or not it is necessary.

Having established that $W_p^d$ is at least a semidistance, we turn our attention to its basic behavioural properties. Firstly, we note that all symmetries of $(\mathbb{P}\mathcal{H},d)$ are inherited by the quantum Wasserstein distances.

\begin{prop} \label{prop:Wp_inherits_symmetries}
Let $1 \leq p \leq \infty$ and $d$ be a metric on $\mathbb{P}\mathcal{H}$. The group of unitary symmetries of the underlying metric $d$ is exactly the group of conjugational symmetries of $W_p^d$.
\end{prop}
\begin{proof}
Let $U$ be a symmetry of $d$. $U$ is invertible, therefore there's a direct correspondence
    \begin{equation}
        \{(q_j,\ket{\psi_j},\ket{\varphi_j})\}_{j \in J} \longleftrightarrow \{(q_j, U\ket{\psi_j},U\ket{\varphi_j})\}_{j \in J}
    \end{equation}
between quantum transport plans from $\rho$ to $\sigma$ and from $U\rho U^\dagger$ to $U \sigma U^\dagger$. The distance $d$ is invariant under $U$ so cost is preserved under this correspondence, therefore the optimal cost is also preserved.

Conversely, let unitary $V$ be a conjugational symmetry of $W_p^d$. Then for all $\ket{\psi}, \ket{\varphi} \in \mathbb{P}\mathcal{H}$, we have
\begin{equation}
    d(\ket{\psi},\ket{\varphi}) = W_p^d(\ketbra{\psi}{\psi},\ketbra{\varphi}{\varphi}) =  W_p^d(V\ketbra{\psi}{\psi}V^\dagger ,V\ketbra{\varphi}{\varphi}V^\dagger ) = d(V\ket{\psi}, V\ket{\varphi})
\end{equation}
and so $V$ is a symmetry of $d$.
\end{proof}
This allows us to prove a result on data processing for mixed unitary channels. The data processing inequality is a central concept in both classical and quantum information theory, and represents the idea that you can't create new information by processing old data: all the information you have about an object is encoded in the rawest data you have. This concept also exists in classical optimal transport: for measures $\mu$ on a metric space $(X,d)$, measurable sets $A,B$ and a measurable map $f: X \to X$, we define the \textit{pushforward} $f_* \mu$ by $f_*\mu (B) = \mu(f_{-1}(B))$. If $f$ is $1$-Lipschitz, then $\mathcal{W}_p^d(\mu,\nu) \geq \mathcal{W}_p^d(f_* \mu, f_* \nu)$.
While many quantities in quantum information theory satisfy data processing for all quantum channels, in the case of the Wasserstein distances we would only expect the data processing inequality to be satisfied for those channels which respect the geometry of the underlying space. For example, for an arbitrary unitary map $U$ we would not expect the channel $\rho \mapsto U\rho U^\dagger$ to satisfy data processing unless $U$ has $d(\ket{\psi},\ket{\phi}) \geq d(U\ket{\psi},U\ket{\varphi})$ for all states. In general, we would only expect channels which respect the geometry of the underlying space to satisfy data processing. For general $d$, the largest class of channels from $\mathcal{D}(\mathcal{H}) \to \mathcal{D}(\mathcal{H})$
that respect the geometry of the underlying space is the class of mixed unitary channels whose unitaries are isometries of $(\mathbb{P} \mathcal{H},d)$.
In this case, we do indeed have data processing as follows:

\begin{prop}
    Suppose $\Phi$ is a mixed unitary channel written as a countable sum
    \begin{equation}
        \Phi(\rho) = \sum_{k \in K} a_k U_k \rho U_k^\dagger \qquad a_k \geq 0, \sum_{k \in K}a_k=1
    \end{equation}
    for which multiplication by each $U_k$ is an isometry with respect to $d$. Then for all orders $1 \leq p \leq \infty$ and all states $\rho$, $\sigma$, we have $W_p^d(\rho,\sigma) \geq W_p^d(\Phi(\rho),\Phi(\sigma))$.
\end{prop}
\begin{proof}
    Let $Q = \{(q_j,\ket{\psi_j},\ket{\varphi_j})\}_{j \in J}$ be any $p^{\text{th}}$-order transport plan between $\rho$ and $\sigma$. We can then define a transport plan $Q' = \{(q_ja_k, U_k\ket{\psi_j},U_k\ket{\varphi_j})\}_{j \in J, k \in K}$ between $\Phi(\rho)$ and $\Phi(\sigma)$ which has the same $p^{\text{th}}$-order transport cost as $Q$. Taking the infimum over $Q$ gives the result.
\end{proof}

It is also possible, in finite dimensions, to guarantee the existence of an optimal transport plan when the underlying distance $d$ is continuous and to bound the number of elements of an optimal transport plan. This applies in particular when $d$ is induced by a norm: that is to say, that there exists a norm $\norm{\cdot}_d$ on the self-adjoint traceless operators on $\mathcal{H}$ such that $d(\ket{\psi},\ket{\varphi}) = \norm{\ketbra{\psi}{\psi}-\ketbra{\varphi}{\varphi}}_d$ everywhere.

\begin{prop} \label{prop:trans_plan_inf_attained}
    Let $d$ be a continuous distance on $\mathbb{P}{\mathcal{H}}$ and let $\mathcal{H}$ have dimension $D < \infty$. The infimum in \eqref{eq:quantum_wass_dist} is attained with a transport plan containing at most $2D^2$ elements.
\end{prop}
\begin{proof}
    To show that the infimum over all plans is the same as the infimum over plans of finite size, let $Q = \{(q_j, \ket{\psi_j}, \ket{\varphi_j} ) \}_{j \in \mathbb{N}}$ be an infinite-sized transport plan between $\rho$ and $\sigma$. Let \begin{equation}
        Q' = \{ (q_j, \ket{\psi_j}, \ket{\varphi_j}): j \in \mathbb{N}, q_j \geq \epsilon \} \cup Q''
    \end{equation}
    for $Q''$ a partial transport plan of the form in equation \eqref{eq:example_transport_plan} between $\rho - \sum_{j : q_j \geq \epsilon} q_j \ketbra{\psi_j}{\psi_j}$ and $\sigma - \sum_{j : q_j \geq \epsilon} q_j \ketbra{\varphi_j}{\varphi_j}$. This $Q'$ is then a finite transport plan between $\rho$ and $\sigma$. As $\epsilon \to 0$, the transport cost of the first part of $Q'$ tends to $T_p^d(Q)$ and the cost of $Q''$ tends to 0, which proves the equality of the infima.

    For the bound $2D^2$, suppose $Q = \{(q_j,\ket{\psi_j},\ket{\varphi_j}\}_{j\in J}$ is a transport plan between $\rho$ and $\sigma$ with more than $2D^2$ elements. We will show that there exists a transport plan $Q'$ with strictly fewer elements than $Q$, whose $p^{\text{th}}$-order transport cost is at most that of $Q$. Let $\mathcal{M}_D^{\text{sa}}$ be the space of self-adjoint operators on $D$ dimensions. The space $\mathcal{M}^{\text{sa}}_D \times \mathcal{M}_D^{\text{sa}}$ has real dimension $2D^2$. Therefore among the elements $(\ketbra{\psi_j}{\psi_j},\ketbra{\varphi_j}{\varphi_j})$ in $\mathcal{M}^{\text{sa}}_D \times \mathcal{M}_D^{\text{sa}}$ we can find a nontrivial linear relation
    \begin{equation}
        \sum_{j \in J} c_j (\ketbra{\psi_j}{\psi_j},\ketbra{\varphi_j}{\varphi_j}) = 0.
    \end{equation}
    Define subsets $K$, $L$ of $J$ by $K = \{k \in J: c_k > 0\}$ and $L = \{l \in J: c_l < 0\}$. We can then rewrite this as
    \begin{equation}\label{eq:trans_plan_inf_attained_nontriv_lin_rel}
        \sum_{k \in K} c_k (\ketbra{\psi_{k}}{\psi_{k}},\ketbra{\varphi_{k}}{\varphi_{k}}) = \sum_{l \in L} (-c_l) (\ketbra{\psi_{l}}{\psi_{l}},\ketbra{\varphi_{l}}{\varphi_{l}})
    \end{equation}
so that the coefficients $c_k$ and $-c_l$ are strictly positive.
Without loss of generality, we may assume that 
\begin{equation} \label{eq:trans_plan_inf_attained_inequality_wlog}
    \sum_{k \in K} c_k d(\ket{\psi_{k}},\ket{\varphi_{k}})^p \geq \sum_{l \in L} (-c_l) d(\ket{\psi_{l}},\ket{\varphi_{l}})^p.
\end{equation}
We will aim to replace a portion of the transport plan corresponding to the left-hand side, with a portion corresponding to the right-hand side.

Let $m = \min_{k \in K} \left\{\frac{q_{k}}{c_k} \right\}$, and let $k'$ be a minimiser. We can then form a new transport plan 
\begin{equation} \label{eq:trans_plan_inf_attained_new_plan}
    Q' = \{(q_j,\ket{\psi_j},\ket{\varphi}_j) \}_{j \notin K \cup L} \cup \{(q_k - mc_k,\ket{\psi_k},\ket{\varphi_k}) \}_{k \in K} \cup \{(q_l - mc_l,\ket{\psi_l},\ket{\varphi_l}) \}_{l \in L}
\end{equation}
replacing $m$ times the left hand side of \eqref{eq:trans_plan_inf_attained_nontriv_lin_rel} with $m$ times the right hand side. $k'$ attains the minimum in the definition of $m$ so any multiple of $(\ket{\psi_{k'}},\ket{\varphi_{k'}})$ is removed from the plan. The linear relation means that the resulting plan is still a transport plan between $\rho$ and $\sigma$, and \eqref{eq:trans_plan_inf_attained_inequality_wlog} means that the cost is not increased. Hence we can find a transport plan between $\rho$ and $\sigma$ with fewer elements, without increasing the transport cost.

We may then optimise the transport cost over all transport plans of size at most $2D^2$. This set is compact and the $p^{\text{th}}$-order cost \eqref{eq:quantum_trans_cost} of a transport plan is a continuous function of the transport plan $Q$ for $d$ continuous, so the infimum is attained.
\end{proof}

The existence of an optimal transport plan will be particularly useful later when looking at examples, applications, and further properties.

Another key property of this $W_p^d$ is continuity, which we can prove provided that $d$ is uniformly continuous. This will be necessary for a coherent definition of a dual in Sec.~\ref{section:dual_picture}.

\begin{prop} \label{prop:Wp_cts}
    Suppose $d$ is uniformly continuous on $\mathbb{P}\mathcal{H}$ and let $1 \leq p < \infty$. Then $W_p^d$ is uniformly continuous.
\end{prop}
\begin{proof}
The proof is quite technical and not particularly instructive, so has been placed in Appendix~\ref{appendix:continuity} for the convenience of the reader.

\end{proof}

For $p=1$ specifically, we furthermore have joint convexity.
\begin{prop} \label{prop:joint_convexity_p=1}
    $W_1^d$ is jointly convex.
\end{prop}
\begin{proof}
    Suppose $\rho_1$, $ \sigma_1$, $ \rho_2$, and $\sigma_2$ are quantum states, and let $r_1 + r_2 = 1$, $r_i \geq 0$. Let $Q_1 = \{(q_{1,j},\ket{\psi_{1,j}},\ket{\varphi_{1,j}})\}_{j \in J}$ be any transport plan between $\rho_1$ and $\sigma_1$, and $Q_2  = \{(q_{2,k},\ket{\psi_{2,k}},\ket{\varphi_{2,k}})\}_{k \in K}$ any transport plan between $\rho_2$ and $\sigma_2$. Then $Q = r_1Q_1 \cup r_2Q_2 \coloneqq \{(r_1q_{1,j},\ket{\psi_{1,j}},\ket{\varphi_{1,j}})\}_{j \in J}j \cup \{(r_2q_{2,k},\ket{\psi_{2,k}},\ket{\varphi_{2,k}})\}_{k\in K}$ is a transport plan between $r_1\rho_1 + r_2\rho_2$ and $r_1\sigma_1 + r_2\sigma_2$.

    Then we have 
    \begin{align}
        T_1^d(Q) 
            = \sum_{j \in J} r_1q_{1,j} d(\ket{\psi_{1,j}},\ket{\varphi_{1,j}}) + \sum_{k  \in K} r_2q_{2,k} d(\ket{\psi_{2,k}},\ket{\varphi_{2,k}})
            = r_1T_1^d(Q_1) + r_2T_1^d(Q_2).
    \end{align}
    Taking the infimum over $Q_1$ and $Q_2$ shows that $W_1^d(r_1\rho_1 + r_2\rho_2, r_1\sigma_1 + r_2\sigma_2) \leq r_1W_1^d(\rho_1,\sigma_1) + r_2W_1^d(\rho_2,\sigma_2)$.
\end{proof}
For $p > 1$, joint convexity does not hold even in the classical case: consider for example the Hamming distance on the cube and measures $\mu_1 = \delta_{000}, \nu_1 = \delta_{001}, \mu_2 = \delta_{111}, \nu_2 = \delta_{111}$ with $r_1 = r_2 = \frac{1}{2}$. To show that joint convexity does not hold in the quantum case for $p > 1$, we give an example in Section \ref{section:Wp1}.

It is difficult to talk about how tensor products interact with $W_p^d$ in general, as for two spaces $\mathcal{H}_1, \mathcal{H}_2$ whose projectivisations are equipped with metrics $d_1, d_2$ respectively, there is no natural choice of a metric on $\mathbb{P}(\mathcal{H}_1 \otimes \mathcal{H}_2)$. However, if we equip $\mathbb{P}(\mathcal{H}_1 \otimes \mathcal{H}_2)$ with a distance which is subadditive with respect to $d_1$ and $d_2$, then $W_p^d$ is also subadditive:

\begin{prop} \label{prop:subadditive}
    Let $\mathcal{H}_a$, $\mathcal{H}_b$ be finite-dimensional Hilbert spaces whose projectivisations $\mathbb{P}\mathcal{H}_a, \mathbb{P}\mathcal{H}_b$ are equipped with metrics $d_a, d_b$ respectively. Let $\mathbb{P}(\mathcal{H}_a \otimes \mathcal{H}_b)$ be equipped with metric $d$ which satisfies
    \begin{equation} \label{eq:prop_subadditive_distance_condition}
        d_a(\ket{\psi_a},\ket{\varphi_a}) + d_b(\ket{\psi_b},\ket{\varphi_b}) \geq d(\ket{\psi_a}\otimes \ket{\psi_b}, \ket{\varphi_a}\otimes \ket{\varphi_b}).
    \end{equation}
    Then for all $p$ and all states $\rho_a, \sigma_a$ on $\mathcal{H}_a$ and $\rho_b, \sigma_b$ on $\mathcal{H}_b$, we have
    \begin{equation}
        W_p^{d_a}(\rho_a,\sigma_a) + W_p^{d_b}(\rho_b, \sigma_b) \geq W_p^d(\rho_a \otimes \rho_b, \sigma_a \otimes \sigma_b).
    \end{equation}
\end{prop}
\begin{proof}
    Let $Q_a = \{(q_{j,a}, \ket{\psi_{j,a}}, \ket{\varphi_{j,a})}\}_{j \in J}$ and $Q_b = \{(q_{k,b}, \ket{\psi_{k,b}}, \ket{\varphi_{k,b}})\}_{k \in K}$ be transport plans from $\rho_a$ to $\sigma_a$ and from $\rho_b$ to $\sigma_b$ respectively. We'll show that the transport plan $Q = \{(q_{j,a}q_{k,b}, \ket{\psi_{j,a}} \otimes \ket{\psi_{k,b}}, \ket{\varphi_{j,a}} \otimes \ket{\varphi_{k,b}} \}_{j \in J, k \in K}$ has
    \begin{equation}
        T_p^{d_a}(Q_a)^{1/p} + T_p^{d_b}(Q_b)^{1/p} \geq T_p^d(Q)^{1/p}.
    \end{equation}
    Letting $X_a$ be the nonnegative random variable taking value $j$ with probability $q_{j,a}$, and similarly $X_b$ independent of $X_a$ taking value $k$ with probability $q_{k,b}$. Let $A_a = d_a(\ket{\psi_{X_a,a}},\ket{\varphi_{X_a,a}})$ and $A_b = d_b(\ket{\psi_{X_b,b}},\ket{\varphi_{X_b,b}})$. We see for $i = a, b $ that $T_p^{d_i}(Q_i)^{1/p}$ is by its definition the $p^{\text{th}}$ root of the $p^{\text{th}}$ moment of $A_i$. As the $p^{\text{th}}$ root of the $p^{\text{th}}$ moment is an $L_p$ norm on random variables on the measure space $(\mathcal{X}, \mu)$ induced by $X_a$ and $X_b$, we know that 
    \begin{equation}
        T_p^{d_a}(Q_a)^{1/p} + T_p^{d_b}(Q_b)^{1/p} \geq \mathbb{E}_\mu \left[\left(A_a + A_b\right)^p\right]^{1/p}.
    \end{equation}
    The equation \eqref{eq:prop_subadditive_distance_condition} means that $A_a + A_b \geq d(\ket{\psi_{j,a}}\otimes \ket{\psi_{k,b}}, \ket{\varphi_{j,a}}\otimes \ket{\varphi_{k,b}})$, and so
    \begin{equation}
        \mathbb{E}_\mu \left[\left(A_a + A_b\right)^p\right]^{1/p} \geq \mathbb{E}_\mu \left[ d(\ket{\psi_{1,X_a}}\otimes \ket{\psi_{2,X_b}}, \ket{\varphi_{1,X_a}}\otimes \ket{\varphi_{2,X_b}})^p \right]^{1/p} = T_p^d(Q)^{1/p}
    \end{equation}
    as claimed.
\end{proof}

Note that this result also holds fo $p = \infty$, as can be seen by replacing $W_p$ by $W_{\infty}$ and $T_p(\cdot)^{1/p}$ by $\sup_{j \in J} d(\ket{\psi_j},\ket{\varphi_j})$ in the proof.

\subsection{Dual picture} \label{section:dual_picture}

We can now take a look at the dual picture for the first-order transport distance, with some regularity conditions on the underlying metric $d$ which allow the dual to be well-defined. As we will see later, this will then allow us to define a version of the Wasserstein distance for $p=1$ that satisfies the triangle inequality and inherits many of the properties of the original $W_1^d$.

\begin{defn} \label{def:d_lipschitz}
    Let $d$ be a metric on $\mathbb{P}\mathcal{H}$. Suppose $d$ is continuous and that the 2-norm is Lipschitz with respect to $d$. 
    We define the \textbf{dual constant of a self-adjoint operator $O$} as
    \begin{equation} \label{eq:d_lipschitz}
        L_{d}(O) = \sup_{\rho \neq \sigma} \frac{\text{{\em Tr}}[O(\rho-\sigma)]}{W_1^d(\rho,\sigma)}.
    \end{equation}
\end{defn}
We refer to this as the $d$-Lipschitz constant. Usually, in the classical case, the Lipschitz constant of a function with respect to a metric is defined by taking the supremum only with respect to point masses, not probability distributions. As we see below in Proposition~\ref{prop:same_pure}, we can also take w.l.o.g.\ the supremum only with respect to pure states and only define it with mixed states for convenience.

\begin{prop}\label{prop:same_pure}
     Let $d$ be a metric on $\mathbb{P}\mathcal{H}$. Suppose that the 2-norm is Lipschitz with respect to $d$. Then
     \begin{equation} \label{eq:lipschitz_achieved_at_pure}
         L_d(O) = \sup_{\ket{\psi} \neq \ket{\varphi}} \frac{\text{{\em Tr}} [O(\ketbra{\psi}{\psi} - \ketbra{\varphi}{\varphi})]}{d(\ket{\psi},\ket{\varphi})}.
     \end{equation}
\end{prop}
\begin{proof}
    For any $\rho$ and $\sigma$, let $Q = \{(q_j,\ket{\psi_j},\ket{\varphi_j})\}_{j \in J}$ be an transport plan between them whose first-order transport cost is at most $W_1^d(\rho,\sigma) + \epsilon$. Then
    \begin{align}
    L_d(O)
            & \leq \sup_{\rho \neq \sigma} \frac{\text{Tr}[O(\rho-\sigma)]}{W_1^d(\rho,\sigma)} \\
            &\leq \sup_{\rho \neq \sigma} \frac{\sum_{j \in J} q_j \text{Tr}[O(\ketbra{\psi_j}{\psi_j} - \ketbra{\varphi_j}{\varphi_j})]}{\sum_{j \in J} q_j (d(\ket{\psi_j},\ket{\varphi_j})-\epsilon)} \\
            &\leq \sup_{\rho \neq \sigma} \max_{j \in J} \frac{\text{Tr}[O(\ketbra{\psi_j}{\psi_j} - \ketbra{\varphi_j}{\varphi_j})]}{ d(\ket{\psi_j},\ket{\varphi_j})-\epsilon} \\
            &\leq\sup_{\ket{\psi} \neq \ket{\varphi}} \frac{\text{Tr}[O(\ketbra{\psi}{\psi} - \ketbra{\varphi}{\varphi})]}{d(\ket{\psi},\ket{\varphi})-\epsilon} \\
            &= \sup_{\ket{\psi} \neq \ket{\varphi}} \frac{\text{Tr}[O(\ketbra{\psi}{\psi} - \ketbra{\varphi}{\varphi})]}{W_1^d(\ketbra{\psi}{\psi},\ketbra{\varphi}{\varphi}) - \epsilon}  \\
            &\to  \sup_{\ket{\psi} \neq \ket{\varphi}} \frac{\text{Tr}[O(\ketbra{\psi}{\psi} - \ketbra{\varphi}{\varphi})]}{W_1^d(\ketbra{\psi}{\psi},\ketbra{\varphi}{\varphi})} \qquad \text{as } \epsilon \to 0 \\
            &\leq L_d(O)
    \end{align}
    and so this is a chain of equalities.
\end{proof}
Importantly, this means that this $L_d$ is a function of the metric $d$ and is independent of our definition of $W_p^d$. It also means that our choice of $p=1$ in the definition is arbitrary: defining $L_{d,p}(O) = \sup_{\rho \neq \sigma}\text{Tr}[O(\rho-\sigma)]/W_p^d(\rho,\sigma) $ gives the same quantity as $L_d(O)$. It is also important to note that under regularity conditions only slightly stricter than already established on $d$, this quantity is a norm.

\begin{prop}
    Let $d$ be a metric on $\mathbb{P}\mathcal{H}$ such that the 2-norm is Lipschitz with respect to $d$. Then $L_{d}(O)$ is a norm on the space of traceless self-adjoint operators.
\end{prop}
\begin{proof}
    For finiteness, we use $|\text{Tr}[O(\ketbra{\psi}{\psi} - \ketbra{\varphi}{\varphi})]| \leq \norm{O}_{\infty} \norm{\ketbra{\psi}{\psi} - \ketbra{\varphi}{\varphi}}_1$ by Hölder's inequality and then that $\norm{O}_{\infty} \norm{\ketbra{\psi}{\psi} - \ketbra{\varphi}{\varphi}}_1 \leq 2 \norm{O}_{\infty} \norm{\ketbra{\psi}{\psi} - \ketbra{\varphi}{\varphi}}_2$ since $\ketbra{\psi}{\psi} - \ketbra{\varphi}{\varphi}$ has rank at most $2$. As Lipschitz is the same as Hölder continuous with exponent $\alpha = 1$, we have that this is at most $2C\norm{O}_{\infty}W_1^d(\ketbra{\psi}{\psi},\ketbra{\varphi}{\varphi})$ for some positive constant $C$ by Proposition \ref{prop:more-than-norm}.

    To show nondegeneracy, suppose $O$ is nonzero and let $\ket{\psi}$, $\ket{\varphi}$ be eigenstates of $O$ with positive and negative eigenvalues respectively, which must exist as $O$ is traceless. Then $\rho = \ketbra{\psi}{\psi}$, $\sigma = \ketbra{\varphi}{\varphi}$ gives $\text{Tr}[O(\rho-\sigma)]/W_1^d(\rho,\sigma) > 0$.

    For the other norm properties, clearly $L_{d}(\lambda O) = |\lambda | L_{d}(O)$. And suppose we have Hermitian operators $O_1$ and $O_2$. For any states $\rho$ and $\sigma$, we have
    \begin{equation} \label{eq:d_lipschitz_is_norm_tri_ineq}
        \text{Tr}[(O_1+O_2)(\rho-\sigma)]/W_1^d(\rho,\sigma) = \text{Tr}[O_1(\rho-\sigma)]/W_1^d(\rho,\sigma) + \text{Tr}[O_2(\rho-\sigma)]/W_1^d(\rho,\sigma) \leq L_{d}(O_1) + L_{d}(O_2)
    \end{equation}
    and taking the supremum of the left over such $\rho$ and $\sigma$ proves the triangle inequality for $L_{d}$.
\end{proof}

We can then dualise one more time and consider the norm
\begin{align}
\|\rho-\sigma\|_{DW_1^d}=\sup\limits_{L_{d}(O)\leq 1}\text{Tr}[O(\rho-\sigma)].
\end{align}

\begin{prop} \label{prop:dual_basic_properties}
    Let $\norm{\cdot}_{DW_1^d}$ be defined by $\norm{X}_{DW_1^d} = \sup\limits_{L_{d}(O)\leq 1}\text{\em {Tr}}[OX]$ on the space of traceless Hermitian operators. Then this is a norm, and on quantum states $\rho,\sigma$ we have
    \begin{equation}
        W_1^d(\rho,\sigma) \geq \norm{\rho-\sigma}_{DW_1^d}.
    \end{equation}
\end{prop}
\begin{proof}
    For the norm, it is clear that when $X=0$ we have $\norm{X}_{WD_1^d}$ equal to zero. The fact that $\norm{\lambda X}_{DW_1^d} = | \lambda | \norm{X}_{DW_1^d} $ is clear from the definition, and the triangle inequality holds as
    \begin{equation}
        \norm{X_1+X_2}_{DW_1^d} = \sup_{L_d(O) \leq 1} \text{Tr}[O(X_1+X_2)] \leq \sup_{L_d(O_1) \leq 1} \text{Tr}[O_1X_1] + \sup_{L_d(O_2) \leq 1} \text{Tr}[O_2X_2] = \norm{X_1}_{DW_1^d}+\norm{X_2}_{DW_1^d}.
    \end{equation}
    For nondegeneracy, take $O$ to be the projector onto the positive part of $X$ normalised to $L_d(O) = 1$, which gives $\text{Tr}[OX] > 0$.

    On states $\rho,\sigma$, let $O$ be any operator with $L_d(O) \leq 1$. Then 
    \begin{equation}
        1 \geq L_d(O) = \sup_{\rho' \neq \sigma'} \frac{\text{Tr}[O(\rho'-\sigma')]}{W_1^d(\rho',\sigma')} \geq \frac{\text{Tr}[O(\rho-\sigma)]}{W_1^d(\rho,\sigma)}
    \end{equation}
    and so $\text{Tr}[O(\rho-\sigma)] \leq W_1^d(\rho,\sigma)$. Taking the supremum over $O$ gives $\norm{\rho-\sigma}_{DW_1^d} \leq W_1^d(\rho,\sigma)$.
\end{proof}

The inequality $W_1^d(\rho,\sigma) \geq \norm{\rho-\sigma}_{DW_1^d}$ is important in our understanding of these quantities. It is as yet unclear whether or not equality holds in this equation, and establishing the conditions for equality is an important open problem, not least because this would allow us to recover the triangle inequality for $W_1^d$. There are some simple necessary conditions for equality in this equation, namely that $W_1^d$ depends only on $\rho - \sigma$, and that it scales like a norm.

However, we can consider the case where there exists a metric $\norm{\cdot}_d$ such that $d(\ket{\psi},\ket{\varphi})=\norm{\ketbra{\psi}{\psi}-\ketbra{\varphi}{\varphi}}_d$. Replacing $\norm{\cdot}_2$ with $\norm{\cdot}_d$ in the proof of Proposition \ref{prop:more-than-norm}, and using $d(\ket{\psi},\ket{\varphi}) = \norm{\ketbra{\psi}{\psi}-\ketbra{\varphi}{\varphi}}_d$, we get that $W_1^d(\rho,\sigma) \geq \norm{\rho-\sigma}_d$. This allows us to recover even more properties of $\norm{\cdot}_{DW_1^d}$:

\begin{prop} \label{prop:more-than-underlying-norm}
    Suppose that there $\mathcal{H}$ is finite-dimensional exists norm $\norm{\cdot}_d$ such that $d(\ket{\psi},\ket{\varphi})=\norm{\ketbra{\psi}{\psi}-\ketbra{\varphi}{\varphi}}_d$. Then:
    \begin{enumerate}
        \item $W_1^d(\rho,\sigma) \geq \norm{\rho-\sigma}_{DW_1^d} \geq \norm{\rho-\sigma}_{d}.$
        \item When $\rho$ and $\sigma$ are pure, this becomes an equality. \label{propitem:equality}
        \item For all $1 \leq q \leq \infty$, we have $W_q^d = W_q^{DW_1^d}$.
    \end{enumerate}
\end{prop}
\begin{proof}
As discussed above, we know that $\norm{\rho-\sigma}_d \leq W_1^d(\rho,\sigma)$. Letting $\norm{\cdot}_D$ be the dual to $\norm{\cdot}_d$ on the space of traceless Hermitian operators, we have then that $\norm{O}_D \geq L_d(O)$. Dualising once more, as in finite dimension the corresponding Banach space is reflexive, we know for all $X$ that $\norm{X}_d \leq \norm{X}_{DW_1^d}$. Combining this with Proposition \ref{prop:dual_basic_properties} gives the chain
        \begin{equation}\label{eq:dual_inequalities_chain}
        W_1^d(\rho,\sigma) \geq \norm{\rho-\sigma}_{DW_1^d} \geq \norm{\rho-\sigma}_{d}.
        \end{equation}
If $\rho, \sigma$ are pure then by Proposition \ref{prop:agree-on-pure-states} we have $W_1^d(\rho,\sigma) = \norm{\rho-\sigma}_d$, and so the whole chain \eqref{eq:dual_inequalities_chain} is an equality.
It follows that for $\ket{\psi}, \ket{\varphi}$, we have
        \begin{equation}
            d(\ket{\psi},\ket{\varphi}) = \norm{\ketbra{\psi}{\psi}-\ketbra{\varphi}{\varphi}}_{DW_1^d}.
        \end{equation}
        This means in turn that for all $1 \leq q \leq \infty$, we have $W_q^d = W_q^{DW_1^d}$.
\end{proof}
We restricted to the finite-dimensional case here to ensure that reflexivity of the Banach space is automatic. In the infinite-dimensional separable case, this result also holds when the Banach space corresponding to $\norm{\cdot}_d$ is reflexive, though this is no longer guaranteed. The first statement of Proposition \ref{prop:more-than-underlying-norm} tells us that $\norm{\cdot}_{DW_1^d}$ is maximal among norms that agree with $\norm{\cdot}_d$ on distances between pure states, and the third statement tells us that we can, without changing $W_1^d$, assume that the underlying norm $\norm{\cdot}_d$ is just $\norm{\cdot}_{DW_1^d}$.

We present this norm here, not as a replacement for $W_p^d$, but because the two complement each other nicely. This method of defining $\norm{\cdot}_{DW_1^d}$ via a supremum of a trace of the difference of two states over Lipschitz observables mirrors exactly the Kantorovich-Rubinstein theorem from classical optimal transport~\cite[p. 34]{villani2003topics}. In taking the double dual, we gain an easy proof of the triangle inequality, albeit at the expense of flexibility of order $p$ and of natural interpretation in terms of transport plans and couplings. We have seen here that $W_p^d$ and $\norm{\cdot}_{DW_1^d}$ are closely related, particularly when $d$ is induced by a norm. We will see later in Sec.~\ref{section:WpH} and Sec.~\ref{section:Wp1} that in many cases $\norm{\cdot}_{DW_1^d}$ essentially recovers the original norm $\norm{\cdot}_d$.

\section{Special instances}\label{sec:special_instances}

\subsection{$W_1^H$ distance on $n$-qudit systems} \label{section:WpH}

\cite{de2020quantum} introduced a quantum Wasserstein distance of order 1 which generalises the Hamming distance on the discrete hypercube. 
This is defined on Hilbert space $\mathcal{H} = \left(\mathbb{C}^d\right)^{\otimes n}$. The distance defined is normed, and we notate it here by $\norm{\rho-\sigma}_{W_1^H}$.

$\norm{\cdot}_{W_1^H}$ has the interesting property that it recovers the classical first-order Wasserstein distance $\mathcal{W}_1^H$ on the Hamming cube, for states that are diagonal in the computational basis. That is, for $r$ and $s$ probability distributions on $\{0,1,\dots, d-1\}^n$, and states $\rho = \sum_{x \in \{0,1,\dots, d-1\}^n} r(x)\ketbra{x}{x}$ and $\sigma = \sum_{y \in \{0,1,\dots,d-1\}^n} s(y) \ketbra{y}{y}$ we have
\begin{equation} \label{eq:w1H_norm_generalises_hamming}
\norm{\rho-\sigma}_{W_1^H} = \mathcal{W}_1^H(r,s).    
\end{equation}

Furthermore, its formulation mirrors the Kantorovich-Rubinstein theorem with the definition of the quantum Lipschitz constant \cite[Definition 8]{de2020quantum}:
\begin{equation} \label{eq:defn_quantum_lipschitz_w1H}
    \norm{O}_L = \max\{ \text{Tr}[O(\rho-\sigma)] : \norm{\rho-\sigma}_{W_1^H} \leq 1\}.
\end{equation}
It then induces a metric $d_H$ on $\mathbb{P}\mathcal{H}$, given by $d(\ket{\psi},\ket{\varphi}) = \norm{\ketbra{\psi}{\psi}-\ketbra{\varphi}{\varphi}}_{W_1^H}$, from which we can define a $p^{\text{th}}$ order Wasserstein distance $W_p^H$ as above. For the special case $p=1$, we will write $W_{p=1}^H$ to avoid confusion.

We might expect that property in equation \eqref{eq:w1H_norm_generalises_hamming} extends to the $W_p^H$ distance: that is, that for classical states $\rho$ and $\sigma$ defined by $r$ and $s$ as above, that
\begin{equation} \label{eq:clas_quant_equal}
W_p^H(\rho,\sigma) = \mathcal{W}_p^H(r,s).    
\end{equation}
However, in this case, we recover some interesting differences between the classical and quantum definitions of the $p$-Wasserstein distances, by taking `quantum shortcuts' in our transport plans.

For $p=1$, equality occurs in \eqref{eq:clas_quant_equal}: an optimal classical coupling $\gamma$ between $r$ and $s$ gives us the transport plan
$    \{ (\gamma(x,y), \ket{x}, \ket{y} \}_{x,y \in \{0,1,\dots, d-1\}^n}$
which has transport cost $\mathcal{W}_1^H(r,s)$, and we know that $W_{p=1}^H(\rho,\sigma) \geq \norm{\rho-\sigma}_{W_1^H} = \mathcal{W}_1^H(r,s)$ from Proposition \ref{prop:more-than-underlying-norm} above and Proposition 6 of \cite{de2020quantum}.

However, for $p > 1$, we find an interesting phenomenon of quantum `shortcuts' between classical states. Indeed, consider the case $n = d = 2$ and the states $\rho = \ketbra{00}{00}$, $\sigma = \frac{\mathbb{I}^{\otimes 2}}{4}$. Classically, the distributions $\delta_{00}$ and the uniform distribution have a $\mathcal{W}_{\infty}^H$ distance of $2$. However, consider the quantum transport plan
\begin{equation} \label{eq:quantum_shortcut_plan}
    \left\{\left( \frac{1}{4}, \ket{00}, \ket{\varphi^{\pm \pm}} \right)\right\}
\end{equation}
for the standard 2-qubit Bell states 
\begin{equation}
\begin{split}
\ket{\varphi^{++}} = \frac{1}{\sqrt{2}}(\ket{00}+\ket{11}) \qquad 
\ket{\varphi^{+-}} = \frac{1}{\sqrt{2}}(\ket{00}-\ket{11}) \\
\ket{\varphi^{-+}} = \frac{1}{\sqrt{2}}(\ket{01}+\ket{10}) \qquad 
\ket{\varphi^{--}} = \frac{1}{\sqrt{2}}(\ket{01}-\ket{10}).
\end{split}
\end{equation}
This corresponds to the separable coupling
\begin{equation}
    \tau = \sum_{\pm \pm}\frac{1}{4} \ketbra{00}{00} \otimes \ketbra{\varphi^{\pm \pm}}{\varphi^{\pm \pm}}
\end{equation}
of $\rho$ and $\sigma$.
Proposition 2 of \cite{de2020quantum} shows that both 
\begin{equation}
    \norm{\ketbra{00}{00} - \ketbra{\varphi^{+ \pm}}{\varphi^{+\pm}}}_{W_1^H} \leq \sqrt{2} \qquad \text{and} \qquad \norm{\ketbra{00}{00} - \ketbra{\varphi^{- \pm}}{\varphi^{-\pm}}}_{W_1^H} \leq \frac{\sqrt{5}+3}{4}.
\end{equation}
This in turn gives us $W_{\infty}^H(\rho,\sigma) \leq \max \left\{\sqrt{2},\frac{\sqrt{5}+3}{4}\right\} = \sqrt{2}$. While these quantum shortcuts are intriguing, it is as yet unclear the full extent to which they can appear and how they behave for different distances $d$. 

From Proposition \ref{prop:Wp_inherits_symmetries} we see that $W_p^H$ does inherit the invariance properties of $\norm{\cdot}_{W_1^H}$, namely invariance under local unitaries and qubit swaps. We can also relate the double-dual norm $\norm{\cdot}_{DW_1^H}$ to the norm $\norm{\cdot}_{W_1^H}$ as follows.
\begin{prop} \label{prop:Hamming_norms_equiv} Let $\rho, \sigma$ be states on $\mathcal{H} = \left(\mathbb{C}^d\right)^{\otimes n}$. Then
\begin{equation} \label{eq:Hamming_norms_equiv}
      2\norm{\rho-\sigma}_{W_1^H} \geq \norm{\rho-\sigma}_{DW_1^H} \geq \norm{\rho-\sigma}_{W_1^H}.
\end{equation}
\end{prop}
\begin{proof}
    The lower bound comes from equation \eqref{eq:dual_inequalities_chain}. For the upper bound, we refer to the dual $\norm{O}_L$ and show that $\norm{O}_L \leq 2L_d(O)$.
    Taking the dual of this equation will give the upper bound required.

    We have from \cite[Proposition 15]{de2020quantum}, that
    \begin{equation}
        2 \max_{1 \leq i \leq n} \norm{O - \text{Tr}_i O \otimes \mathbb{I}_i}_{\infty} \geq \norm{O}_L.
    \end{equation}
    So we will show that $L_d(O) \geq \max_{1 \leq i \leq n} \norm{O - \text{Tr}_i O \otimes \mathbb{I}_i}_{\infty}$.

    Fix $i$, and let $\ket{\psi}$ be an eigenvector of $O - \text{Tr}_i O \otimes \mathbb{I}_i$ with an eigenvalue which is maximal in absolute value, assuming without loss of generality that this is positive. Let $\mu_{\text{Haar}}$ be the Haar measure on the unitary group $\mathcal{U}(d)$ acting on the $i^{\text{th}}$ qudit. Writing $U_i$ for a unitary on the $i^{\text{th}}$ qubit, $\mathbb{I}_i$ for the identity on the $i^{\text{th}}$ qubit, and $\mathbb{I}_{\hat{i}}$ for the identity acting on all except the $i^{\text{th}}$ qubit, we note that 
    \begin{equation}
        \mathbb{E}_{U_i \sim \mu_{\text{Haar}}}\text{Tr}[(O - \text{Tr}_i O \otimes \mathbb{I}_i)U_i \otimes \mathbb{I}_{\hat{i}} \ketbra{\psi}{\psi}U^{\dagger} \otimes \mathbb{I}_{\hat{i}}] = \text{Tr}[(O - \text{Tr}_i O \otimes \mathbb{I}_i) \text{Tr}_i \ketbra{\psi}{\psi} \otimes \mathbb{I}_i /d] = 0.
    \end{equation}
    Therefore taking any $U_i \in \mathcal{U}(d)$ such that $U_i\ket{\psi}$ has $\text{Tr}[(O - \text{Tr}_i O \otimes \mathbb{I}_i)U_i \otimes \mathbb{I}_{\hat{i}}\ketbra{\psi}{\psi}U_i^{\dagger}\otimes \mathbb{I}_{\hat{i}}] < 0$, let this $U_i\otimes \mathbb{I}_{\hat{i}}\ket{\psi}$ be $\ket{\varphi}$. Note that $\text{Tr}_i[\ketbra{\psi}{\psi}] = \text{Tr}_i[\ketbra{\varphi}{\varphi}]$ and therefore $\norm{\ketbra{\psi}{\psi} - \ketbra{\varphi}{\varphi}}_{W_1^H} \leq 1$. This gives
    \begin{align}
        \frac{\text{Tr}[O(\ketbra{\psi}{\psi}-\ketbra{\varphi}{\varphi})]}{\norm{\ketbra{\psi}{\psi} - \ketbra{\varphi}{\varphi}}_{W_1^H}} 
            &\geq \frac{\text{Tr}[O(\ketbra{\psi}{\psi}-\ketbra{\varphi}{\varphi})]}{1} \\
            &= \text{Tr}[O(\ketbra{\psi}{\psi}-\ketbra{\varphi}{\varphi})] + \text{Tr}[\text{Tr}_i O (\text{Tr}_i \ketbra{\psi}{\psi} - \text{Tr}_i \ketbra{\varphi}{\varphi}])] \\
            &= \text{Tr}[(O - \text{Tr}_i O \otimes \mathbb{I}_i)(\ketbra{\psi}{\psi}-\ketbra{\varphi}{\varphi})] \\
            & \geq \norm{O - \text{Tr}_i O \otimes \mathbb{I}_i}_{\infty}
    \end{align}
    which concludes the proof.
\end{proof}

For this specific instance, we have a natural way to look at $p^{\text{th}}$-order quantum Wasserstein distances under tensor products. Indeed, for qudit systems $\mathcal{H}_a$ and $\mathcal{H}_b$ each equipped with $\norm{\cdot}_{W_1^H}$ on $n_a$ and $n_b$ qudits respectively, $\mathcal{H}_a \otimes \mathcal{H}_b$ can be equipped with $\norm{\cdot}_{W_1^H}$ on $n_a + n_b$ qudits. We can therefore apply the general result \ref{prop:subadditive} to $W_p^H$ as $\norm{\cdot}_{W_1^H}$ is additive under tensor products \cite[Proposition 4]{de2020quantum}.

\subsection{Trace distance} \label{section:Wp1}
In this section, we consider the case where $d$ is induced by the trace distance (that is, half the trace norm). For pure states $\ket{\psi}, \ket{\varphi}$, we have
\begin{equation}
    d(\ket{\psi},\ket{\varphi}) = \frac{1}{2}\norm{\ketbra{\psi}{\psi} - \ketbra{\varphi}{\varphi}}_1 = \sqrt{1-\left\lvert \braket{\psi | \varphi}\right\rvert^2}.
\end{equation}

The trace distance, when applied to pure states on a $D$-dimensional space, is a direct analogue of the discrete metric on a space with $D$ elements. This means that it is the simplest case to consider in terms of optimal transport. The trace distance on mixed states can also be considered as a quantum generalisation of the total variation distance $d_{TV}(\mu,\nu) = \sum_{i=1}^D |\mu(i)-\nu(i)|$, which is the classical $\mathcal{W}_1$ distance on this discrete space.

We write the $p^{\text{th}}$-order transport cost associated with this distance as $T_p^1$ and the associated quantum Wasserstein distance as $W_p^1$. We note from Proposition 2 of \cite{DePalma2021} that this agrees with $W_p^H$ in the case of a single qudit. From this we inherit all the properties of $W_p^H$, notably Proposition \ref{prop:Hamming_norms_equiv} which gives us $\norm{\rho-\sigma}_1 \geq \norm{\rho-\sigma}_{DW_1^1} \geq \frac{1}{2}\norm{\rho-\sigma}_1$. However, we can go one step further.

\begin{prop} \label{prop:Trace_norms_equivalent}
    Let $\rho, \sigma$ be states on $\mathbb{C}^d$. Then
    \begin{equation} \label{eq:Trace_norms_equivalent}
        \norm{\rho - \sigma}_{DW_1^1} = \frac{1}{2} \norm{\rho-\sigma}_{1}.
    \end{equation}
\end{prop}
\begin{proof}
    That $\norm{\rho - \sigma}_{DW_1^1} \geq \frac{1}{2} \norm{\rho-\sigma}_{1}$ is a direct consequence of Proposition \ref{prop:more-than-underlying-norm}. For the other direction, we once again consider the dual. The double dual of the $1$-norm is itself, and so to show inequality in the other direction we need only show for any traceless Hermitian operator $O$ that 
    \begin{equation}
        \sup_{\rho \neq \sigma} \frac{\text{Tr}[O(\rho-\sigma)]}{\frac{1}{2}\norm{\rho-\sigma}_1} \leq L_1(O)
    \end{equation}
    as taking the dual of this equation would give the inequality in the other direction.

    For any $\rho \neq \sigma$, let $\rho - \sigma = \rho' - \sigma'$ for $\rho', \sigma'$ positive semidefinite operators with $\rho' \perp \sigma'$. Then write $\rho' = \sum_{j \in J} \mu_j \ketbra{\psi_j}{\psi_j}$ and $\sigma' = \sum_{k\in K} \nu_k \ketbra{\varphi_k}{\varphi_k}$ in spectral decompositions. Let $\gamma$ be a classical coupling of the measures $\mu$ and $\nu$. It then follows that
    \begin{equation}
    \frac{1}{2} \norm{\rho-\sigma}_1 = \frac{1}{2}\sum_{j \in J, k \in K} \gamma_{j,k} \norm{\ketbra{\psi_j}{\psi_j} - \ketbra{\varphi_k}{\varphi_k}}_1
    \end{equation}
    and so
    \begin{equation}
        \frac{\text{Tr}[O(\rho-\sigma)]}{\frac{1}{2}\norm{\rho-\sigma}_1} = \frac{\sum_{j \in J,k \in K} \gamma_{j,k}\text{Tr}[O(\ketbra{\psi_j}{\psi_j} - \ketbra{\varphi_k}{\varphi_k})]}{\sum_{j \in J,k \in K}\gamma_{j,k} \norm{\ketbra{\psi_j}{\psi_j} - \ketbra{\varphi_k}{\varphi_k}}_1} \leq \sup_{\ket{\psi} \neq \ket{\varphi}} \frac{ \text{Tr}[O(\ketbra{\psi}{\psi} - \ketbra{\varphi}{\varphi})]}{ \norm{\ketbra{\psi}{\psi} - \ketbra{\varphi}{\varphi}}_1} = L_1(O)
    \end{equation}
    which concludes the proof.
\end{proof}

We can also now give an example to show that joint convexity does not hold in the case $p > 1$, analogously to the classical case. Using the notation of Proposition \ref{prop:joint_convexity_p=1} consider $\rho_1 = \ketbra{0}{0}$, $\sigma_1 = \ketbra{0}{0}$, $\rho_2 = \ketbra{0}{0}$, and $\sigma_2 = \ketbra{1}{1}$, with $r_1 = r_2 = \frac{1}{2}$. In this case $r_1 W_p^1(\rho_1,\sigma_1) + r_2W_p^1(\rho_2,\sigma_2) = \frac{1}{2}$ no matter the value of $p$, and $r_1\rho_1 + r_2 \rho_2 = \ketbra{0}{0}$, $r_1 \sigma_1 + r_2 \sigma_2 = \frac{\mathbb{I}}{2}$. By the proof of Proposition \ref{prop:more-than-underlying-norm} we have that for any transport plan $Q = \{ (q_j, \ket{0}, \ket{\varphi_j} )\}_{j \in J}$ between $\ketbra{0}{0}$ and $\frac{\mathbb{I}}{2}$ we have
\begin{equation}
    T_p^1 (Q) ^{1/p} \geq T_1^1(Q) \geq \norm{\ketbra{0}{0} - \frac{\mathbb{I}}{2}}_1 = \frac{1}{2}.
\end{equation}

Equality in the first comparison happens if and only $d(\ket{0},\ket{\varphi_j})$ is constant in $j$ since $x \mapsto x^p$ is strictly convex. This means that the $\ket{\varphi_j}$ must all take form $\alpha \ket{0} + \sqrt{1-|\alpha|^2} e^{i \theta_j} \ket{1}$ for some fixed $\alpha$. We must have $\alpha = \frac{1}{\sqrt{2}}$ to ensure the end state is indeed $\frac{\mathbb{I}}{2}$. Equality in the second comparison then happens only if all $\ketbra{0}{0} - \ketbra{\varphi_j}{\varphi_j}$ commute, meaning all phases $e^{i \theta_j}$ must be the same. But then the end state cannot be $\frac{\mathbb{I}}{2}$, therefore a plan $Q$ between $\ketbra{0}{0}$ and $\frac{\mathbb{I}}{2}$ satisfying equality here cannot exist. We conclude that joint convexity does not hold for $p > 1$.

\subsection{Complexity geometry} \label{section:WpC}

A distance $d(I, U)$ giving lower bounds for the complexity of synthesising a unitary $U \in \mathcal{SU}(2^n)$ from a universal one- and two-qubit gate set was defined in~\cite{Nielsen2006} and refined in~\cite{Dowling2008, Nielsen_2006_second}. This distance $d$ on $\mathcal{SU}(2^n)$ is a geodesic distance of a Riemannian manifold, where the Riemannian metric is chosen such that local travel is fast in directions corresponding to multiplication by low-weight unitaries, and slow in directions corresponding to multiplication by high-weight unitaries. 

This idea of expressing quantum gate complexity in terms of Riemannian geometry has seen renewed interest in recent years, from applications in black hole thermodynamics~\cite{brown2016blackholes, heller2023complexreview} to rigid bodies~\cite{brown2019complexitygeom} and the complexity of typical unitaries~\cite{brown2023lowerbounds}. However, this metric is originally defined as a distance between unitaries, with the complexity of unitary expressed in terms of its distance from the identity. This extends naturally to distances between pure states as the lowest complexity of a unitary which transforms one into the other. Our optimal transport formulation allows for a natural extension of this metric to mixed states, in a way that can be considered as quantifying the lowest possible complexity of transforming one mixed state into another. A related approach in ~\cite{ruan2021thesis} extends a variation of this complexity geometry metric, one in which multiplication by a unitary is independent of weight, to mixed states using purification methods described in~\cite{agon2019subsystem}. Another approach in \cite{li2022wasserstein} instead studies complexity using the $\norm{\cdot}_{W_1^H}$ norm.

Formally, \cite{Nielsen2006} defines a right-invariant Riemannian metric $g$ on $\mathcal{SU}(2^n)$ given by the following inner product on $T_I\mathcal{SU}(2^n)$. A vector $X$ in the tangent space $T_U \mathcal{SU}(2^n)$ to $\mathcal{SU}(2^n)$ at $U$ is identified with its Hamiltonian representation $H$, given by $X =\left[ \frac{\text{d}}{\text{d}t} e^{-iHt}U \right]_{t=0}$. $H$ can be decomposed as $H_P + H_Q$, where $H_P$ is a linear combination of Pauli matrices of weight at most $2$, and $H_Q$ is a linear combination of Pauli matrices of weight at least $3$. 
Defining operations $\mathcal{P}$ and $\mathcal{Q}$ by $\mathcal{P}(H) = H_P$ and $\mathcal{Q}(H) = H_Q$, the inner product on $T_U \mathcal{SU}(2^n)$ is then defined by
\begin{equation}
    \langle H, J \rangle = \frac{\text{Tr}(H\mathcal{P}(J)) + q\text{Tr}(H\mathcal{Q}(J))}{2^n}
\end{equation}
where $q > 4^n$ is a penalty parameter. The length $\sqrt{\langle H, H \rangle }$ will then be of order $1$ for 2-local Hamiltonians, and of order $q$ otherwise.
$q > 4^n$ is chosen such that for any geodesic, the associated Hamiltonian path has an approximately low-weight decomposition.

The length of a curve $U(t)$ defined by evolution $\dot{U}(t) = -iH(t)U(t)$ is then defined as standard by
\begin{equation}
    \int \sqrt{\langle H(t), H(t) \rangle} \text{d}t
\end{equation}
and the distance $d(I, U)$ is simply the geodesic distance on this Riemannian manifold.

The main purpose of $d$ is to find a geometric interpretation of gate complexity, and~\cite[Equation 3]{Dowling2008} gives bounds for $d$ in terms of gate complexity. Let $G(U)$ be the exact gate complexity of $U$, i.e. the minimal number of one- and two-qubit gates required to synthesise $U$ exactly. For $\epsilon\geq0$, let $G(U,\epsilon)$ be the minimal number of one- and two-qubit gates required to synthesise a gate $V$ such that $\norm{U-V}_{\infty} \leq \epsilon$, also known as the $\epsilon$-approximate gate complexity. We then have bounds

\begin{equation} \label{eq:complexity_distance_approxcompl_bounds}
    \frac{\kappa G(U,\epsilon)^{1/3}\epsilon^{2/3}}{n^{2}} \leq d(I,U) \leq G(U)
\end{equation}
for some constant $\kappa > 0$.

With these unitaries acting on the $n$-qubit space $\mathcal{H} = (\mathbb{C}^2)^{\otimes n}$, we can define a metric $d_C$ on $\mathbb{P}\mathcal{H}$ by
\begin{equation}
    d_C(\ket{\psi},\ket{\varphi}) = \min\left\{d(I,U): U \in \mathcal{SU}(2^n), U\ket{\psi} = \ket{\varphi} \right\}.
\end{equation}
 The metric properties of $d_C$ come directly from the metric properties and right-invariance of $d$. Operationally, this gives a lower bound for the minimum circuit complexity required to synthesise $\ket{\varphi}$ from $\ket{\psi}$. Indeed, 
 \begin{equation}\label{eq:complexity_distance_approxcompl_bounds_for_states}
         \frac{\kappa (\min \{G(U,\epsilon) : U\ket{\psi} = \ket{\varphi}\})^{1/3}\epsilon^{2/3}}{n^{2}} \leq d_C(\ket{\psi},\ket{\varphi}) \leq \min \{G(U) : U\ket{\psi} = \ket{\varphi} \}.
 \end{equation}

In order to say anything useful about the quantum Wasserstein distance induced by this metric $d_C$, we must demonstrate some basic continuity properties. Continuity of $d_C$ is inherent from the fact that $d$ is a geodesic distance on a Riemannian manifold and Hölder continuity of the 2-norm with respect to $d_C$ comes from the following proposition.
\begin{prop} \label{prop:unitary_complexity_to_2norm}
    Let $U \in \mathcal{SU}(2^n)$. Then
    \begin{equation} \label{eq:unitary_complexity_to_2norm}
        d(I,U) \geq 2^{-n/2} \norm{U-I}_2 .
    \end{equation}
\end{prop}
\begin{proof}
    Suppose $U(t)$ is a geodesic curve in $\mathcal{SU}(2^n)$ with $U(t) = I$ and $U(T) = U$, generated by Hamiltonian $H(t)$ satisfying $\dot{U}(t) = -iH(t)U(t)$. By a family of smooth curves uniformly approximating $U$ we may assume that both $U$ and $H$ are smooth, with first-order Taylor expansion $U(t+h) = U(t) - iH(t)U(t)h + \delta(t,h)h$ for some $\delta$ such that for all $t$, $\delta(t,h) \to 0$ with $h$. By smoothly extending $U$ we may assume that $\delta$ is defined on some $[ 0, T ] \times (-\epsilon, \epsilon)$ and as $[0,T]$ is compact this convergence to $0$ is uniform in $t$. Then for any $N$, we have
    \begin{align}
        \norm{U-I}_2 
            &\leq \sum_{j=0}^{N-1} \norm{U\left((j+1)T/N\right) - U\left(jT/N\right)}_2 \\
            &\leq \sum_{j=0}^{N-1} \norm{i(T/N)H\left(jT/N\right)U\left(jT/N\right) + (T/N)\delta(jT/N, T/N)}_2  \\
            &\leq \sum_{j=0}^{N-1}\left( \frac{T}{N}\norm{iH\left(jT/N\right)U(jT/N)}_2 + \frac{T}{N}\norm{\delta(jT/N, T/N)}_2 \right) \qquad \text{by the triangle inequality}\\
            &\leq \left( \sum_{j=0}^{N-1}\frac{T}{N}\norm{H\left(jT/N\right)}_2 \right)+ \sup_t \norm{\delta(t,T/N)}_2 \qquad \ \text{by unitary invariance of $\norm{\cdot}_2$} \\
            &\to \int_0^T \norm{H(t)}_2 \text{d}t \qquad \text{ as } N \to \infty \text{ by smoothness of $H$} \\
            &\leq 2^{n/2} \int_0^T \sqrt{\langle H(t),H(t) \rangle} \text{d}t
    \end{align}
    
using the inequality $2^{-n/2}\norm{H}_2 \leq \sqrt{\langle H,H \rangle}$ in the last line. Taking the infimum over such curves gives $d(I,U) \geq 2^{-n/2} \norm{U-I}_2 $.
\end{proof}

\begin{cor} \label{prop:state_complexity_to_2norm}
    Let $\ket{\psi},\ket{\varphi} \in \mathbb{P}\mathcal{H}$. Then
    \begin{equation} \label{eq:state_complexity_to_2norm}
        2^{(n+1)/2}d_C(\ket{\psi},\ket{\varphi}) \geq \norm{\ketbra{\psi}{\psi}-\ketbra{\varphi}{\varphi}}_2.
    \end{equation}
\end{cor}
\begin{proof}
    Let $U \in \mathcal{SU}(2^n)$ have $U\ket{\psi} = \ket{\varphi}$. Then
    \begin{align}
        \norm{U-I}_2^2 
            &= \text{Tr}[(U-I)^\dagger (U-I)] \\
            &\geq \braket{\psi | (U-I)^\dagger(U-I) | \psi} \\
            &\geq 2(1 - |\braket{\psi | \varphi}|).
    \end{align}

Using the formula $\norm{\ketbra{\psi}{\psi}-\ketbra{\varphi}{\varphi}}_2^2 = 2(1-|\braket{\psi | \varphi}|)^2$, we get
\begin{align}
    \norm{U-I}_2^2 
        &\geq 2\left(1-\sqrt{1-\frac{1}{2}\norm{\ketbra{\psi}{\psi}-\ketbra{\varphi}{\varphi}}_2^2} \right) \\
        &\geq \frac{1}{2}\norm{\ketbra{\psi}{\psi}-\ketbra{\varphi}{\varphi}}_2^2.
\end{align}
Combining this with Proposition~\ref{prop:unitary_complexity_to_2norm} gives the result.
\end{proof}

It follows from Proposition \ref{prop:more-than-norm} that all $W_p^C$ defined from this $d_C$ are nondegenerate. This gives a natural extension of the ideas of state complexity to mixed states, to which we can apply results such as those which will be discussed in Section \ref{section:operational} on classical-quantum (cq) sources.

Looking at $\rho$ and $\sigma$ in their eigenbases, we can give a concrete interpretation of these values. Indeed, let
\begin{equation}
\rho = \sum_{b \in \{0,1\}^n} r_b \ketbra{\psi_b}{\psi_b} \qquad \sigma = \sum_{c \in \{0,1\}^n} s_c \ketbra{\varphi_c}{\varphi_c}
\end{equation}
for $r$ and $s$ some classical probability distributions on $\{0,1\}^n$. Letting $U$ be any unitary such that all $U\ket{\psi_b} = \ket{b}$, $V$ any unitary such that all $V\ket{\varphi_c} = \ket{c}$, and $q$ be an optimal $p^{\text{th}}$-order classical coupling of the distributions $r$ and $s$, we can consider transport plan
\begin{equation}
    Q = \{(q_{bc}, \ket{\psi_b}, \ket{\varphi_c})\}_{b,c \in \{0,1\}^n}
\end{equation}
which has a first-order transport cost of at most $G(U) + G(V) + \mathcal{W}_1^H(r,s)$. This allows us to conclude that
\begin{equation}
    W_1^C(\rho,\sigma) \leq G(U)+G(V)+\mathcal{W}_1^H(r,s).
\end{equation}
This splits the quantum transport distance arising from complexity geometry into two parts: the classical transport cost between the states, and their quantum complexity as a whole.

For this specific instance we also have subadditivity as discussed in Proposition \ref{prop:subadditive}, as there is a natural choice for the distance $d_C$ on the tensor product. Namely, for two systems $\mathcal{H}_a$, $\mathcal{H}_b$ on $n_a$, $n_b$ qubits respectively, we can equip $\mathbb{P}(\mathcal{H}_a \otimes \mathcal{H}_b)$ with the distance $d_C$ on $n_a + n_b$ qubits. For states $\ket{\psi_a}, \ket{\varphi_a}$ on $\mathcal{H}_a$ and $\ket{\psi_b}, \ket{\varphi_b}$ on $\mathcal{H}_b$, let $U_a\ket{\psi_a} = \ket{\varphi_a}$, and $U_b\ket{\psi_b} = \ket{\varphi_b}$.
We have
\begin{align}
    d(I_a \otimes I_b, U_a \otimes U_b)
        &\leq d(I_a \otimes I_b, U_a \otimes I_b) + d(U_a \otimes I_b, U_a \otimes U_b) \\
        &= d(I_a \otimes I_b, U_a \otimes I_b) + d(I_a \otimes I_b, I_a \otimes U_b) \\
        &\leq d(I_a , U_a) + d(I_b , U_b)
\end{align}
and therefore $d_C(\ket{\psi_a}\otimes \ket{\psi_b}, \ket{\varphi_a}\otimes \ket{\varphi_b}) \leq d_C(\ket{\psi_a}, \ket{\varphi_a}) + d_C(\ket{\psi_b}, \ket{\varphi_b})$. The condition of Proposition \ref{prop:subadditive} is satisfied, and so under the natural choice $d_C$ on $\mathbb{P}(\mathcal{H}_a \otimes \mathcal{H}_b)$ the $W_p^d$ distances are subadditive.

\section{Applications}
\subsection{Results for random quantum states} \label{section:random_states}

To understand these quantities in general, it is useful to look at how they behave on random quantum states. We will look at both the versions stemming from $\norm{\cdot}_{W_1^H}$ and from complexity geometry.

We look at a few regimes in the definition of `random' states. For random pure states, we generate $\ket{\psi}$ according to the uniform measure on the unit sphere in $\mathcal{H}$, and take $\rho = \ketbra{\psi}{\psi}$. For mixed states, we adjoin an auxiliary system $\mathcal{A}$ of dimension $s$, generate a random pure state $\ket{\psi}$ on $\mathcal{H}\otimes \mathcal{A}$, then take $\rho = \text{Tr}_{\mathcal{A}} \ketbra{\psi}{\psi}$. The distribution of $\rho$ depends entirely on the values of $s$ and $\text{dim} \mathcal{H}$ chosen.

Note that for the two versions studied, the underlying space $\mathcal{H}$ is a qudit space $\mathcal{H} = (\mathbb{C}^d)^{\otimes n}$. It is convenient, in this case, to write $s = d^m$, as this allows us to consider the qudit ratio $c = \frac{m}{n}$ between the auxiliary and base systems. Note that while $s$ is always an integer, $m$ need not be, i.e. we will also consider auxiliary dimensions that are not a power of $d$. We can then consider the regime $c < 1$ as `low rank', and regime $c > 1$ as `high rank'. We will see that there is a marked phase transition between values $c < 1$, in which $W_1^H$ and $W_1^C$ grow on average like $\text{diam}(\mathbb{P}\mathcal{H})$ in $n$, and between values $c > c^*$ for some threshold $c^* \geq 1$ depending on $\mathcal{H}$, in which all $W_1^d$ decay exponentially with $n$ on average.

Previous works have highlighted similar properties for other quantites. For example,
\cite{PageConjecture}
shows that for the regime $c > 1$ we have $\mathbb{E}[S(\rho)] \geq n \log d - \frac{1}{2}d^{-(c-1)n}$, from which we can show that the quantum relative entropy $D(\rho \| \sigma) = \text{Tr}[(\rho \log \rho - \rho \log \sigma)]$ has $\mathbb{E}[D(\rho \| \sigma)]$ decaying exponentially in $n$ for $c > 1$. Meanwhile, two random states $\rho, \sigma$ with $c > 1$ have $\text{span } \sigma \not\subseteq \text{span } \rho$ with probability $1$, and so $D(\rho \| \sigma)$ is in general infinite. Similar results have been found for the trace distance~\cite{Telles_de_Miranda_2023}, showing that for $c < 1$ we have $\mathbb{E}\left[ \frac{1}{2}\norm{\rho-\sigma}_1\right] \to_{n \to \infty} 1$, and for $c > 1$ $\mathbb{E}\left[ \frac{1}{2} \norm{\rho-\sigma}_1 \right] \to_{n \to \infty} 0$. However, these existing analyses relate to the notion of distinguishability of states, whereas analysing the phase transition in $W_1^d$ allows us to go beyond this regime to discuss other properties such as the computational complexity between random states.

For mixed states in the high-rank regime, we see exponential decay in the expected $W_1^d$ distance between two i.i.d. states, no matter the underlying metric $d$. This is summarised in the following proposition. Note that while it is phrased for qudit systems, it applies for any $\mathcal{H}, \mathcal{A}$ where $c = \frac{\log \dim \mathcal{A}}{\log \dim \mathcal{H}}$ and $d^n = \dim \mathcal{H}$. We write $\text{diam}_d(\mathbb{P}\mathcal{H})$ for the diameter of $\mathbb{P}\mathcal{H}$ under metric $d$.

\begin{prop} \label{prop:random_mixed_largerank}
    Let $\mathcal{H} = (\mathbb{C}^d)^{\otimes n}$ with i.i.d. random mixed states $\rho, \sigma$ generated by a auxiliary system of $\mathcal{A}$ of dimension $s = d^m$. Let $c = \frac{m}{n}$, and suppose $c > 1$. Then for any $\beta > 0$, we have 
    \begin{equation} \label{eq:random_mixed_largerank}
        \mathbb{P}\left[W_1^d(\rho,\sigma) \geq \beta d^{-(c-3)n/2} \text{{\em diam}}(\mathbb{P}\mathcal{H})\right] \leq \frac{1}{\beta^2}.
    \end{equation}
\end{prop}
\begin{proof}
For random $\rho, \sigma$ generated from large auxiliary systems, we generally expect both to be close to maximally mixed. And so letting the minimum eigenvalue among $\rho$ and $\sigma$ be $\frac{1}{d^n} - \delta$, we can split up $\rho$ into parts $\rho - \left(\frac{1}{d^n}-\delta\right)\mathbb{I}$ and $\left(\frac{1}{d^n}-\delta\right)\mathbb{I}$, both of which are positive semidefinite. We can do the same for $\sigma$. We can then transport $\left(\frac{1}{d^n}-\delta\right)\mathbb{I}$ onto $\left(\frac{1}{d^n}-\delta\right)\mathbb{I}$ at zero cost, and transport $\rho - \left(\frac{1}{d^n}-\delta\right)\mathbb{I}$ onto $\sigma - \left(\frac{1}{d^n}-\delta\right)\mathbb{I}$ via any partial transport plan, at a maximum cost of $\text{Tr}\left[\rho - \left(\frac{1}{d^n}-\delta\right)\mathbb{I} \right]\text{diam}_d(\mathbb{P}\mathcal{H}) = \delta d^n \text{diam}_d(\mathbb{P}\mathcal{H})$. We will show that this is most likely very small.

Focusing on $\rho$, we know from \cite{PageConjecture} that
\begin{equation}
    \mathbb{E}[S(\rho)] \geq n \log d - \frac{1}{2}d^{-(c-1)n}.
\end{equation}
Knowing also that for the von Neumann entropy $S$ we have $S(\rho) \leq n \log d$, and using the top-down Markov inequality, we get that for any $\alpha > 0$ we have
\begin{equation}
    \mathbb{P}\left[S(\rho) \geq n \log d - \frac{\alpha}{2}d^{-(c-1)n}\right] \geq 1 - \frac{1}{\alpha}.
\end{equation}
Using then that $D \left( \rho || \frac{I}{d^n} \right) = n \log d - S(\rho)$, and the quantum Pinsker's inequality, we have
\begin{equation}
    \mathbb{P}\left[\norm{\rho-\frac{I}{d^n}}_1 \leq \sqrt{\frac{\alpha}{2}}d^{-(c-1)n/2} \right] \geq 1- \frac{1}{\alpha}
\end{equation}
and so
\begin{equation}
    \mathbb{P}\left[\norm{\rho-\frac{I}{d^n}}_{\infty} \leq \sqrt{\frac{\alpha}{2}}d^{-(c-1)n/2} \right] \geq 1- \frac{1}{\alpha}.
\end{equation}
Reintroducing $\sigma$, we then get that
\begin{equation}
    \mathbb{P}\left[\max \left\{\norm{\rho-\frac{I}{d^n}}_\infty,\norm{\sigma-\frac{I}{d^n}}_\infty\right\} \leq \sqrt{\frac{\alpha}{2}}d^{-(c-1)n/2} \right] \geq \left(1- \frac{1}{\alpha}\right)^2 \geq 1-\frac{2}{\alpha}.
\end{equation}
Note that for any $\rho,\sigma$ with $\max \left\{\norm{\rho-\frac{I}{d^n}}_\infty,\norm{\sigma-\frac{I}{d^n}}_\infty\right\} \leq \sqrt{\frac{\alpha}{2}}d^{-(c-1)n/2}$, we can take $\delta = \sqrt{\frac{\alpha}{2}}d^{-(c-1)n/2}$. Then setting $\beta = \sqrt{\alpha / 2}$, and applying the transport plan described above, we have
\begin{equation}
    \mathbb{P}\left[W_1^d(\rho,\sigma) \geq \sqrt{\frac{\alpha}{2}} d^{-(c-1)n/2} \cdot d^n \text{diam}_d(\mathbb{P}\mathcal{H}) \right] \leq \frac{2}{\alpha}
\end{equation}
and so
\begin{equation}
    \mathbb{P}\left[W_1^d(\rho,\sigma) \geq \beta d^{-(c-3)n/2} \text{diam}_d(\mathbb{P}\mathcal{H}) \right] \leq \frac{1}{\beta^2}.
\end{equation}

\end{proof}

This result can be easily applied to the two underlying distances mentioned above. For qudit ratios $c > 3 + \frac{2}{n} \log_d \text{diam}_d(\mathbb{P}\mathcal{H})$ we can define $\beta$ by, for any $\lambda \in (0,1)$, value 
\begin{equation}
    \log_d \beta = \lambda \left(c - 3 - \frac{2}{n} \log_d \text{diam}_d(\mathbb{P}\mathcal{H})\right)\frac{n}{2}
\end{equation}
to show that the probability of an exponentially small deviation falls exponentially as the number $n$ of qudits increases.

For the $W_1^H$ case, we know $\text{diam}_d(\mathbb{P}\mathcal{H}) = n$, and so this exponential decay applies for qubit ratio $c > 3$ and large enough $n$. Applying this to the expectation gives, taking $\lambda = 1/3$ for an optimal decay rate,
\begin{align}
    \mathbb{E}_{\rho,\sigma}\left[W_1^H(\rho,\sigma)\right]
        &\leq \left( 1 - \frac{1}{\beta^2}\right)\beta d^{-(c-3)n/2} \text{diam}_d(\mathbb{P}\mathcal{H}) + \frac{1}{\beta^2} \text{diam}_d(\mathbb{P}\mathcal{H}) \\
        &\leq n^{1-\lambda}\exp_d\left(-\frac{1}{2}(1-\lambda)(c-3)n \right) + \exp_d n^{1+2\lambda}\left(\lambda (c-3)n \right) \\
        &= \exp_d\left(-\frac{1}{3}(c-3)n + \frac{2}{3}\log_d n \right) + \exp_d \left(-\frac{1}{3}(c-3)n +\frac{5}{3}\log_d n\right) \\
        &= \left(n^{2/3}+n^{5/3}\right)\exp_d\left( -\frac{1}{3}(c-3)n \right)
\end{align}

For the $W_1^C$ case, any $n$-qudit gate can be synthesised in at most $2^n(2^n-1)$ one- and two-qubit gates~\cite{nielsen_chuang_2010}, and so the metric space $\mathbb{P}\mathcal{H}$ has diameter at most $2^{2n}$. It then follows that for any qudit ratio $c > 7$, the probability of an exponentially large deviation becomes exponentially small.

Applying equation \eqref{eq:random_mixed_largerank} to the expectation gives, taking $\lambda = \frac{c-3}{3(c-7)}$ for an optimal decay rate,
\begin{align}
    \mathbb{E}_{\rho,\sigma}\left[W_1^C(\rho,\sigma)\right]
        &\leq \left( 1 - \frac{1}{\beta^2}\right)\beta 2^{-(c-3)n/2} \text{diam}_d(\mathbb{P}\mathcal{H}) + \frac{1}{\beta^2} \text{diam}_d(\mathbb{P}\mathcal{H}) \\
        &\leq 2^{-(1-\lambda)(c-3)n/2} 2^{(1-\lambda)2n} + 2^{-\lambda(c-3)n}2^{(1+2\lambda)2n} \\
        &= 2^{-(1-\lambda)(c-7)n/2} + 2^{-(\lambda(c-7)-2)n} \\
        &= 2 \cdot 2^{(c - 9)n/3}
\end{align}
giving exponential decay in expectation for qubit ratios $c > 9$.

For the low-rank setting, we look first at the $W_1^H$ distance. There are two lines of intuition here. The first is that $W_1^H$ generalises the Hamming distance, and the Hamming distance quantifies the local distinguishability of $d$-nary strings. If this property propagated to the quantum setting, we'd expect $W_1^H$ to be small on average as random pure states are generally locally indistinguishable. This was the behaviour conjectured in~\cite{de2022limitations}. The second is that the average Hamming distance between two $d$-nary strings of length $n$ is $n(1-1/d)$, and so we might also expect the average $W_1^H$ distance between random pure strings to grow linearly with the number of qudits.

Let $\mathcal{H} = (\mathbb{C}^d)^{\otimes n}$ and let $m = \log_d s$, noting again that, while $s$ is an integer, $m$ need not be. In the case $m < n$, we can apply Theorem 9.1 of \cite{De_Palma_2023} to lower bound the expected distance between two low-rank random mixed states. A similar result was noted independently in \cite{depalma2023shadows}.

\begin{prop} \label{prop:random_hamming_lowrank}
    Let $\rho$, $\sigma$ be two i.i.d.\ random mixed states on $\mathcal{H} = (\mathbb{C}^d)^{\otimes n}$ generated using an auxiliary system of dimension $s = d^m$ for $m < n$. Write $c = \frac{m}{n}$. Then
    \begin{equation}
    \mathbb{E}_{\rho,\sigma}\left[W_{p=1}^H(\rho,\sigma)\right] \geq \lambda_c n
    \end{equation}
    where $\lambda_c$ satisfies $\left(1-c\right)\log d = h_2(\lambda) + \lambda \log(d^2-1)$ for $h_2$ the binary entropy.
\end{prop}
\begin{proof}
    First note from Proposition \ref{prop:more-than-underlying-norm} that $W_{p=1}^H(\rho,\sigma) \geq \norm{\rho-\sigma}_{W_1^H}$, and so we prove that $\mathbb{E}_{\rho,\sigma} \left[\norm{\rho-\sigma}_{W_1^H}\right] \geq \lambda n$.

    Fix $\rho$, and note that averaging over $\sigma$ and using convexity of the norm we have
    \begin{align}
        \mathbb{E}_{\sigma} \left[\norm{\rho-\sigma}_{W_1^H}\right]
            &\geq \norm{\rho - \mathbb{E}_{\sigma}\sigma}_{W_1^H} \\
            &= \norm{\rho-\frac{\mathbb{I}_d^{\otimes n}}{d^n}}_{W_1^H}. \label{prop:random_hamming_lowrank_use_convexity}
    \end{align}
Applying Theorem 9.1 of \cite{De_Palma_2023} then gives
\begin{equation} 
    \left(1-c\right)\log d \leq \frac{1}{n}\left|S(\rho)-S\left(\frac{\mathbb{I}_d^{\otimes n}}{d^n}\right)\right| \leq h_2\left(\frac{\norm{\rho - \mathbb{I}_d^{\otimes n}/d^n}_{W_1^H}}{n}\right) + \frac{\norm{\rho - \mathbb{I}_d^{\otimes n}/d^n}_{W_1^H}}{n}\log(d^2-1).
\end{equation}
Noting then that the function $g(t) = h_2(t) + t \log(d^2-1)$ takes the value $\left(1-c\right)\log d$ at exactly one value $\lambda \in [0,1]$, and that for for $t < \lambda$ we have $g(t) < g(\lambda)$ and for $t > \lambda$ we have $g(t) > g(\lambda)$, we conclude that $\norm{\rho - \mathbb{I}_d^{\otimes n}/d^n}_{W_1^H} > \lambda n$. Averaging over $\rho$ gives the result.
\end{proof}
For $\rho, \sigma$ random pure states, we simply take the case $m = 0$. In general, this shows that the expected $W_{p=1}^H$ distance between two random states generated using small auxiliary systems grows linearly with the number $n$ of qudits. This is in direct contrast to initial conjecture, and brings up interesting ideas about the nature of $\norm{\cdot}_{W_1^H}$. On one hand, this gives insight into the behaviour of the widely celebrated $\norm{\cdot}_{W_1^H}$ norm, showing that it behaves qualitatively very differently in the quantum setting to the classical setting. It also highlights the importance of entangled states to analysis of $\norm{\cdot}_{W_1^H}$. This in turn opens up question about the qualitative nature of Lipschitz operators according to its dual norm $\norm{\cdot}_L$. We have shown here that Lipschitz operators applied to $\rho-\sigma$ do not necessarily detect local distinguishability, so what do they detect? On the other hand, this calls into question the assumptions made in~\cite{Kiani2022} that the weighted sum of Pauli coefficients of a Hermitian operator $O$ is a good approximation to $\norm{O}_L$, as such an approximation links directly to local distinguishability in the dual setting. This will be discussed further in Section~\ref{section:GANs}.

This property can be understood more concretely by looking at the distance operator~\cite{Eldar_2017} of a state $\ket{\psi}$. The distance operator $\Lambda_{\ket{\psi}} = \sum_{i=1}^n i\Pi_{\mathcal{V}_i \cap \mathcal{V}_{i-1}^{\perp}}$ is defined as a weighted sum of projectors onto subspaces $\mathcal{V}_i \cap \mathcal{V}_{i-1}^\perp$, where
\begin{equation}
    \mathcal{V}_i = \text{span}\{ O \ket{\psi} : O \text{ acts on at most $i$ qudits}\}.
\end{equation}
This has been specifically constructed so that $\text{Tr}[\Lambda_{\ket{\psi}}\ketbra{\psi}{\psi}] = 0$, and that if $\ket{\varphi}$ can be written as a sum of states which differ from $\psi$ in at most $k$ qubits, then $\text{Tr}[\Lambda_{\ket{\psi}}\ketbra{\varphi}{\varphi}] \leq k$. It was proven in~\cite{De_Palma_2023} that $\norm{\Lambda_{\ket{\phi}}}_{L} \leq 1$, and analysis of the dimension of $\mathcal{V}_i$ shows that $\text{Tr}[\Lambda_{\ket{\psi}}] = \mathcal{O}(nd^n)$. It follows that, for fixed $\ket{\psi}$, $\mathbb{E}_{\ket{\varphi}\sim \mu_{\text{Haar}}}\left[\left|\text{Tr}[\Lambda_{\ket{\psi}}(\ketbra{\psi}{\psi} - \ketbra{\varphi}{\varphi})] \right|\right] = \mathcal{O}(n)$.

Turning our attention to the $W_1^C$ distance generated from complexity geometry on $\mathbb{P}\left((\mathbb{C}^2)^{\otimes n}\right)$, we see a similar picture for low-rank states.
As noted earlier, for the approximate gate complexity $G(U,\epsilon)$ and the gate complexity $G(U)$, we have the bound
\begin{equation} \label{eq:bound_dIU_by_complexities}
    \frac{\kappa G(U,\epsilon)^{1/3}\epsilon^{2/3}}{n^{2}} \leq d(I,U) \leq G(U)
\end{equation}
for some constant $\kappa > 0$.

For pure states we also know that $W_p^C(\ket{\psi}\bra{\psi},\ketbra{\varphi}{\varphi}) = d_C(\ket{\psi},\ket{\varphi})$. And so to show the behaviour of $W_p^C$ distances on low-rank states, we look at $d_C$.

\begin{lemma} \label{lemma:link_complexity_allgates_universalset}
    Let $\mathcal{H} = (\mathbb{C}^2)^{\otimes n}$ be an $n$-qudit space, and let $\mathcal{S} \subseteq{\mathcal{SU}(4)}$ be a finite universal gate set with inverses. Letting $G_{\mathcal{S}}(U,\epsilon)$ be the $\epsilon$-approximate gate complexity of $U$ from set $\mathcal{S}$ viewed as a set of gates on 2 qubits, and $G(U,\epsilon)$ the $\epsilon$-approximate gate complexity of $U$ using any one- or two-qubit gates, we have
    \begin{equation} \label{eq:link_complexity_allgates_universalset}
        G(U,\epsilon)\text{{\em poly}}(\log(G(U,\epsilon))+\log(\epsilon^{-1})) \geq G_{\mathcal{S}}(U,2\epsilon)
    \end{equation}
\end{lemma}
\begin{proof}
    Let $V_1,\dots, V_{G(U,\epsilon)}$ be any circuit of one- and two-qubit gates to synthesise $V$ such that $\norm{U-V}_{\infty} \leq \epsilon$. Using Solovay-Kitaev, each of these can be approximated to within error $\epsilon/G(U,\epsilon)$ in $\text{poly}(\log \left(G(U,\epsilon)/\epsilon\right))$ gates from $\mathcal{S}$. Compounding errors linearly, these form a circuit of length $G(U,\epsilon)\text{ poly}(\log(G(U,\epsilon))+\log(\epsilon^{-1}))$ of gates from $\mathcal{S}$ which synthesises $U$ to within operator norm $2\epsilon$.
\end{proof}

\begin{lemma} \label{lemma:approx_complexity_universalset_low_mixed}
    Let $\mathcal{S} \subseteq{\mathcal{SU}(4)}$ be a finite universal gate set with inverses and $\rho$, $\sigma$ i.i.d. random quantum states on $\mathcal{H} = \left(\mathbb{C}^2\right)^{\otimes n}$ generated by auxiliary system $\mathcal{A}$ of integer dimension $s = 2^{cn}$ where $0 \leq c < 1$. Let $G_{\mathcal{S}}(U_{\rho \to \sigma}^{\text{opt}},\epsilon) = \min \left\{G_{\mathcal{S}}\left(U_{\ket{\psi} \to \ket{\varphi}}^{\text{opt}},\epsilon\right) : \ket{\psi} \in \text{{\em span }}\rho, \ket{\varphi} \in \text{{\em span }} \sigma\right\}$. Then
    \begin{equation} \label{eq:lemma_approx_complexity_universalset_low_mixed}
        \mathbb{P}_{\rho,\sigma} \left[ G_{\mathcal{S}}(U^{\text{opt}}_{\rho \to \sigma}, \epsilon)  \leq 2^{(1-\delta) n} \right] \leq e^{-\Omega \left(2^n \log(1/\epsilon) \right)}.
    \end{equation}
\end{lemma}
\begin{proof}
Let $\mathcal{C}_x$ be the set of circuits of $\mathcal{S}$ of length $x$.
    \begin{align}
        G_{\mathcal{S}}(U^{\text{opt}}_{\rho \to \sigma}, \epsilon)  \leq 2^{(1-\delta) n}
            &\implies \exists  \ket{\psi} \in \text{span }\rho, \ket{\varphi} \in \text{span }\sigma, C \in \mathcal{C}_{2^{(1-\delta)n}} \text{ s.t. } C\ket{\psi} \approx_\epsilon \ket{\varphi}.
    \end{align}
    Note that span $\rho$ and span $\sigma$ are i.i.d. hyperplanes distributed according to the Haar measure of the Grassmannian $\text{Gr}_s\mathcal{H}$. Fix element $R$ of $\text{Gr}_s\mathcal{H}$ and choose an $\epsilon$-ball covering of $(R \setminus \{0 \}) / \mathbb{C}$ of minimal size $N = e^{\mathcal{O}\left( 2^{cn}\log(1/\epsilon)\right)}$ centred on elements $\{y_i \}_{i=1}^N$. Choose fixed unitaries $U_\rho$ and $U_\sigma$ such that $U_\rho R = \text{span } \rho$ and $U_\sigma R = \text{span } \sigma$, and independent random unitaries $V_\rho, V_\sigma \sim \mu_{\text{Haar}}$ on $ \text{span } \rho$ , $\text{span } \sigma$ respectively. This gives a randomly chosen independent $\epsilon$-covers $\{Y_i = V_\rho U_\rho y_i\}_{i=1}^N$ of $\text{span } \rho$ and $\{ Z_j = V_\sigma U_\sigma y_j \}_{j=1}^N$ of $\text{span } \sigma$. Each $Y_i, Z_j$ is distributed according to the Haar measure on $\text{span } \rho$, $\text{span } \sigma$ respectively. And so
    \begin{align}
        \mathbb{P} &\left[ \exists  \ket{\psi} \in \text{span }\rho, \ket{\varphi} \in \text{span }\sigma, C \in \mathcal{C}_{2^{(1-\delta)n}} \text{ s.t. } C\ket{\psi} \approx_\epsilon \ket{\varphi}  \right] \\
            &\leq  \mathbb{P} \left[ \exists  1 \leq i,j \leq N, C \in \mathcal{C}_{2^{(1-\delta)n}} \text{ s.t. } CY_i \approx_{3\epsilon} Z_j\right]  \\
            &\leq \sum_{i,j = 1}^N \sum_{C \in \mathcal{C}_{2^{(1-\delta)n}}} \mathbb{P}\left[ CY_i \approx_{3\epsilon} Z_j \right] \\
            &= N^2|\mathcal{C}_{2^{(1-\delta)n}}| \mathbb{P}_{\ket{\psi}, \ket{\varphi} \sim_{\text{i.i.d.}} \mu_{\text{Haar}}} \left[\ket{\psi} \approx_{3\epsilon} \ket{\varphi} \right] \\
            &= e^{\mathcal{O}\left( 2^{cn}\log(1/\epsilon)\right)} \left( n(n-1)|\mathcal{S}| \right)^{2^{(1-\delta)n}} \frac{\text{vol} (3 \epsilon \text{-ball})}{\text{vol}(\mathbb{P}\mathcal{H})} \\
            &= e^{\mathcal{O}\left( 2^{cn}\log(1/\epsilon)\right)} e^{\mathcal{O}(2^{(1-\delta)n} \log n)} e^{-\Omega \left(2^n \log(1/\epsilon) \right)} \\
            &= e^{-\Omega \left(2^n \log(1/\epsilon) \right)}.
    \end{align}
\end{proof}

Combining these two propositions with equation \eqref{eq:complexity_distance_approxcompl_bounds} gives the result.
\begin{cor} \label{cor:approx_complexity_allgates_high}
Let $\rho$, $\sigma$ be i.i.d. states on $\mathcal{H} = \left( \mathbb{C}^2 \right)^{\otimes n}$ generated by an auxiliary system $\mathcal{A}$ of dimension $s = 2^{cn}$ where $0 \leq c < 1$. For all $\delta > 0$,
    \begin{equation}
        \mathbb{P}_{\ket{\varphi}} \left[W_1^C(\rho,\sigma) \leq \epsilon^{2/3}n^{-1}\kappa \left(\frac{2^{(1-\delta)n}}{\text{{\em poly}}(n,\log \epsilon^{-1})} \right)^{1/3} \right] \leq e^{-\Omega(2^n \log ((2\epsilon)^{-1}))}.
    \end{equation}
\end{cor}
\begin{proof}
    The proof is technical and not particularly instructive, so has been placed in Appendix~\ref{appendix:approx_complexity_allgates_high}.
\end{proof}
In other words, the chance of two independent low-rank states on $n$ qubits being less than exponentially far apart in the complexity geometry quantum Wassterstein distance becomes exponentially small as $n$ tends towards infinity.

As far as we are aware, no such link between Wasserstein distances and circuit complexity has been made in the classical case due to the discreteness of classical circuits.
In particular, such comments on the computational complexity cannot ever be made from a classical point of view, as the maximum classical computational complexity between any two $n$-bit strings is at most $n$, whereas the quantum computational complexity between any two $n$-qubit states is, in general, exponential in $n$. Statements which accurately quantify the computational complexity between mixed quantum states will always require quantum tools, which again points to the importance of the $W_p^C$ distances.

The first-order Wasserstein distance of \cite{de2020quantum} has been used in the quantum setting to give lower bounds on the circuit complexity of shallow random quantum circuits \cite{li2022wasserstein}, though, just as from the classical point of view, the bound $\norm{\rho - \sigma}_{W_1^H} \leq n$ greatly restricts the effectiveness of this lower bound. The results are only significant for circuits with a number of gates linear in the number of qubits, as the maximum lower bound possible using $\norm{ \cdot }_{W_1^H}$ is $n$. 
One potential application of $W_p^C$ could be to extend this result to give lower bounds on the circuit complexity of random quantum circuits of arbitrary depth, as it does not suffer from this linear constraint.

\subsection{Operational interpretation in terms of classical-quantum sources} \label{section:operational}
In this section, we show that the $W_p^d$ distances have an operational significance for distances between classical-quantum (cq) states and classical-quantum (cq) sources. This application is particularly relevant to the quantum Wasserstein distances presented in this work because its value enhanced by both the adaptability in terms of the underlying distance $d$, and by the generality in order $p$. Indeed, let $R$ and $S$ be two cq sources, each controlled by a classical random variable $X$ on $\{1,\dots, N\}$ with probabilities $p_i$. On input $i$, let $R$ output $\ketbra{\psi_i}{\psi_i}$ and $S$ output $\ketbra{\varphi_i}{\varphi_i}$. We define $r$ to be the random variable which is the output of $R$, and $s$ the random variable which is the output of $S$. Let the output be in finite-dimensional Hilbert space $\mathcal{H}$ and equip $\mathbb{P}\mathcal{H}$ with distance $d$. These sources can be simulated by measuring the first (classical) register of the following cq states in the standard basis:
\begin{equation}
\Tilde{\rho} = \sum_{i=1}^N p_i \ketbra{i}{i}\otimes\ketbra{\psi_i}{\psi_i} \qquad \Tilde{\sigma} = \sum_{i=1}^N p_i \ketbra{i}{i}\otimes\ketbra{\varphi_i}{\varphi_i}.
\end{equation}

Letting then $\rho = \sum_{i=1}^N p_i \ketbra{\psi_i}{\psi_i}$ and $\sigma = \sum_{i=1}^N p_i \ketbra{\varphi_i}{\varphi_i}$, these $\rho$ and $\sigma$ are the expected outputs of $R$ and $S$ respectively. Using $d$ we can then talk about the distance between the outputs $r$ and $s$ as a random variable taking value $d(\ket{\psi_i},\ket{\varphi_i})$ when $X = i$. The $W_p^d$ distance provides a lower bound between the $p^{\text{th}}$ moment of the distance between the outputs. Broadly speaking, we interpret the distance as the cost of moving between $R$ and $S$.

\begin{prop} \label{prop:cq-sources}
    Let $R$ and $S$ be cq sources with expected outputs $\rho$ and $\sigma$ respectively on dimension $D$, and let $1 \leq p < \infty$. Given access to the output register and classical control register, the expected distance $d$ between the outputs $r$ and $s$ satisfies
    \begin{equation}
        \mathbb{E}_X[d(r,s)^p] \geq W_p^d(\rho,\sigma)^p
    \end{equation}
    and this bound is sharp.
\end{prop}
\begin{proof}
    For the lower bound, consider quantum transport plan $Q = \{(p_i,\ket{\psi_i},\ket{\varphi_i})\}_{i=1}^N$ between $\rho$ and $\sigma$. The $p^{\text{th}}$-order cost of this transport plan is
    \begin{equation}
        T_p^d(Q) = \sum_{i=1}^N p_i d(\ket{\psi_i},\ket{\varphi_i})^p = \mathbb{E}_X[d(r,s)^p]
    \end{equation}
    which is lower bounded by the optimal $p^{\text{th}}$-order quantum transport cost $W_p^d(\rho,\sigma)$.
    
    For sharpness, let $Q = \{(q_j,\ket{\psi_j},\ket{\varphi_j})\}_{j \in J}$ be any finite transport plan between $\rho$ and $\sigma$. The sources 
    \begin{equation} \label{eq:operational_sharpness_example}
        R = \sum_{j \in J} q_j \ketbra{j}{j} \otimes \ketbra{\psi_j}{\psi_j} \qquad
        S = \sum_{j \in J} q_j \ketbra{j}{j} \otimes \ketbra{\varphi_j}{\varphi_j}
    \end{equation}
    controlled by random variable $Y$ taking values in $J$ with probabilities $q_j$
    then have $p^{\text{th}}$ moment
    \begin{equation}
    \mathbb{E}_Y[d(r,s)^p] = \sum_{j \in J} q_j d(\ket{\psi_j},\ket{\varphi_j})^p.
    \end{equation}
    Taking the infimum over all transport plans gives sharpness.
    
\end{proof}

We also get from Proposition \ref{prop:trans_plan_inf_attained} that when $d$ is continuous, we can take the first register to have size at most $2D^2$ in the equality case.

In the case where $d$ is the complexity geometry metric $d_C$, this means that $W_p^C$ effectively quantifies the $p^{\text{th}}$ moment of the gate complexity of transforming one source into another, post-output. Indeed, for sources $R$ and $S$ as above, and $U_{\ket{\psi} \to \ket{\varphi}}^{\text{opt}} \in \mathcal{SU}(2^n)$ a unitary with $U_{\ket{\psi} \to \ket{\varphi}}^{\text{opt}}\ket{\psi} = \ket{\varphi}$, and minimal complexity among all such $U$, we know that
\begin{equation}
\mathbb{E}_X[G(U_{r \to s}^{\text{opt}})^p] \geq W_p^C(\rho,\sigma)^p.
\end{equation}
From Proposition \ref{prop:trans_plan_inf_attained}, for any $\rho,\sigma$ we may take an optimal $p^{\text{th}}$-order transport plan $Q$ from $\rho$ to $\sigma$ and let $R$ and $S$ be cq sources defined from $Q$ as above. From equation \eqref{eq:bound_dIU_by_complexities} we have 
\begin{equation}
d_C(\ket{\psi_j}, \ket{\varphi_j}) \geq \frac{\kappa G\left(U^{\text{opt}}_{\ket{\psi_j} \to \ket{\varphi_j}},\epsilon\right)^{1/3} \epsilon^{2/3}}{n^2}
\end{equation}
and therefore there exist cq sources $R$, $S$ controlled by the same random variable with expected outputs $\rho$, $\sigma$ respectively such that
\begin{equation} \label{eq:op_interpr_lower_bound}
    W_p^C(\rho,\sigma)^p \geq \frac{\kappa \mathbb{E}\left[G\left(U^{\text{opt}}_{r \to s},\epsilon\right)^{p/3}\right] \epsilon^{2/3}}{n^2}.
\end{equation}
In other words, for any $\rho$, $\sigma$ there exist cq sources $R$ and $S$ with expected outputs $\rho$, $\sigma$ such that the $p^{\text{th}}$ power of the $p^{\text{th}}$-order Wasserstein distance between their expected outputs upper bounds the $(p/3)^{\text{th}}$ moment of the $\epsilon$-approximate gate complexity of transforming between them post-output. We note that such a broad application to arbitrary moments is only possible for our $W_p^d$, as existing quantum definitions do not cover beyond $p=1$ and $p=2$.

In the infinite-order setting, the $W_\infty^d$ distance gives a lower bound for the highest possible value of $d(r,s)$. It follows from the definition of $W_\infty^d$ that 
\begin{equation}
    \max_{1 \leq i \leq N} d(\ket{\psi_i},\ket{\varphi_i}) \geq W_\infty^d(\rho,\sigma)
\end{equation} and taking the infimum of the left-hand side over $\mathcal{Q}(\rho,\sigma)$ shows that the bound is sharp. Operationally, for the $W_\infty^C$ distance, this means $W_\infty^C(\rho,\sigma)$ is a lower bound for the worst-case scenario cost of transforming $R$ into $S$, post-output. For the lower bound, we cannot guarantee an analogue to equation \eqref{eq:op_interpr_lower_bound} as the existence of an optimal transport plan requires continuity of the Wasserstein distances, which is not given (either classically or quantumly) for $p = \infty$.

\subsection{Hypercontractivity and noise} \label{section:hypercontractivity}

In the hierarchy of $W_p^d$ in order $p$, we saw in equation \eqref{eq:hierarchy_in_p} that if $p_1 < p_2 $ then $W_{p_1}^d (\rho,\sigma) \leq W_{p_2}^d (\rho,\sigma)$. This hierarchy mirrors the hierarchy in the standard $L^p$ norms, for which the notion of hypercontractivity \cite{bonamiPHD,biswal2011hypercontractivityA,keevash2021global}
has been used in various ways to quantify the noise of an operation \cite{keevash2021global, odonnellboolean}. Broadly speaking, an operator $T$ on a space of functions is hypercontractive if, for some $p_2 > p_1$, we have for all functions $f$ that $\norm{Tf}_{p_2} \leq \norm{f}_{p_1}$. Hypercontractivity and the noise of an operator are most closely linked in the hypercontractivity theorem \cite[Chapter~7]{bonamiPHD}, which demonstrates the hypercontractive properties of the standard Boolean noise operator.

We will show that this idea carries over to the quantum $W_p^d$ distances, and that the ratio $\frac{W_{p_1}^d(\rho,\sigma)}{W_{p_2}^d(N(\rho), N(\sigma))}$ can be considered as a measure of noise in the channel $N$ and be used as a tool to derive other useful properties, like concentration inequalities. This is a technique which is unique to this definition of a quantum Wasserstein distance, as it requires direct comparison of quantum Wasserstein distances of two different orders $p_1, p_2$. Such a comparison has not yet been possible as no other definition covers more than one order $p$. While different definitions of a quantum Wasserstein distance exist for both $p=1$ and $p=2$, they are qualitatively so different that they cannot yet be compared in any meaningful way to discuss hypercontractivity. Furthermore, our definition of hypercontractivity does not require us to restrict to quantum channels that have a faithful fixed point, like it is the case for the standard approach~\cite{bardet2018hypercontractivity}.

To illustrate this advantage, we will study hypercontractive properties of the replacement channel $R_{\delta, x}$, given by
\begin{equation}
    R_{\delta, x}(\rho) = (1-\delta) \rho + \delta \ketbra{x}{x}
\end{equation}
and the depolarising channel $S_\delta $ given by
\begin{equation}
    S_\delta (\rho) = (1-\delta) \rho + \delta \mathbb{I}/D
\end{equation}
where $D = \dim \mathcal{H}$. These are both examples of a more general type of depolarising map
\begin{equation}
    S_{\delta, \tau}(\rho) = (1-\delta) \rho + \delta \tau
\end{equation}
for some fixed quantum state $\tau$. But for the case of $R_{\delta, x}$, it is clear that it does not admit a faithful fixed point. We leave studying hypercontractive properties of more general channels to future work.

\begin{prop}
    Let $\rho$, $\sigma$ be two quantum states on $\mathcal{H}$ and let $1 \leq p_1 < \infty$, and suppose $W_{p_1}^d(\rho,\sigma) = M$. For $p_2 > p_1$, let $1-\delta \leq (M/\text{{\em diam}}_d (\mathbb{P}\mathcal{H}))^{p_2-p_1}$. Then the general depolarising channel $S_{\delta, \tau}$ has
    \begin{equation}
        W_{p_2}^d(S_{\delta,\tau}(\rho),S_{\delta,\tau}(\sigma)) \leq W_{p_1}^d(\rho,\sigma).
    \end{equation}
\end{prop}
\begin{proof}
    Let $Q = \{(q_j,\ket{\psi_j},\ket{\varphi_j})\}_{j \in J}$ be any ${p_1}^{\text{th}}$-order transport plan from $\rho$ to $\sigma$, and let 
    \begin{equation}
        Q' = (1-\delta)Q \cup \delta\{(\lambda_i, \ket{\omega_i}, \ket{\omega_i})\}_{i \in I}
    \end{equation}
    for $\tau = \sum_{i \in I} \lambda_i \ketbra{\omega_i}{\omega_i}$ a spectral decomposition. This is a transport plan from $S_{\delta,\tau}(\rho)$ to $S_{\delta,\tau}(\sigma)$.
    This gives
    \begin{equation}
        T_{p_2}^d(Q')^{1/p_2} = \left( (1-\delta)\sum_{j \in J} q_j d(\ket{\psi_j},\ket{\varphi_j})^{p_2}\right)^{1/{p_2}} = (1-\delta)^{1/{p_2}}T_{p_2}^d(Q)^{1/{p_2}} \leq \left(\frac{M}{\text{diam}_d(\mathbb{P}\mathcal{H})} \right)^{1-{p_1}/{p_2}}T_{p_2}^d(Q)^{1/{p_2}}.
    \end{equation}
    So it remains to show that
    \begin{equation}
        T_{p_1}^d(Q)^{1/{p_1}} \geq \left(\frac{M}{\text{diam}_d(\mathbb{P}\mathcal{H})} \right)^{1-{p_1}/{p_2}}T_{p_2}^d(Q)^{1/{p_2}}.
    \end{equation}
    Given quantities $\{A_j\}_{j \in J}$ subject to constraints 
    \begin{equation}
        \sum_{j \in J} q_j A_j^{p_1} = M^{p_1} 
    \end{equation}
    \begin{equation}
        0 \leq A_j \leq \text{diam}_d(\mathbb{P}\mathcal{H}),
    \end{equation}
    the maximum value of $\sum_{j \in J} q_j A_j^q$ is achieved when all $A_j$ are either $\text{diam}_d(\mathbb{P}\mathcal{H})$ or $0$, with values $q_1 = (M/\text{diam}_d(\mathbb{P}\mathcal{H}))^{p_1}$ and $q_2 = 1-q_1$. This gives
    \begin{equation}
        T_{p_2}^d(Q) \leq (M/\text{diam}_d(\mathbb{P}\mathcal{H}))^{p_1} \text{diam}_d(\mathbb{P}\mathcal{H})^{p_2}
    \end{equation}
    and therefore that
    \begin{equation}
        \frac{T_{p_1}^d(Q)^{1/{p_1}}}{T_{p_2}^d(Q)^{1/{p_2}}} \geq \frac{M}{\text{diam}_d(\mathbb{P}\mathcal{H})} \left(\frac{\text{diam}_d(\mathbb{P}\mathcal{H})}{M}\right)^{{p_1}/{p_2}} = \left(\frac{M}{\text{diam}_d(\mathbb{P}\mathcal{H})} \right)^{1-{p_1}/{p_2}}.
    \end{equation}
    This gives $T_{p_2}^d(Q')^{1/p_2} \leq T_{p_1}^d(Q)^{1/p_1}$, and taking the infimum over $Q$ gives the result.
\end{proof}

We note from the condition $1-\delta \leq (M/\text{diam}_d (\mathbb{P}\mathcal{H}))^{p_2-p_1}$ that this result is applicable only when $W_{p_1}^d$ is close to $\text{diam}_d (\mathbb{P}\mathcal{H})$ or when $p_2-p_1$ is small. Otherwise, we would require $\delta$ so large that the channels are no longer of practical use. Our results on random states from Sec.~\ref{section:random_states} demonstrate that this is indeed a relevant case, as for random low-rank states, $W_p^H$ and $W_p^C$ are generally high.

Particularly in the case of complexity geometry, this shows that channel noise reduces complexity. These results also demonstrate the value of considering arbitrary $p$ in our construction of the $W_p^d$ distances. Such consideration is not possible without being able to relate Wasserstein distances of different orders to one another. In general, for arbitrary $p_1, p_2$, this shows that $\frac{W_{p_1}^d(\rho,\sigma)}{W_{p_2}^d(N(\rho), N(\sigma))}$ can be well considered as a measure of noise in the channel $N$. This is applicable, in particular, in situations where we have no knowledge of the actual structure of the noise of a channel, and so cannot determine it efficiently, but still want to gauge how noisy it is on a qualitative level.

In situations where we do not have $M / \text{diam}_d(\mathbb{P}\mathcal{H})$ large, we can still use the notion of hypercontractivity to explore the nature of quantum channels. For example, hypercontractivity can tell us how a channel $N$ affects the concentration of measure of one state with respect to the span of another. This is formalised in the following proposition. Here, we assume that the span of $N(\sigma)$ is very small. We can then use hypercontractivity to discuss how much the mass of $N(\rho)$ concentrates within a distance $K$ of the span of $N(\sigma)$.
\begin{prop} \label{prop:hypercontractivity_conc_of_measure}
    Let $\rho, \sigma$ be states on $\mathcal{H}$ and suppose $\mathbb{P}\mathcal{H}$ is equipped with metric $d$. Let $N: \mathcal{D}(\mathcal{H})  \to  \mathcal{D}(\mathcal{H})$ be a quantum channel, and suppose that we have hypercontractivity in $N$ with respect to orders $p_1 < p_2$: that is, $W_{p_1}^d(\rho,\sigma) \geq W_{p_2}^d(N(\rho),N(\sigma))$. Then for $K \in \left(W_{p_1}^d(\rho,\sigma), \text{{\em diam}}_d(\mathbb{P}\mathcal{H})\right)$, there is a coupling of $N(\rho)$ and $N(\sigma)$, $Q = \{(q_j,\ket{\psi_j},\ket{\varphi_j})\}_{j \in J}$, s.t. 
    \begin{equation} \label{eq:hypercontractivity_conc_of_measure}
\sum\limits_{j:d(\ket{\psi_j},\ket{\varphi_j})>K}q_j\leq \left(\frac{W_{p_1}^d(\rho,\sigma)}{K}\right)^{p_2}
\end{equation}
\end{prop}
\begin{proof}
    Let $M = W^d_{p_1}(\rho,\sigma) \geq W_{p_2}^d(N(\rho),N(\sigma))$. It follows that in the limit of optimal $p_2^{\text{th}}$-order transport plans $Q = \{ (q_j, \ket{\psi_j},\ket{\varphi_j}) \}_{j \in J}$ between $N(\rho)$ and $N(\sigma)$, we have 
    \begin{equation}
    \sum_{j \in J} q_j d\left(\ket{\psi_j},\ket{\varphi_j}\right)^{p_2} \leq M^{p_2}.
\end{equation}
Let $L = \sum_{j \in J : d(\ket{\psi_j},\ket{\varphi_j}) > K} q_j$. This also has $L = \mathbb{P}\left( d\left(\ket{\psi_X},\ket{\varphi_X}\right) \geq K \right)$, where $X$ is a random variable taking values in $J$ according to law $\mathbb{P}(X = j) = q_j$. Note therefore that by Markov's inequality
\begin{equation}
    L \leq \frac{\sum_{j \in J} q_j d\left(\ket{\psi_j},\ket{\varphi_j}\right)^{p_2}}{K^{p_2}} \leq \left(\frac{M}{K}\right)^{p_2}.
\end{equation}
\end{proof}
Thus, we see that a hypercontractive inequality implies that most of the transport plan is concentrated on states that are close, and this concentration becomes more pronounced the lower the value of $p_1$ (as the initial distance is monotonically increasing in $p_1$) and the higher the value of $p_2$.

\subsection{Quantum Wasserstein generative adversarial networks} \label{section:GANs}

One of the most promising near-term applications of quantum devices is in quantum machine learning. Among these methods, one framework which has attracted attention in recent years is that of generative adversarial networks, or GANs.

GANs~\cite{goodfellow2014generativeadversarialnetworks} form a framework in machine learning which seeks to generate new data which are indistinguishable from the training data. For example, they can be used to generate images which look like authentic photographs when trained on datasets of photographs~\cite{brock2019largescalegantraining}, or provide image-to-text translation for automatic image description for visually impaired users~\cite{liuAIA}. The main feature of GANs is learning via a zero-sum game between a {\em generator}, which is trained to generate data, and a {\em discriminator}, which is trained to distinguish between the generated data and the training data. The standard GAN takes its definition of `far' to be `large Jensen-Shannon divergence'.

Wasserstein GANs~\cite{arjovsky2017wgan} were proposed in 2017 as an alternative to standard GANs, which redefine `far' to mean `large first-order Wasserstein distance'. These provide numerous improvements over standard GANs, including the reduction of mode collapse, and allowing reference distributions which are concentrated to small dimension in a space of large dimension.

Quantum GANs~\cite{Dallaire_Demers_2018} are a natural quantum equivalent of classical GANs, where a quantum state generator $G$ parameterised by $\theta$ is trained to a target distribution $\rho_{\text{tar}}$. The standard architecture for such a generator is a shallow quantum circuit, whose gates are functions of $\theta$.

With standard distinguishing functions, such as the trace distance or fidelity, these suffer greatly from the problem of barren plateaus~\cite{McClean_2018}
- areas where, as the number $n$ of qudits grows, the gradient of the objective function decays to $0$. As in the classical setting, GANs which use the theory behind the first-order quantum Wassertein distance of De Palma et al \cite{de2020quantum} have been proposed
\cite{Kiani2022, De_Palma_2024}
to eliminate this problem.

The form of such a quantum Wasserstein GAN (qWGAN) is as follows. The generated state is parameterised by $\theta$, and is stored as a set $\{(p_1, U_1(\theta), \dots, p_r U_r(\theta)\}$ where $p(\theta) = \{p_1, \dots, p_r\}$ is a probability distribution also parameterised by $\theta$. This represents the state
\begin{equation}
    G(\theta) = \sum_{i=1}^r p_i U_i(\theta) \ketbra{0}{0} U_i(\theta)^{\dagger}.
\end{equation}

The discriminator is a Lipschitz observable $O$, which aims to discriminate between $G(\theta)$ and the target state $\rho_{\text{tar}}$. One round of theoretical qWGAN optimisation takes the following form:

\begin{enumerate}
    \item Replace $O$ by $O_{\max} = \text{argmax}_{\norm{O}_L \leq 1} \text{Tr}[O(G(\theta) - \rho_{\text{tar}})]$,
    \item Calculate the gradients in $\theta$ of $\text{Tr}[O_{\max}(G(\theta) - \rho_{\text{tar}})]$,
    \item Update $p(\theta)$ and each $U_i(\theta)$ according to these gradients.
\end{enumerate}

For practical reasons, linked to the difficulty of optimising over the set $\{ O : \norm{O}_L \leq 1\}$, the algorithm in practice uses an approximation to $\norm{O}_L$ rather than $\norm{O}_L$ itself. This has, at least for some toy examples such as the GHZ state, been able to eliminate the issue of barren plateaus seen in standard quantum GANs.
This highlights the need for flexibility in the distinguishability metric used in quantum GANs.

Another problem among similar lines, seen frequently in generative networks `in the wild', is the misalignment of human ideas of indistinguishability and computer ideas of indistinguishability. Many structures which generate data are only able to mimic the training data at a surface level, making them seem correct at first glance but with unsettling artefacts of computer generation or with fundamental incorrectness or impossibility of the data generated. This again highlights the need for broad flexibility in the distinguishability metric used in quantum Wasserstein GANs.

Such broad flexibility can theoretically be easily achieved by replacing the Lipschitz norm $\norm{\cdot}_{L}$ associated to the $\norm{\cdot}_{W_1^H}$ norm, with the more general $L_d(O)$ constant defined in equation~\ref{eq:d_lipschitz}, for an underlying distance $d$ satisfying the required continuity conditions for this to exist. This would, in theory, allow us once again to avoid barren plateaus, but also to prioritise qualitatively different ideas of discrimination via the choice of the metric $d$.

This would give rounds of qWGAN optimisation of the following form:
\begin{enumerate}
    \item Replace $O$ by $O_{\max} = \text{argmax}_{L_d(O) \leq 1} \text{Tr}[O(G(\theta) - \rho_{\text{tar}})]$,
    \item Calculate the gradients in $\theta$ of $\text{Tr}[O_{\max}(G(\theta) - \rho_{\text{tar}})]$,
    \item Update $p(\theta)$ and each $U_i(\theta)$ according to these gradients.
\end{enumerate}

This model opens a theoretical avenue to qWGANs with flexible notions of distinguishability, though we note that implementation of this lies beyond the scope of this work due to its highly variable nature. In particular, the big challenge is accurately calculating and maximising $L_d(O)$, and it is likely that this will require bespoke approximations to $L_d$ for each choice of $d$. 

Any such approximation would need to be careful to preserve the nature of the distance chosen, in order to avoid the issue of misalignment of notions of distinguishability. For example, the above mentioned works~\cite{Kiani2022, De_Palma_2024} 
discuss qWGANs which supposedly give convergence in $\norm{\cdot}_{W_1^H}$, but the approximation used for the Lipschitz constant of $O$ is a weighted sum of the coefficients of the Pauli decomposition of $O$, which directly measures local distinguishability as discussed in~\cite{Kiani2022}. We have shown in Section~\ref{section:random_states} that the $\norm{\cdot}_{W_1^H}$ does not reflect local distinguishability in general.

More concretely, they only optimise over operators $O$ of the form
\begin{equation}
    O = \omega_{\mathbb{I}} \mathbb{I} + \sum_{i=1}^n \sum_{P \in \{ X,Y,Z\}} \omega^{(i)}_P P^{(i)} + \sum_{1 \leq i < j \leq n} \sum_{P, Q \in \{ X,Y,Z\}} \omega^{(i,j)}_{P,Q} P^{(i)} \otimes Q^{(j)}
\end{equation}
where $P^{(i)}$ is the Pauli operator $P$ acting on qubit $i$, and tuples $(i,j)$ are unordered. To ensure these have $\norm{O}_L \leq 1$, they force the coefficients $\omega_{P,Q}$ to satisfy
\begin{equation} \label{eq:WGAN_W_H_Lipschitz_approx_coeff_conditions}
    2 \max_{1 \leq i \leq n} \sum_{P \in \{ X,Y,Z\}}\left|\omega_P^{(i)}\right| + \sum_{j \neq i} \sum_{P, Q \in \{X,Y,Z\}} \left|\omega_{P,Q}^{(i,j)}\right| \leq 1.
\end{equation}

Let us call the set of such operators on $n$ qubits $\tilde{L}_n$.

\begin{prop} \label{prop:WGAN_W_H_is_limited}
    Let $\ket{\psi}, \ket{\varphi}$ be independent pure states on $n$ qubits distributed according to the Haar measure. 
    \begin{equation} \label{eq:WGAN_W_H_is_limited}
        \mathbb{P}\left[ \max_{O \in \tilde{L}_n} \text{{\em Tr}}[O (\ketbra{\psi}{\psi}-\ketbra{\varphi}{\varphi})] \geq 
        n 2^{4-\frac{n}{4}} \right] \leq   6(3n^2-1)2^{-\frac{n}{4}}.
    \end{equation}
\end{prop}
\begin{proof}
    Let $R$ be a Pauli string on one or two qubits. Note that $\norm{R}_{\infty} = 1$, and so $\text{Tr}[R(\ketbra{\psi}{\psi}-\ketbra{\varphi}{\varphi})] \leq \norm{\rho-\sigma}_1 \leq W_1^1(\rho,\sigma)$ where $\rho$ and $\sigma$ are the marginals of $\ketbra{\psi}{\psi}$ and $\ketbra{\varphi}{\varphi}$ respectively on two qubits containing the support of $R$. So from Proposition \ref{prop:random_mixed_largerank} applied with qubit ratio $c = \frac{n}{2}-1$, underlying dimension $4$, and parameter $\beta = 2^{n/4}$, we have that
    \begin{equation}
            \mathbb{P}\left[\text{Tr}[R(\ketbra{\psi}{\psi}-\ketbra{\varphi}{\varphi})] \geq 2^{4-\frac{n}{4}} \right] \leq 4\cdot 2^{-\frac{n}{2}}.
    \end{equation}
    Letting $\omega$ be the sum of absolute values of the Pauli coefficients of $O$ (except $\omega_I$). From equation \eqref{eq:WGAN_W_H_Lipschitz_approx_coeff_conditions} we have $\omega \leq n$. It follows via a union bound that
    \begin{equation}
        \mathbb{P}\left[\text{Tr}[O(\ketbra{\psi}{\psi}-\ketbra{\varphi}{\varphi})] \geq n 2^{4-\frac{n}{4}} \right] \leq 6(3n^2-1)2^{-\frac{n}{4}}.
    \end{equation}
\end{proof}
This shows that the probability of an exponentially small deviation from zero decays exponentially in $n$. In terms of the qWGAN algorithm of~\cite{Kiani2022}, this means that even with a convergence threshold which decays exponentially in $n$, the ability of the discriminator to distinguish between randomly chosen states decays exponentially too. This has significant practical implications. Namely, even for exponentially low convergence thresholds, the algorithm  will on average fail at the first iteration, as a randomly chosen initial $\rho_{\text{init}} = \ketbra{\psi}{\psi}$ will be indistinguishable from the target $\rho_{\text{tar}} = \ketbra{\varphi}{\varphi}$ by any $O$ allowed by the algorithm. However, our results on random states show that random pure states $\ketbra{\psi}{\psi}$ and $\ketbra{\varphi}{\varphi}$ are on average maximally far apart in $\norm{\cdot}_{W_1^H}$, so the algorithm will believe that the states have converged at the initial step despite them being maximally far apart.

When considering qWGANs when defined with respect to $W_1^d$, avoiding such misalignment in approximation of $L_d(O)$ would be an important consideration for any practical implementation of such algorithms.

\section{Conclusion}
In this work, we have introduced a novel definition of the quantum Wasserstein distance by combining the coupling method and a metric on the set of pure states. This novel definition successfully captures the essence of the classical Wasserstein distance. A significant aspect of our approach is its inherent adaptability, effortlessly incorporating established metrics on quantum states such as trace distance and naturally extending pure-state metrics, for instance, Nielsen's complexity metric, to cater to mixed states.

Furthermore, we established various properties of this new definition. While we acknowledge that these properties are not exhaustive, they nevertheless offer advantages compared to other recent generalisations, such as the fact that one definition can cover several examples in the literature in a unified way. A significant challenge remains: establishing the triangle inequality in its generality. Our findings, though, hint at the possibility that under suitable conditions the dual approach yields a proof of the triangle inequality.

Our exploration of specific cases reveals the $W_1^1$ distance as a potentially powerful tool for bounding the trace distance between quantum states, drawing parallels with classical approaches for determining total variation distance and mixing time bounds. Additionally, our work enriches the understanding of the complexity geometry of quantum states, offering a new lens through which to view and quantify the complexity of transformations within quantum ensembles.

Our examination of the behaviour of the Wasserstein distance under random quantum states unveils various phase transitions. These results, derived from entropic inequalities and continuity bounds, debunk some existing speculations (e.g. that $\norm{\cdot}_{W_1^H}$ captures local distinguishability) while affirming others (the complexity of small subsystems of random quantum states is low). 

In conclusion, our research represents a significant stride forward in the pursuit of understanding optimal transport in quantum mechanics. The novel quantum Wasserstein distance that we have proposed holds great promise, not only as an analytical tool but also as a medium for further exploration and discovery in quantum computation and information. While our work has paved the way, many challenges and open problems remain, especially when it comes to practical implementation of the applications laid out here.
\section{Conflicts of interest}
The authors declare no conflicts of interest.
\section{Acknowledgements}
DSF and EB would like to thank Alexander M\"uller-Hermes, Raul Garc\'ia-Patr\'on, Philippe Faist, Fanch Coudreuse, and Guillaume Aubrun for interesting discussions. This project received support from the PEPR integrated project EPiQ ANR-22-PETQ-0007 part of Plan France 2030.

\bibliographystyle{alpha}
\bibliography{neatref}

\appendix

\section{Comparison to other proposed definitions of a quantum Wasserstein distance} \label{appendix:asym_bad}

Until now, finding a coherent way to define a quantum optimal transport cost, regardless of the underlying metric on the Hilbert space, has proved elusive. Some definitions (such as in \cite{DePalma2021}, \cite{Carlen_2014} and \cite{Golse2015OnTM}) work well for specific orders $p$, and specific underlying metrics. A further definition was proposed recently in \cite{chakrabarti2019quantum} using standard quantum couplings and using projection onto the asymmetric subspace as a cost function, with the motivation of mimicking the definition of total variation distance as the Wasserstein distance corresponding to the trivial metric. 

Originally, \cite{chakrabarti2019quantum} defined a second-order Wasserstein semi-distance $W$ on states on $\mathcal{H} = \mathbb{C}^d$ as follows. Letting $\mathbb{F} \ket{a} \otimes \ket{b} = \ket{b}\otimes \ket{a}$ define the flip operator on $\mathcal{H} \otimes \mathcal{H}$, we can then define the symmetric projector $P_{\text{sym}}(d)$ as the projector onto the $1$-eigenspace of $\mathbb{F}$, and $P_{\text{asym}}(d)$ the projector onto its $(-1)$-eigenspace. Note $P_{\text{sym}}(d) + P_{\text{asym}}(d) = \mathbb{I}_{\mathbb{C}^d \otimes \mathbb{C}^d}$.

They then define the quantum optimal transport cost

\begin{equation}\label{equ:w_through_transport}
    T(\rho,\sigma) = \min_{\tau_{AB} \in \mathcal{D}(\mathbb{C}^d \otimes \mathbb{C}^d), \tau_A = \rho, \tau_B = \sigma} \text{Tr}[\tau_{AB}P_{\text{asym}}(d)]
\end{equation}
and $W(\rho,\sigma) = \sqrt{T(\rho,\sigma)}$. This gives a 2-Wasserstein semi-distance which is equivalent to the trace distance.
This definition was further studied in \cite{Friedland_2022} and \cite{cole2022quantum}, and refined in \cite{muellerhermes2022monotonicity}.

Specifically, \cite{muellerhermes2022monotonicity} showed that this original definition does not satisfy a data-processing inequality, so it does not mirror the original trace distance in this regard. The definition in Eq.~\eqref{equ:w_through_transport} was changed to give the data processing inequality by considering a complete version:

\begin{equation}
    T_s(\rho,\sigma) = T(\rho \otimes \mathbb{I}/2, \sigma \otimes \mathbb{I}/2)
\end{equation}
and $W_s(\rho,\sigma) = \sqrt{T_s(\rho,\sigma)}$. Note that in general, the complete version of a quantity requires an optimisation over auxiliary systems of arbitrary size, however \cite{muellerhermes2022monotonicity} shows that the maximally mixed state of one qubit suffices.

There are two main differences between this path to obtaining a generalisation of the Wasserstein distance and ours. First, in this path, the cost of transforming one state into another is given by the expectation value of an observable on the bipartite system. In our case, we depart from a metric on the set of point masses/pure states. Although the metric version is in direct analogy with the classical version, optimising over the expectation value of an observable certainly has a more operational interpretation. Second, in the definition of Eq.~\eqref{equ:w_through_transport} we see that we indeed optimise over all couplings, including separable ones. In our definition, we only optimise over separable ones, which gives the definition in Eq.~\eqref{equ:w_through_transport} a more quantum flavor. Furthermore, computing Eq.~\eqref{equ:w_through_transport} corresponds to an SDP, so it can be computed efficiently in the dimension, whereas it is unclear if we can efficiently compute our version in general.
However, one of the main motivations for our work was to obtain a unified way of obtaining various versions of $W_p$ present in the literature, including the definition of~\cite{de2020quantum}. However, it appears that approaches like that of Eq.~\eqref{equ:w_through_transport} cannot mimic the behaviour of the Wasserstein distance of De Palma et al. Consider, for example, the potential extension to an $n$-qudit space given by
\begin{equation}
T(\rho,\sigma) = \min \left\{ \text{Tr}\left[\tau_{AB}\left(\sum_{i=1}^n P_{i, \text{asym}}(d)\right)\right] : \tau_{AB} \in \mathcal{D}\left((\mathbb{C}^d)^{\otimes n} \otimes (\mathbb{C}^d)^{\otimes n}\right): \tau_A = \rho, \tau_B = \sigma \right\}
\end{equation}
where $P_{i, \text{asym}}(d)$ is the projection onto the asymmetric subspace of the $i^{\text{th}}$ qudits (and is the identity on the other qudits). In stabilising, it makes little difference whether we stabilise each qudit individually with its own copy of $\mathbb{I}/2$ or share one, so we stabilise this definition to
\begin{align}
    T_s(\rho,\sigma) = 
    \min \Bigg\{ \text{Tr} & \left[\tau_{AB} \left(\sum_{i=1}^n P_{i, \text{asym}}(d \otimes 2)\right) \right] : \\ & \tau_{AB} \in \mathcal{D}\left((\mathbb{C}^d \otimes \mathbb{C}^2)^{\otimes n} \otimes (\mathbb{C}^d \otimes \mathbb{C}^2)^{\otimes n}\right), \tau_A = \rho \otimes (\mathbb{I}/2)^{\otimes n}, \tau_B = \sigma \otimes (\mathbb{I}/2)^{\otimes n} \Bigg\}
\end{align}
and $W_s(\rho,\sigma) = \sqrt{T_s(\rho,\sigma)}$.

We'll show, however, that any such type of generalisation leads to some states $\rho$ with $T_s(\rho,\rho) > 0$. Take, in this example, $n=d=2$ and $\rho$ to be a Bell state $\ketbra{\psi^{+}}{\psi^{+}}$. As $\rho$ is pure, the set of possible couplings $\tau$ is very limited. In the non-stabilised definition, $\tau$ can only be $\ketbra{\psi^{+}}{\psi^{+}}\otimes \ketbra{\psi^{+}}{\psi^{+}}$. In the stabilised definition, $\tau$ must have form $\ketbra{\psi^{+}}{\psi^{+}}\otimes \ketbra{\psi^{+}}{\psi^{+}} \otimes \omega$ for some coupling $\omega$ of $\mathbb{I}^{\otimes 2}/4 $ and $\mathbb{I}^{\otimes 2}/4$.

However, when looking at applying $P_{\text{asym}}$ to the individual qubits, we note that as the Bell state has qubit marginals $\mathbb{I}/2$, we will have nonzero transport cost from $\rho$ to itself. Indeed in the non-stabilised case,
\begin{align}
    T(\rho,\sigma) 
        &= \sum_{i=1}^2 \text{Tr} \left[\ketbra{\psi^{+}}{\psi^{+}}\otimes \ketbra{\psi^{+}}{\psi^{+}} \left(P_{i,\text{asym}}(2) \right) \right] \\
        &= 2 \text{Tr} \left[\left(\frac{\mathbb{I}}{2} \otimes \frac{\mathbb{I}}{2} \right) P_{\text{asym}}(2) \right] \\
        &= \frac{1}{2} \text{Tr} [P_{\text{asym}}(2)] = \frac{1}{2}.
\end{align}

In the stabilised case, we still have for any $\tau$, that
\begin{align}
    &\text{Tr}\left[\tau_{AB}  \left(\sum_{i=1}^2 P_{i, \text{asym}}(2 \otimes 2)\right) \right]\\
        &= \sum_{i=1}^2 \text{Tr} \left[\ketbra{\psi^{+}}{\psi^{+}}\otimes \ketbra{\psi^{+}}{\psi^{+}} \otimes \omega \left(P_{i,\text{asym}}(2 \otimes 2) \right) \right] \\
        &= \sum_{i=1}^2 \text{Tr} \left[\left(\frac{\mathbb{I}}{2} \otimes \frac{\mathbb{I}}{2} \otimes \omega_i \right) P_{\text{asym}}(2 \otimes 2) \right] \label{eq:w2doesntwork_wi} \\
        &=\sum_{i=1}^2 \text{Tr} \left[\left(\frac{\mathbb{I}}{2} \otimes \frac{\mathbb{I}}{2} \otimes \omega_i \right) (P_{\text{asym}}(2) \otimes P_{\text{sym}}(2) + P_{\text{sym}}(2 )\otimes P_{\text{asym}}(2 )) \right] \\
        &= \sum_{i=1}^2 \frac{1}{4}\text{Tr} \left[\omega_i P_{\text{sym}}(2) \right] + \frac{3}{4} \text{Tr} \left[\omega_i P_{\text{asym}}(2) \right] \\
        &\geq \sum_{i=1}^2 \frac{1}{4} \big(\text{Tr} \left[\omega_i P_{\text{sym}}(2) \right] + \text{Tr} \left[\omega_i P_{\text{asym}}(2) \right]\big) = \sum_{i=1}^2 \frac{1}{4} \text{Tr}[\omega_i] = \frac{1}{2}
\end{align}
where in line \eqref{eq:w2doesntwork_wi}, $\omega_i$ is a coupling of $\mathbb{I}/2$ with itself, on the $i^{\text{th}}$ stabilising qubits.

We see from the details of this example that, no matter whether or not we stabilise using many qubits or a single qubit, and no matter the size of the space, any attempt to put a $P_{\text{asym}}$ projector on a subspace of $\mathcal{H}$ will result in highly entangled pure states, those which have marginals close to maximally mixed on subspaces where $P_{\text{asym}}$ is placed, having nonzero self-distance. This is largely because the coupling of a pure state with itself is forced to be the product of the state with itself, and because the identity on $\mathbb{C}^d \otimes \mathbb{C}^d$ has a nonzero asymmetric and a nonzero symmetric component. Thus, it appears that a version of the Wasserstein distance that is not based on a single observable, like ours, is more suited to obtain generalisations of quantities like the one of De Palma et al.

\section{Development of the proposed definition}
\subsection{Definition of a transport plan}
During the development of this proposal, many different ways of defining a transport plan were considered. They all had different versions of the conditions
\begin{equation}
    \sum_{j \in J} q_j \ketbra{\psi_j}{\psi_j} = \rho, \sum_{j \in J} q_j \ketbra{\varphi_j}{\varphi_j} = \sigma, q_j > 0.
\end{equation}
As in the definition of the $W_1$ norm in \cite{de2020quantum}, a first proposal was simply the condition
\begin{equation}
    \sum_{j \in J} q_j ( \ketbra{\psi_j}{\psi_j}-\ketbra{\varphi_j}{\varphi_j}) = \rho - \sigma, q_j > 0.
\end{equation}
It quickly becomes clear, that using a transport plan of this form containing telescoping sums, could lead to the degeneracy of $W_p^d$ when $p > 1$. The main stumbling block here centred on the unboundedness of $\sum_{j \in J} q_j$.

Two proposals for restricting $\sum_j q_j$ were then considered, the first being 
\begin{equation}\sum_{j \in J} q_j = 1 \label{eq:proposal_qj1} \end{equation} and the second \begin{equation}\sum_{j \in J} q_j = \frac{1}{2}\norm{\rho-\sigma}_1. \label{eq:proposal_qj1norm}\end{equation} Both of these were discounted when considering the transport plans that were allowed under these regimes, and their impact on the norm $W_\infty^H$.

For transport plans of the form \eqref{eq:proposal_qj1}, consider the $W_\infty^H$ distance on $\mathcal{H} = \left(\mathbb{C}^3\right)^{\otimes 2}$, and the states
$\rho = \frac{1}{2}\ketbra{00}{00} + \frac{1}{2}\ketbra{22}{22}$, $\sigma = \frac{1}{2}\ketbra{11}{11} + \frac{1}{2}\ketbra{22}{22}$. Classically, the infinite-order transport distance between the measures $\mu = \frac{1}{2}\delta_{00} + \frac{1}{2}\delta{22}$ and $\nu = \frac{1}{2}\delta_{11} + \frac{1}{2}\delta{22}$ on the Hamming cube $\{0,1,2\}^2$ is $2$. With the definition from \eqref{eq:proposal_qj1} however, we could define a transport plan
\begin{equation}
    Q = \left\{ \left( \frac{1}{2},\ket{00},\ket{01}\right),\left(\frac{1}{2},\ket{01},\ket{11}\right)\right\}
\end{equation}
with a maximum $d(\ket{\psi_j},\ket{\varphi_j})$ of $1$. Though as in \eqref{eq:quantum_shortcut_plan} we are perfectly happy that the $W_\infty^H$ distance might not match Ornstein's $\bar{d}$-distance for classical measures, it certainly should not differ because of a transport plan of a classical nature such as this one.

The picture for condition \eqref{eq:proposal_qj1norm} is very similar. Consider $W_\infty^H$ on $\mathcal{H} = \left(\mathbb{C}^2\right)^{\otimes 2}$, with states $\rho = \frac{1}{2}\ketbra{00}{00} + \frac{1}{2}\ketbra{01}{01}$, $\sigma = \frac{1}{2}\ketbra{01}{01} + \frac{1}{2}\ketbra{11}{11}$. Classically, these have an infinite-order transport distance of $1$ on the Hamming cube $\{0,1\}^2$. However, because $\frac{1}{2}\norm{\rho-\sigma}_1$ is small in this case, the only permitted transport plan is
\begin{equation}
    Q = \left\{ \left(\frac{1}{2},\ket{00},\ket{11}\right)\right\}
\end{equation}
with a maximum $d(\ket{\psi_j},\ket{\varphi_j})$ of $2$. The condition \eqref{eq:proposal_qj1norm} does not permit quantum versions of the standard classical transport plans, and so this definition of a transport plan seems strictly worse than the one we propose in this work.

\subsection{Entangled couplings and transport plans}\label{appendix:entangled_transport}
The main reason we need to restrict to separable couplings is the fact that it is not obvious at first how to attribute a distance to an entangled state starting from a distance on $\mathbb{P}\mathcal{H}$ and then attribute a transportation cost to transport plans that include entangled pure states and satisfy the marginal constraints.

One possibility is as follows: given a state $\ket{\psi}$ in $\mathbb{P}(\mathcal{H}\otimes \mathcal{H})$ and $SD(\ket{\psi})$ the set of Schmidt decompositions of $\ket{\psi}$, we define $d(\ket{\psi})$ as:
\begin{align}
d(\ket{\psi})=\inf\limits_{\sum_i\sqrt{p_i}\ket{\phi_1}\otimes\ket{\phi_2}\in SD(\ket{\psi})}\sum_ip_id(\ket{\phi_1},\ket{\phi_2}).
\end{align}
It is then easy to see that we recover the original distance when considering product states. However, it is also easy to see that, at least for the case $p=1$, enlarging the set of possible states to include entangled states offers no advantage. Indeed, given any Schmidt decomposition of an entangled $\ket{\psi}=\sum_i\sqrt{p_i}\ket{\phi_1}\otimes\ket{\phi_2}$, adding instead $\{(p_i,\ket{\phi_1},\ket{\phi_2})\}_i$ to the transport plan will give the same cost and still satisfy the marginal constraints. Thus, at least for this possible generalisation of the distance to entangled states, there is no advantage in considering entangled couplings.

\subsection{Transport plans defined from absolutely continuous measures} \label{appendix:integrable_plans}

The definition of a transport plan as given in equation~\eqref{eq:defn_transport_plan} takes a countable set $\{q_j\}_{j \in J}$ of positive weights such that $\sum_{j \in J} q_j = 1$, and applies these weights to pairs $(\ket{\psi_j}, \ket{\varphi_j})$ in $\mathbb{P}\mathcal{H}_1 \times \mathbb{P}\mathcal{H}_2$. This is equivalent to a probability measure $\mu$ on $\mathbb{P}\mathcal{H}_1 \times \mathbb{P}\mathcal{H}_2$ of the form
\begin{equation} \label{eq:appendix_integrable_countable_plan_measure}
    q = \sum_{j \in J} q_j \delta_{(\ket{\psi_j}, \ket{\varphi_j})}
\end{equation}
such that
\begin{equation} \label{eq:appendix_integrable_measure_conditions}
    \int_{\mathbb{P}\mathcal{H}_1} \ketbra{\psi}{\psi} \text{d}(\pi_{1 \#}q)(\ket{\psi}) = \rho, \qquad \int_{\mathbb{P}\mathcal{H}_2} \ketbra{\varphi}{\varphi} \text{d}(\pi_{2 \#}q)(\ket{\varphi}) = \sigma
\end{equation}
for $\pi_1$ and $\pi_2$ the projections onto the first and second coordinates of $\mathbb{P}\mathcal{H}_1 \times \mathbb{P}\mathcal{H}_2$ respectively.

We could, instead, relax this definition to consider transport plans defined by arbitrary probability measures $q$ on $\mathbb{P}\mathcal{H}_1 \times \mathbb{P}\mathcal{H}_2$ satisfying equation \eqref{eq:appendix_integrable_measure_conditions}. For a distance $d$ on $\mathbb{P}\mathcal{H}$, taking $\mathcal{H} = \mathcal{H}_1 = \mathcal{H}_2$, this would give transport cost
\begin{equation}
    T_p^d(q) = \int_{\mathbb{P}\mathcal{H} \times \mathbb{P}\mathcal{H}} d(\ket{\psi},\ket{\varphi})^p \text{d}q(\ket{\psi},\ket{\varphi}).
\end{equation}
In order to replicate the proof of continuity from Appendix~\ref{appendix:continuity}, we would take $q$ to be absolutely continuous with respect to the Haar measure on $\mathbb{P}\mathcal{H} \times \mathbb{P}\mathcal{H}$, and then all of the proofs in this work can be replicated replacing sums with integrals and $q_j$ with the density $q(\ket{\psi},\ket{\varphi})$ with respect to the Haar measure.

The following discussion shows that the infimum over transport plans defined from absolutely continuous measures is at least the infimum over countable transport plans, and so we lose nothing in transport cost by only considering countable plans.
Again, we restrict to uniformly continuous $d$ for continuity.

Let $q$ be a measure on $\mathbb{P}\mathcal{H} \times \mathbb{P}\mathcal{H}$ which has density $q(\ket{\psi}, \ket{\varphi})$ with respect to the Haar measure. Take a dense subset $\{ (\ket{\psi_j} ,\ket{\varphi_j}) \}_{j=1}^\infty $ of $\mathbb{P}\mathcal{H} \times \mathbb{P}\mathcal{H}$. Let $B_j(\epsilon)$ be the open ball of radius $\epsilon$ around $(\ket{\psi_j}, \ket{\varphi_j})$, with respect to the distance $\norm{\ketbra{\psi}{\psi} - \ketbra{\psi'}{\psi'}}_2 + \norm{\ketbra{\varphi}{\varphi} - \ketbra{\varphi'}{\varphi'}}_2$. Let $\{ u_j \}_{j=1}^\infty$ be a partition of unity over $B_j(\epsilon)$. We can then define
\begin{equation}
    q_j = \int_{B_j(\epsilon)} q(\ket{\psi},\ket{\varphi})u_j(\ket{\psi},\ket{\varphi}) \text{d}\mu_{\text{Haar}}(\ket{\psi},\ket{\varphi}).
\end{equation}
and let 
\begin{equation}
    Q = \{ (q_j , \ket{\psi_j},\ket{\varphi_j} ) \}_{j=1}^\infty
\end{equation}
which is a countable transport plan between some states $\rho', \sigma'$. We will show that $\rho'$ and $\sigma'$ are very close to $\rho$ and $\sigma$, and that $Q$ has a very similar transport cost to $q$.

We have that 
\begin{align}
    &\norm{q_j \ketbra{\psi_j}{\psi_j} - \int_{B_j(\epsilon)} \ketbra{\psi}{\psi}q(\ket{\psi},\ket{\varphi}) u_j(\ket{\psi},\ket{\varphi}) \text{d}\mu_{\text{Haar}}(\ket{\psi},\ket{\varphi})}_2 \\
        &= \norm{\int_{B_j(\epsilon)} (\ketbra{\psi_j}{\psi_j} - \ketbra{\psi}{\psi})  q(\ket{\psi},\ket{\varphi}) u_j(\ket{\psi},\ket{\varphi}) \text{d}\mu_{\text{Haar}}(\ket{\psi},\ket{\varphi})}_2 \\
        &\leq \int_{B_j(\epsilon)} \norm{\ketbra{\psi_j}{\psi_j} - \ketbra{\psi}{\psi}}_2  q(\ket{\psi},\ket{\varphi}) u_j(\ket{\psi},\ket{\varphi}) \text{d}\mu_{\text{Haar}}(\ket{\psi},\ket{\varphi}) \\
        &\leq \int_{B_j(\epsilon)} \epsilon  q(\ket{\psi},\ket{\varphi}) u_j(\ket{\psi},\ket{\varphi}) \text{d}\mu_{\text{Haar}}(\ket{\psi},\ket{\varphi}) \\
        &= \epsilon q_j
\end{align}
and therefore that $\norm{\rho' -\rho}_{\infty} \leq \norm{\rho' -\rho}_{2} \leq \epsilon$. The same holds for $\sigma$.

For the similarity of cost, 
\begin{align}
    \left| T_d^p(Q) - T_d^p(q) \right| 
        &= \left| \sum_{j=1}^{\infty}q_j d(\ket{\psi_j},\ket{\varphi_j})^p - \int_{\mathbb{P}\mathcal{H} \times \mathbb{P}\mathcal{H} } d(\ket{\psi},\ket{\varphi})^p q(\ket{\psi},\ket{\varphi}) \text{d}\mu_{\text{Haar}}(\ket{\psi},\ket{\varphi})\right| \\
        &= \left| \sum_{j=1}^{\infty} \int_{B_j(\epsilon)} q(\ket{\psi},\ket{\varphi})u_j(\ket{\psi},\ket{\varphi}) \left(d(\ket{\psi_j},\ket{\varphi_j})^p - d(\ket{\psi},\ket{\varphi})^p \right) \text{d}\mu_{\text{Haar}}(\ket{\psi},\ket{\varphi})\right| \label{eq:appendix_integrable_similar_cost_swap_limits}\\
        &\leq  \sum_{j=1}^{\infty} \int_{B_j(\epsilon)} q(\ket{\psi},\ket{\varphi})u_j(\ket{\psi},\ket{\varphi}) \left|d(\ket{\psi_j},\ket{\varphi_j})^p - d(\ket{\psi},\ket{\varphi})^p \right| \text{d}\mu_{\text{Haar}}(\ket{\psi},\ket{\varphi}) \\
        &\leq \sum_{j=1}^\infty q_j \sup_{\norm{\ketbra{\psi}{\psi} - \ketbra{\psi'}{\psi'}}_2 + \norm{\ketbra{\varphi}{\varphi} - \ketbra{\varphi'}{\varphi'}}_2 \leq \epsilon} \left|d(\ket{\psi'},\ket{\varphi'})^p - d(\ket{\psi},\ket{\varphi})^p \right| \\
        &=  \sup_{\norm{\ketbra{\psi}{\psi} - \ketbra{\psi'}{\psi'}}_2 + \norm{\ketbra{\varphi}{\varphi} - \ketbra{\varphi'}{\varphi'}}_2 \leq \epsilon} \left|d(\ket{\psi'},\ket{\varphi'})^p - d(\ket{\psi},\ket{\varphi})^p \right|.
\end{align}
swapping limits in line \ref{eq:appendix_integrable_similar_cost_swap_limits} as the terms in each part are all positive, and combining into one as the sums and integrals are all bounded. This final line tends to $0$ with $\epsilon$ due to the uniform continuity of $d$.

From the proof of continuity in Appendix~\ref{appendix:continuity}, we can find a countable transport plan $Q'$ from $\rho$ to $\sigma$ with cost at most the cost of $Q$ in the limit $\epsilon \to 0$. Therefore in the limit $\epsilon \to \infty$, $T_p^d(Q') \leq T_p^d(q)$ and we are done.

\section{Towards the triangle inequality} \label{appendix:tri_ineq}

As with many other attempts to generalise the classical Wasserstein distances to the quantum setting, the triangle inequality eludes us. The central barrier to the triangle inequality is the quantum marginal problem: even in the case where we do not require our couplings to be separable, there is no coherent way to build a coupling $\tau_{13}$ from coupling $\tau_{12}$ of $\rho_1$ and $\rho_2$ and coupling $\tau_{23}$ of $\rho_2$ and $\rho_3$.

All known proofs of the triangle inequality for the classical $\mathcal{W}_p$ distances rely on this idea: given a coupling $\gamma_{12}$ of measures $\mu_1$ and $\mu_2$, and coupling $\gamma_{23}$ of $\mu_2$ and $\mu_3$, we can find measure $\gamma_{123}$ with $(1,2)$-marginal $\gamma_{12}$ and $(2,3)$-marginal $\gamma_{23}$. We can then take the $(1,3)$-marginal $\gamma_{13}$ of this overarching measure, and show that \cite{clement2008}
\begin{equation}
    T_p(\gamma_{13})^{1/p} \leq T_p(\gamma_{12})^{1/p} + T_p(\gamma_{23})^{1/p}
\end{equation}
via Minkowski's inequality. Without an equivalent $\tau_{123}$ in the quantum setting from which to form coupling $\tau_{13}$, there is as of yet no clear path to a triangle inequality.

For the case $p=1$, we can focus on the case where $\norm{\rho-\sigma}_{DW_1^d} = W_1^d(\rho,\sigma)$, as $\norm{\cdot}_{DW_1^d}$ does satisfy the triangle inequality. As mentioned in Sec.~\ref{section:dual_picture}, this remains a key open problem of our work.

\section{Auxiliary proofs} \label{appendix:proofs}

\subsection{Proof of Proposition \ref{prop:Wp_cts} on the continuity of $W_p^d$} \label{appendix:continuity}

\noindent \textbf{Proposition \ref{prop:Wp_cts}.} \textit{Suppose $d$ is continuous on $\mathbb{P}\mathcal{H}$ and let $1 \leq p < \infty$. Then $W_p^d$ is uniformly continuous.}

\vspace{2mm}
\noindent \textit{Proof.}
Take $\epsilon > 0$. Let $\rho, \sigma, \rho', \sigma'$ be states on $\mathbb{P}\mathcal{H}$ with $\norm{\rho-\rho'}_{\infty} \leq \delta$ and $\norm{\sigma-\sigma'}_{\infty} \leq \delta$, for some $\delta$ to be chosen later. For any transport plan from $\rho$ to $\sigma$, we will form a transport plan from $\rho'$ to $\sigma'$ with a similar cost. Note that as $\mathcal{H}$ is finite-dimensional, $\mathbb{P}\mathcal{H}$ is compact and so $d$ is uniformly continuous.

    Let $c \gg 1$ (also to be chosen later, but note scale $1 \gg c\delta$) and let $S_{\rho'}$, $S_{\sigma'}$ respectively be the span of the eigenvectors of $\rho', \sigma'$ whose eigenvalues are at least $c\delta$. Let $\Pi_{\rho'}$, $\Pi_{\sigma'}$ be the projectors onto $S_{\rho'}, S_{\sigma'}$ respectively.

    Let 
    \begin{equation}
        Q = \{(q_j, \ket{\psi_j}, \ket{\varphi_j})\}_{j \in J}
    \end{equation}
    be any $p^{\text{th}}$-order transport plan from $\rho$ to $\sigma$. Define then $\ket{\psi_j} = \ket{\psi_j^{\rho'}} + \ket{\psi_j^{\perp}}$ where $\ket{\psi_j^{\rho'}} = \Pi_{\rho'} \ket{\psi_j}$ and $\ket{\psi_j^{\perp}}$ is perpendicular to $S_{\rho'}$. Define $\ket{\varphi_j^{\sigma'}}$ and $\ket{\varphi_j^{\perp}}$ analogously.

We can then begin to define a transport plan from $\rho'$ to $\sigma'$. Let
\begin{equation}
        \tilde{\rho} = \Pi_{\rho'} \rho \Pi_{\rho'} = \sum_{j \in J} q_j \ket{\psi_j^{\rho'}}\bra{\psi_j^{\rho'}}
    \end{equation}
    and define $\tilde{\sigma}$ analogously. Note that $\norm{\Pi_{\rho'}(\rho-\rho')\Pi_{\rho'}}_{\infty} \leq \norm{\rho-\rho'}_{\infty} < \delta$. Therefore
    \begin{equation}
        \tilde{\rho} \leq \Pi_{\rho'} \rho' \Pi_{\rho'} + \delta \mathbb{I}_{S_{\rho'}} \leq \left(1 + \frac{1}{c} \right)\Pi_{\rho'} \rho' \Pi_{\rho'} \leq \left(1+\frac{1}{c}\right) \rho'
    \end{equation}
    and so $\frac{c}{c+1}\tilde{\rho} \leq \rho'$.

The same holds for $\sigma$, so $\frac{c}{c+1}\tilde{\sigma} \leq \sigma'$.

We may then begin to build our new transport plan starting with the partial transport plan
\begin{equation}
    Q_1 = \left\{\left(q_j', \frac{\ket{\psi_j^{\rho'}}}{\sqrt{\braket{\psi_j^{\rho'} |\psi_j^{\rho'} } }},\frac{\ket{\varphi_j^{\sigma'}}}{\sqrt{\braket{\varphi_j^{\sigma'} |\varphi_j^{\sigma'} } }}\right) \right\}_{j \in J}
\end{equation}
where $q_j' = \frac{c}{c+1} q_j \min\{ \braket{\psi_j^{\rho'} |\psi_j^{\rho'} }  , \braket{\varphi_j^{\sigma'} |\varphi_j^{\sigma'} }  \}$. Note then that
\begin{equation}
    \rho' \geq \frac{c}{c+1}\tilde{\rho} \geq \sum_{j \in J} q_j' \frac{\ket{\psi_j^{\rho'}}\bra{\psi_j^{\rho'}}}{\braket{\psi_j^{\rho'} |\psi_j^{\rho'} } } \qquad \text{and} \qquad \sigma' \geq \frac{c}{c+1} \tilde{\sigma} \geq \sum_{j \in J} q_j' \frac{\ket{\varphi_j^{\sigma'}}\bra{\varphi_j^{\sigma'}}}{\braket{\varphi_j^{\sigma'} |\varphi_j^{\sigma'} } }.
\end{equation}
We can then transport the positive semidefinite operator $\rho' - \sum_{j \in J} q_j' \frac{\ket{\psi_j^{\rho'}}\bra{\psi_j^{\rho'}}}{\braket{\psi_j^{\rho'} |\psi_j^{\rho'} } }$ onto $\sigma' - \sum_{j \in J} q_j' \frac{\ket{\varphi_j^{\sigma'}}\bra{\varphi_j^{\sigma'}}}{\braket{\varphi_j^{\sigma'} |\varphi_j^{\sigma'} } }$ via any partial transport plan. Let these transport plan elements form set $Q_2$. It follows that $Q_1 \cup Q_2$ is a transport plan from $\rho'$ to $\sigma'$. 

We can now attempt bounding the cost of this transport plan. We will show that $Q_2$ has a very small cost and that $Q_1$ has a cost very close to $T_p^d(Q)$. Starting with the $Q_2$ part, we know that
\begin{equation}
    T_p^d(Q_2) \leq \text{Tr}\left[ \rho' - \sum_{j \in J} q_j' \frac{\ket{\psi_j^{\rho'}}\bra{\psi_j^{\rho'}}}{\braket{\psi_j^{\rho'} |\psi_j^{\rho'} } }\right] \text{diam}_d(\mathbb{P}\mathcal{H})^d. \label{eq:cts_bound_Q_2}
\end{equation}

We will show that this trace part is very small. Indeed
\begin{align}
    \text{Tr}\left[ \rho' - \sum_{j \in J} q_j' \frac{\ket{\psi_j^{\rho'}}\bra{\psi_j^{\rho'}}}{\braket{\psi_j^{\rho'} |\psi_j^{\rho'} } }\right]
    &=\norm{\rho' - \sum_{j \in J} q_j' \frac{\ket{\psi_j^{\rho'}}\bra{\psi_j^{\rho'}}}{\braket{\psi_j^{\rho'} |\psi_j^{\rho'} } }}_{1}\\
        &\leq \norm{\rho-\rho'}_1 + \norm{\rho - \sum_{j \in J} q_j' \frac{\ket{\psi_j^{\rho'}}\bra{\psi_j^{\rho'}}}{\braket{\psi_j^{\rho'} |\psi_j^{\rho'} } }}_1 \\
        &\leq \delta \text{dim}\mathcal{H} + \norm{\rho - \sum_{j \in J} q_j' \frac{\ket{\psi_j^{\rho'}}\bra{\psi_j^{\rho'}}}{\braket{\psi_j^{\rho'} |\psi_j^{\rho'} } }}_1 \label{eq:cts_split_Q_2}
\end{align}

Choose $M > 0$ (again to be determined later), and select $L > 0$ such that whenever $\text{Tr}\ketbra{\psi_j^{\rho'}}{\psi_j^{\rho'}} \geq 1-L$ and $\text{Tr}\ketbra{\varphi_j^{\sigma'}}{\varphi_j^{\sigma'}} \geq 1-L$ we have $d(\ket{\psi_j^{\rho'}},\ket{\varphi_j^{\sigma'}})^p < d(\ket{\psi_j},\ket{\varphi_j})^p + M$. This is possible by uniform continuity of $d$. We then split the set of $j$ into those for which the projections onto $S_{\rho'}$ and $S_{\sigma'}$ are large, and those for which they are not.

Let \begin{equation}
    \mathcal{K} = \left\{ j \in J : \text{Tr} \ketbra{\psi_j^{\rho'}}{\psi_j^{\rho'}}, \text{Tr}\ketbra{\varphi_j^{\sigma'}}{\varphi_j^{\sigma'}} > 1-L \right\}.
\end{equation}

and we can split up these $1$-norm terms into $j \in \mathcal{K}$ and $j \notin \mathcal{K}$. For $j \in \mathcal{K}$, we have 
\begin{align}
\norm{q_j \ketbra{\psi_j}{\psi_j} - q_j' \frac{\ketbra{\psi_j^{\rho'}}{\psi_j^{\rho'}}}{\text{Tr}\ketbra{\psi_j^{\rho'}}{\psi_j^{\rho'}}}}_1 
    &\leq   q_j\norm{\ketbra{\psi_j}{\psi_j}-\frac{\ketbra{\psi_j^{\rho'}}{\psi_j^{\rho'}}}{\text{Tr}\ketbra{\psi_j^{\rho'}}{\psi_j^{\rho'}}}}_1 + |q_j - q_j'|\\
    &\leq 2q_j\sqrt{L} + \left(L + \frac{1}{c+1}\right)q_j = \left(2\sqrt{L} + L + \frac{1}{c+1} \right) q_j
\end{align}
 and for $j \notin \mathcal{K}$, we have $\norm{q_j \ketbra{\psi_j}{\psi_j} - q_j' \frac{\ketbra{\psi_j^{\rho'}}{\psi_j^{\rho'}}}{\text{Tr}\ketbra{\psi_j^{\rho'}}{\psi_j^{\rho'}}}}_1 \leq 2q_j$.
In order to bound $\sum_{j \notin \mathcal{K}} q_j$, let $\{\ket{\alpha}\}_{\alpha}$ be some orthonormal basis for the orthogonal complement of $S_{\rho'}$. We have
\begin{equation}
    \norm{\rho-\rho'}_{\infty} \text{dim}\mathcal{H} \geq \sum_{\alpha} \braket{ \alpha | \rho-\rho' | \alpha} \geq \text{Tr} \left[\sum_j q_j (\ketbra{\psi_j^{\perp}}{\psi_j^{\perp}} + \ketbra{\varphi_j^{\perp}}{\varphi_j^{\perp}})\right] - c\delta \text{dim} \mathcal{H}.
\end{equation}
The same applies to $\sigma$. Hence
\begin{align}
    \sum_{j \notin \mathcal{K}} q_j &\leq \frac{1}{L} \sum_{j \notin \mathcal{K}} q_j \text{Tr} [\ketbra{\psi_j^{\perp}}{\psi_j^{\perp}} + \ketbra{\varphi_j^{\perp}}{\varphi_j^{\perp}}] \\
    &\leq \frac{1}{L} \text{Tr} \left[\sum_j q_j (\ketbra{\psi_j^{\perp}}{\psi_j^{\perp}} + \ketbra{\varphi_j^{\perp}}{\varphi_j^{\perp}})\right] \\
    &\leq \frac{1}{L} \text{dim}\mathcal{H} \left(\norm{\rho-\rho'}_\infty + \norm{\sigma - \sigma'}_{\infty} + 2c\delta \right) \\
    &\leq \frac{2\delta(c+1)}{L} \text{dim}\mathcal{H}.
\end{align}
This gives us
\begin{equation}
 \norm{\rho - \sum_{j \in J} q_j' \frac{\ket{\psi_j^{\rho'}}\bra{\psi_j^{\rho'}}}{\braket{\psi_j^{\rho'} |\psi_j^{\rho'} } }}_1 \leq 2\sqrt{L} + L  + \frac{1}{c+1} + \frac{4\delta(c+1)}{L} \text{dim}\mathcal{H}.
\end{equation}
By symmmetry this holds for $\sigma$, and so substituting into equation \eqref{eq:cts_split_Q_2} and then \eqref{eq:cts_bound_Q_2} gives
\begin{equation}
    T_p^d(Q_2) \leq \left(\delta \text{dim}{\mathcal{H}} + 2\sqrt{L} + L + \frac{1}{c} + \frac{4\delta(c+1)}{L}\text{dim}\mathcal{H} \right) \text{diam}_d(\mathbb{P}\mathcal{H})^p.
\end{equation}
This bounds $T_p^d(Q_2)$ above.

For the bounding of $Q_1$, we have
\begin{align}
    T_p^d(Q_1) 
        &= \sum_{j \in J} q_j' d(\ket{\psi_j^{\rho'}},\ket{\varphi_j^{\sigma'}})^p \\
        &= \sum_{j \in \mathcal{K}} q_j' d(\ket{\psi_j^{\rho'}},\ket{\varphi_j^{\sigma'}})^p  + \sum_{j \notin \mathcal{K}} q_j' d(\ket{\psi_j^{\rho'}},\ket{\varphi_j^{\sigma'}})^p \\
        &\leq \sum_{j \in \mathcal{K}} q_j d(\ket{\psi_j},\ket{\varphi_j})^p + M + \frac{2\delta}{L}\text{dim}\mathcal{H}\cdot \text{diam}_d(\mathbb{P}\mathcal{H})^p \\
        &\leq T_p^d(Q) + M + \frac{2\delta}{L}\text{dim}\mathcal{H}\cdot \text{diam}_d(\mathbb{P}\mathcal{H})^p.
\end{align}
It follows then that
\begin{equation}
    T_p^d(Q_1 \cup Q_2) \leq T_p^d(Q) + M + \left(\frac{(4c+6)\delta}{L}\text{dim}\mathcal{H} + \delta\text{dim}\mathcal{H} + 2\sqrt{L} + L + \frac{1}{c+1} \right)\text{diam}_d(\mathbb{P}\mathcal{H})^p.
\end{equation}

And so choosing $M < \frac{1}{4}\epsilon$, $L$ sufficiently small such that the continuity condition for $M$ is satisfied and such that $2\sqrt{L} + L \leq \frac{\epsilon}{4\text{diam}_d(\mathbb{P}\mathcal{H})^p}$, $c = \frac{4\text{diam}_d(\mathbb{P}\mathcal{H})^p}{\epsilon}$, and then choosing $\delta < \frac{\epsilon}{4(1 + (4c+6)/L)\text{dim}\mathcal{H}\text{diam}_d(\mathbb{P}\mathcal{H})^p}$ gives $T_p^d(Q_1\cup Q_2) \leq T_p^d(Q) + \epsilon$. Thus we have a transport plan for $\rho'$ to $\sigma'$ with cost at most $\epsilon$ more than the transport cost of plan $Q$ between $\rho$ and $\sigma$.

Taking the infimum over all such transport plans $Q$, we see that whenever $\norm{\rho-\rho'}_{\infty} < \delta$ and $\norm{\sigma-\sigma'}_{\infty} < \delta$, we have
\begin{equation}
    \inf_{Q' \in \mathcal{Q}(\rho',\sigma')} T_p^d(Q') \leq \epsilon + \inf_{Q \in \mathcal{Q}(\rho,\sigma)} T_p^d(Q).
\end{equation}
It follows that $W_p^d$ is uniformly continuous. \hfill $\square$

\subsection{Proof of Corollary \ref{cor:approx_complexity_allgates_high} on the approximate gate complexity of random pure states} \label{appendix:approx_complexity_allgates_high}

\noindent \textbf{Corollary \ref{cor:approx_complexity_allgates_high}.}
\textit{Let $\rho$, $\sigma$ be i.i.d. states on $\mathcal{H} = \left( \mathbb{C}^2 \right)^{\otimes n}$ generated by an auxiliary system $\mathcal{A}$ of dimension $s = 2^{cn}$ where $0 \leq c < 1$. For all $\delta > 0$,}
    \begin{equation}
        \mathbb{P}_{\ket{\varphi}} \left[W_1^C(\rho,\sigma) \leq \epsilon^{2/3}n^{-1}\kappa \left(\frac{2^{(1-\delta)n}}{\text{{\em poly}}(n,\log \epsilon^{-1})} \right)^{1/3} \right] \leq e^{-\Omega(2^n \log ((2\epsilon)^{-1}))}.
    \end{equation}
\vspace{2mm}
\noindent \textit{Proof.}
    For any fixed $\ket{\psi}, \ket{\varphi}$, using equation \eqref{eq:complexity_distance_approxcompl_bounds} and Lemma \ref{lemma:link_complexity_allgates_universalset}, we have
    \begin{align}
        d_C(\ket{\psi},\ket{\varphi}) 
            &\geq \min_{U \in \mathcal{SU}(2^n), U\ket{\psi}=\ket{\varphi}}  \kappa G(U,\epsilon)^{1/3}\epsilon^{2/3} n^{-2}   \\
            &\geq \min_{U \in \mathcal{SU}(2^n), U\ket{\psi}=\ket{\varphi}} \kappa \left(\frac{G_{\mathcal{S}}(U,2\epsilon)}{\text{poly}(\log(G(U,\epsilon))+\log(\epsilon^{-1}))}\right)^{1/3}\epsilon^{2/3} n^{-2} \\
            &\geq \min_{U \in \mathcal{SU}(2^n), U\ket{\psi}=\ket{\varphi}} \kappa \left(\frac{G_{\mathcal{S}}(U,2\epsilon)}{\text{poly}(n,\log(\epsilon^{-1}))}\right)^{1/3}\epsilon^{2/3} n^{-2}.
    \end{align}
    Given that $W_1^C(\rho,\sigma) \geq \min\left\{ d_C(\ket{\psi},\ket{\varphi}) : \ket{\psi} \in \text{ span } \rho, \ket{\varphi} \in \text{ span } \sigma \right\}$,
    it follows from Lemma \ref{lemma:approx_complexity_universalset_low_mixed} that
    \begin{align}
        &\mathbb{P}_{\rho, \sigma} \left[ W_1^C(\rho,\sigma) 
            \leq \epsilon^{2/3}n^{-1}\kappa \left(\frac{2^{(1-\delta)n}}{\text{poly}(n,\log \epsilon^{-1})} \right)^{1/3} \right] \\
            &\leq \mathbb{P}_{\rho,\sigma} \left[\min_{U : U\ket{\psi}=\ket{\varphi}, \ket{\psi} \in \text{ span } \rho, \ket{\varphi} \in \text{ span } \sigma } \kappa \left(\frac{G_{\mathcal{S}}(U,2\epsilon)}{\text{poly}(n,\log(\epsilon^{-1}))}\right)^{1/3}\epsilon^{2/3} n^{-2} \leq \epsilon^{2/3}n^{-2}\kappa\left(\frac{2^{(1-\delta)n}}{\text{poly}(n,\log \epsilon^{-1})} \right)^{1/3} \right] \\
            &= \mathbb{P}_{\rho,\sigma} \left[G_{\mathcal{S}}\left(U_{\rho \to \sigma}^{\text{ opt}},2\epsilon\right) \leq 2^{(1-\delta)n} \right]\\
            &\leq e^{-\Omega(2^n \log (1/2\epsilon)}
    \end{align}
    as claimed.  \hfill $\square$

\end{document}